%% file: main.tex
\documentclass[12pt]{article}
\usepackage[utf8]{inputenc}
\usepackage{graphicx}
\usepackage{amsmath}
\usepackage{amsthm}
\usepackage{natbib}
\usepackage{url}
\newtheorem{claim}{Claim}
\newtheorem{lemma}{Lemma}
\newtheorem{assum}{Assumption}

\begin{document}

\title{New numerical methods for calculating statistical equilibria of two-dimensional turbulent flows, strictly based on the Miller-Robert-Sommeria theory}
\author{K Ryono\footnote{Current affiliation: Research Institute for Applied Mechanics, Kyushu University.}, K Ishioka\\Graduate School of Science, Kyoto University, \\
Kitashirakawa-Oiwake-cho, Sakyo-ku, Kyoto 606-8502, Japan}
\date{Accepted: 3 October 2022}

\begin{abstract}
New numerical methods are proposed for the mixing entropy maximization problem in the context of Miller-Robert-Sommeria's statistical mechanics theory of two-dimensional turbulence, particularly in the case of spherical geometry. Two of the methods are for the canonical problem; the other is for the microcanonical problem. The methods are based on the original MRS theory and thus take into account all Casimir invariants. Compared to the methods proposed in previous studies, our new methods make it easier to detect multiple statistical equilibria and to search for solutions with broken zonal symmetry. The methods are applied to a zonally symmetric initial vorticity distribution which is barotropically unstable. Two statistical equilibria are obtained, one of which has a wave-like structure with zonal wavenumber 1, and the other has a wave-like structure with zonal wavenumber 2. While the former is the maximum point of the mixing entropy, the wavenumber 2 structure of the latter is nearly the same as the structure that appears in the end state of the time integration of the vorticity equation. The new methods allow for efficient computation of statistical equilibria for initial vorticity distributions consisting of many levels of vorticity patches without losing information about all the conserved quantities. This means that the statistical equilibria can be obtained from an arbitrary initial vorticity distribution, which allows for the application of statistical mechanics to interpret a wide variety of flow patterns appearing in geophysical fluids.
\end{abstract}
\maketitle

\textit{This is the Accepted Manuscript version of an article accepted for publication in Fluid Dynamics Research. IOP Publishing Ltd is not responsible for
 any errors or omissions in this version of the manuscript or any version derived from it. The Version of Record is available online at \url{https://doi.org
/10.1088/1873-7005/ac9713}.}

\input{Ryono_Ishioka_main_R2}

\clearpage

\setcounter{section}{0}
\setcounter{figure}{0}
\setcounter{table}{0}
\setcounter{equation}{0}

\renewcommand{\thesection}{S\arabic{section}}
\renewcommand{\thefigure}{S\arabic{figure}}
\renewcommand{\thetable}{S\arabic{table}}
\renewcommand{\theequation}{S\arabic{equation}}

\begin{center}
    \large \bf Supplementary Material for ``New numerical methods for calculating statistical equilibria of two-dimensional turbulent flows, strictly based on the Miller-Robert-Sommeria theory''
\end{center}
\vspace*{5mm}
\input{Ryono_Ishioka_supplement_R2}

\end{document}

%% file: Ryono_Ishioka_main_R2.tex
\section{Introduction}
The effort to introduce equilibrium statistical mechanics into the dynamics of two-dimensional fluids began with the pioneering work of \cite{onsager1949statistical}. The negative temperature state found by Onsager was investigated in more detail by \cite{joyce1973negative} and \cite{montgomery1974statistical}. In the negative temperature state, vortices of the same sign cluster together. Based on research since the 1980s, this can be regarded as corresponding to the formation of coherent vortices in decaying two-dimensional turbulence discovered by \cite{mcwilliams1984emergence}. In studies prior to the 1970s, models of statistical mechanics were based on the dynamics of point vortices; however, \cite{miller1990statistical} and \cite{robert1991statistical} independently presented theories applicable to continuous vorticity fields. Following their theory (Miller-Robert-Sommeria theory or MRS theory), a number of studies have attempted to interpret characteristic flow structures arising in (nearly) two-dimensional systems as statistical mechanical equilibria. These efforts have been particularly evident in the study of flows on the Earth and other planets, including Jupiter's Great Red Spot and jets \citep[e.g.,][]{miller1992statistical,michel1994statistical,turkington2001statistical,bouchet2002emergence,chavanis2005statistical}, the polar vortex \citep{prieto2001analytical,yasuda2017new}, oceanic jets and ocean rings \citep{venaille2011oceanic}, and the collapse of typhoon eyewalls \citep{schubert1999polygonal}. 

Notably, however, methods for computing the statistical equilibrium of individual systems are not yet well established. The statistical equilibrium state is defined as the redistribution of the initial vorticity where the mixing entropy is maximized, under the constraints of conservation laws. To compute the equilibrium, it is necessary to solve an optimization problem under nonlinear and nonconvex constraints, a mathematically difficult problem. In particular, there may be more than one point of local maximum of the mixing entropy, and it is necessary to compare the values of their mixing entropy. To date, a variety of numerical methods have been proposed, including the relaxation method by \cite{robert1992relaxation}, the iterative method by \cite{turkington1996statistical}, and the continuation method by \cite{thess1994inertial}. These are ingenious methods for finding a critical point of mixing entropy. However, the relaxation method tends to be computationally expensive and might not be practical unless the initial vorticity distribution is relatively simple (e.g., two or three levels of vorticity patches). The iterative method is fast, but the computed equilibria are determined in a discontinuous and uncontrolled way \citep{bouchet2012statistical}, which makes it difficult to produce a global picture of the statistical equilibria. Furthermore, it is not always easy to set appropriate initial points from which the iteration converges. The continuation method was devised for a relatively special system, and it is doubtful that it can be applied to general initial vorticity distributions in a realistic computational time. In addition, it is not easy to detect multiple critical points with this method.

For geophysical applications, it would seem natural to consider the MRS theory for a two-dimensional sphere. Moreover, in pure fluid dynamics, two-dimensional turbulence on a sphere is an interesting subject, and it has been confirmed in numerical experiments that unique structures appear on a rotating sphere, such as jets at high latitudes \citep[e.g.,][]{yoden1993numerical}. Whether such structures can be interpreted from the viewpoint of statistical mechanics is an interesting question. In flow domains with some symmetry, such as a two-dimensional sphere, there can be both critical points with zonal symmetry and critical points with broken zonal symmetry. If the initial flow has a zonal symmetry, it is difficult to reach the non-zonal critical points computationally using the above-mentioned methods. For example, \cite{prieto2001analytical} applied MRS theory to a zonal initial flow on a sphere that mimics a polar night jet and computed a statistical equilibrium with zonal symmetry using the iterative method. However, computing multiple non-zonal equilibria with this method would appear to be quite difficult. Even if we were to extend the relaxation method of \cite{robert1992relaxation} on a sphere, the zonal symmetry of the initial vorticity field cannot be broken and a non-zonal critical point cannot be reached. This is due to the nature of the relaxation method, which is to mix the initial vorticity field in a time-evolving manner so that the mixing entropy increases. \cite{ishioka1998} devised a numerical method that can be applied to general vorticity distributions which was successful in computing a symmetry breaking solution. However, due to its large computational cost, the application of this method is limited to the computation of very coarse spatial resolutions.

Largely due to these technical limitations, when MRS theory is applied to geophysical flows, simplified versions of the original mixing entropy maximization problem have often been used to compute the statistical equilibria. For example, \cite{chavanis1996classification} proposed a method in which a linear relationship was assumed between vorticity and the stream function in equilibrium by considering a strong mixing limit, which relates to energy-enstrophy theory \citep{kraichnan1975statistical,kraichnan1980two}.
\cite{herbert2012statistical} and \cite{herbert2013additional} also used this simplified version of the mixing entropy maximization problem in their studies of the statistical equilibria of flows on a spherical domain. As another simplification, a variational problem of the stream function was considered to determine the shape of jets appearing in statistical equilibria \citep[e.g.,][]{bouchet2002emergence,venaille2011oceanic}. (For more details on the relationships between these simplified methods and the original mixing entropy maximization problem, see \cite{bouchet2008simpler}.) It should be noted, however, that the results obtained by these simplified methods do not, in general, satisfy the conservation of the vorticity distribution determined by the initial vorticity field. In other words, most of the infinitely many Casimir invariants of inviscid incompressible two-dimensional fluids are not conserved. Although these simplifications have had some success in qualitatively evaluating the equilibrium state of the flow, the initial vorticity distribution needs to be taken into account if we want a more quantitative treatment of individual flows.

In the present manuscript, we propose new numerical methods to compute the statistical equilibria of flows on a sphere. These methods are able to find the local maxima of the mixing entropy maximization problem in its original form, as presented by Miller, Robert and Sommeria, and thus produce statistical equilibria that conserve the vorticity distribution determined by the initial vorticity field. In other words, the resulting solutions conserve all Casimir invariants from a microscopic viewpoint. Furthermore, even if the given initial vorticity distribution is a general one, the statistical equilibria can be computed by using the new methods in a realistic computational time. The new methods also allow for finding multiple local maxima of the mixing entropy and make it possible to compute equilibria with broken zonal symmetry. The remainder of the paper is organized as follows. In section \ref{MRS_theory}, we give a short description of the MRS theory and define two types of problems: canonical problems and microcanonical problems. The proposed numerical methods are described in section \ref{methods}. In section \ref{example}, we show an example of the application of the proposed new methods to a specific initial flow and compare the equilibrium state of the flow resulting from the time integration with the computed statistical equilibrium. A discussion of the results and their implications is provided in section \ref{discussion}, and conclusions are presented in section \ref{conclusion}. Definitions of the spherical harmonics and the associated Legendre functions, matters related to numerical integration, and notes on implementation and computational complexity are given in the two appendices. The mathematical arguments underlying the numerical methods in the study are detailed in the Supplement to the main manuscript.
 
\section{MRS theory and definitions of MRS statistical equilibria}\label{MRS_theory} 
In this section, we briefly explain the statistical mechanics theory (Miller-Robert-Sommeria theory or MRS theory) of \cite{miller1990statistical} and \cite{robert1991statistical}. We also define two problems: the microcanonical problem and the canonical problem. In solving each of these problems, we obtain the statistical equilibria.

\subsection{MRS theory and mixing entropy}
Although the original theories of \cite{miller1990statistical} and \cite{robert1991statistical} consider two-dimensional flows on a plane, the MRS theory can be extended to two-dimensional flows on a sphere, as in \cite{prieto2001analytical}, \cite{herbert2012statistical} and \cite{herbert2013additional}. In the following, we focus solely on the MRS theory as it applies to a sphere.

Let \(S\) be a sphere. The radius of \(S\) may be assumed to be \(1\) using a scaling argument. Let \(\lambda\,(0\leq \lambda\leq 2\pi)\) and \(\theta\,(-\frac{\pi}{2}<\theta<\frac{\pi}{2})\) be the longitude and latitude, respectively, of points on \(S\), excluding the north and south poles. We will use \(x=(\lambda,\mu)\,(\mu=\sin \theta)\) as the coordinates on the sphere \(S\). Then the area element on \(S\) can be denoted by 
\({\rm d} S = {\rm d} \lambda {\rm d} \mu\).
The incompressible inviscid two-dimensional flow on \(S\) obeys the Euler equation
\begin{align}
\frac{\partial q}{\partial t}+\left(\frac{\partial \psi}{\partial \lambda}\frac{\partial q}{\partial \mu}-\frac{\partial q}{\partial \lambda}\frac{\partial \psi}{\partial \mu}\right)=0. \label{R=1eq}
\end{align}
Here, \(\psi\) is the stream function, which is related to the zonal velocity \(u\) and meridional velocity \(v\) of the flow according to
\begin{align*}
u=-\sqrt{1-\mu^2}\frac{\partial \psi}{\partial \mu}, \qquad v=\frac{1}{\sqrt{1-\mu^2}}\frac{\partial \psi}{\partial \lambda},
\end{align*}
and \(q=\Delta \psi\) (\(\Delta\) is the Laplace-Beltrami operator of \(S\)) is the (absolute) vorticity. That is, the vorticity \(q\) of each parcel of the fluid is advected by the flow, and the value of \(q\) is materially conserved. 
This system has the conserved quantities described below: 

Energy
\begin{align}
E= \frac{1}{4\pi}\int_S\frac{1}{2} |\nabla \psi |^2 {\rm d} S = -\frac{1}{4\pi}\int_S \frac{1}{2}\psi q {\rm d} S.
\label{energy}
\end{align}

Angular momentum
\begin{align}
M_1 = \frac{1}{4\pi}\int_S q \mu {\rm d} S, \label{am1}\\
M_2= \frac{1}{4\pi}\int_S q\sqrt{1-\mu^2}\cos \lambda {\rm d} S, \label{am2}\\
M_3 = \frac{1}{4\pi}\int_S q\sqrt{1-\mu^2}\sin\lambda {\rm d} S. \label{am3}
\end{align}
Note that there are three components of angular momentum, corresponding to the three rotational degrees of freedom of a two-dimensional sphere. In addition, since the vorticity is advected by the incompressible flow, all of the Casimir invariants
\begin{align*}
C_f = \frac{1}{4\pi}\int_S f(q) {\rm d} S
\end{align*}
are conserved. Here, \(f\) is an arbitrary function. Thus, the system has an infinite number of conserved quantities.

We consider initial vorticity fields consisting of \(K\) vorticity patches (the value of the vorticity of each patch is \(Q_1,\cdots, Q_K\)), and let \(S_k\) be the area of the \(k\)-th vorticity patch. In general, as time goes on, the flow becomes turbulent and these vorticity patches are well mixed, keeping the vorticity value of each patch unchanged. If sufficient mixing occurs, the patches are repeatedly elongated and folded, and become filamentary. In the MRS theory, such fine-scaled, filamentary vorticity fields are considered microscopic states, and we assume that a macroscopic state corresponds to each microscopic state by taking the local average of the vorticity field. To formulate the above idea, we introduce \(r_k(x)\,(k=1,\cdots,K)\) as the probability of the \(k\)-th patch covering point \(x\). In other words, \(r_k(x)\) is the probability of the \(k\)-th patch observed at \(x\in S\). 
Due to the incompressible nature of the flow, it is necessary that
\begin{align}
\int_S r_k(x) {\rm d} S =S_k \qquad(k=1,\cdots,K).\label{c_imcp}
\end{align}
Since the sum of the probabilities must be unity at each point, 
\begin{align}
\sum_{k=1}^K r_k(x) = 1\qquad(\forall x \in S).\label{prob}
\end{align}
The macroscopic vorticity field \(\overline{q}\) is defined by
\begin{align*}
\overline{q}(x) = \sum_{k=1}^K Q_kr_k(x)\qquad (x\in S).
\end{align*}
The combination of \(r_1,\cdots,r_K\) is called the macroscopic state or macrostate \citep{robert1991statistical}. The macroscopic state must be consistent with the conservation laws of the flow system. Hence, we impose constraints corresponding to the conservation laws \eqref{energy}, \eqref{am1}--\eqref{am3} by replacing the microscopic field \(q\) with the macroscopic vorticity field \(\overline{q}\). The conservation of energy is given by
\begin{align}
-\frac{1}{4\pi}\int_S\frac{1}{2} \overline{\psi}\,\overline{q} {\rm d}S = E_{0}, \label{c_energy}
\end{align}
where \(\overline{\psi}\) is the macroscopic stream function fulfilling \(\Delta \overline{\psi}=\overline{q}\), and \(E_0\) is the energy value of the initial field. The conservation of angular momentum becomes 
\begin{align}
&\frac{1}{4\pi}\int_S \overline{q} \mu {\rm d} S = M_1^{\rm ini}, \label{c_am1}\\
&\frac{1}{4\pi}\int_S \overline{q} \sqrt{1-\mu^2}\cos\lambda {\rm d}S = M_2^{\rm ini}, \label{c_am2}\\
&\frac{1}{4\pi}\int_S \overline{q} \sqrt{1-\mu^2}\sin\lambda {\rm d} S = M_3^{\rm ini}, \label{c_am3}
\end{align}
where \(M_1^{\rm ini},M_2^{\rm ini},M_3^{\rm ini}\) are the initial values of each component of the angular momentum. Replacing \(q\) with \(\overline{q}\) is justified by the fact that these conserved quantities are the integrals of \(q\) multiplied by some smooth functions \citep{robert1991statistical}. The conservation of Casimir invariants can be regarded as a natural consequence of the condition of incompressibility \eqref{c_imcp}.

For a macroscopic state \(r_1,\cdots,r_K\), the mixing entropy \(S_{\rm mix}\) is defined by
\begin{align}
 S_{\rm mix} := -\frac{1}{4\pi}\sum_{k=1}^K\int_S r_k(x) \log r_k(x) {\rm d} S.
 \label{entropy_def}
\end{align}
\cite{robert1991maximum} proved that the vast majority of microscopic states concentrate on a macroscopic state at which the mixing entropy \(S_{\rm mix}\) is maximized. Hence, it is expected that the organized structure appearing after turbulent mixing can be obtained by finding such a macroscopic state. Thus, we need to consider a maximizing problem of the mixing entropy \(S_{\rm mix}\) under the constraints \eqref{prob}, \eqref{c_imcp}, \eqref{c_energy}, \eqref{c_am1}--\eqref{c_am3}. We call this problem the mixing entropy maximization problem, and call the solutions of this problem the statistical equilibria. We also call critical points of the mixing entropy  statistical equilibria, since not only the global maximum of \(S_{\rm mix}\), but also local maxima or saddles may, in some cases, be important.

\subsection{Microcanonical problem}
The mixing entropy maximization problem described in the previous subsection is referred to as the microcanonical problem, just as in general statistical mechanics, since energy conservation is included in the constraints of the maximization. In the microcanonical problem, the feasible region of the maximization is nonconvex due to the energy constraint. Thus, the problem is a nonlinear and nonconvex problem, in which there can be multiple local maxima.

\subsection{Canonical problem}
Let us introduce the inverse temperature parameter \(\beta\) and consider maximization of the free energy function
\begin{align*}
 F_{\beta}=S_{\rm mix}-\beta E
\end{align*}
under the constraints \eqref{prob}, \eqref{c_imcp}, \eqref{c_am1}--\eqref{c_am3}. Note that the energy constraint is not imposed. This problem is called the canonical problem, similar to the case of  statistical mechanics in which a system is attached to a large heat bath. Here, \(\beta\) can be negative. Admitting negative temperatures, as discovered by \cite{onsager1949statistical}, is an important feature of the statistical mechanics of two-dimensional fluids.

Mathematically, the canonical problem can be analyzed more easily than the microcanonical problem, since the feasible region of maximization is now convex. \cite{robert1991statistical} showed that there exists some \(\beta_c<0\) such that if \(\beta > \beta_c\), then the solution to the maximization problem is unique. On a sphere, this proposition leads to the fact that the solution to the problem has a zonal symmetry about some rotational axis of the sphere if \(\beta>\beta_c\) (in particular if \(\beta=0\)). From the method of Lagrange multiplier, the sets of critical points for the microcanonical and canonical problems coincide. We can see from the general theory of mathematical optimization that a local maximum of the microcanonical problem is also a local maximum of the canonical problem for some \(\beta\). Meanwhile we cannot tell that a local maximum of the canonical problem is also a local maximum of the microcanonical problem for the corresponding energy value. If the nature of a critical point (local maximum, saddle, or local minimum) differs between the two problems, then it is called that there is an ``ensemble inequivalence'' \citep{bouchet2005classification}.

\section{Numerical methods}\label{methods}
In this section, we propose numerical methods of computing the statistical equilibria, i.e., solutions to the microcanonical or canonical problems described in the previous section. For the canonical problem, two different methods are proposed. For the microcanonical problem, we propose an algorithm which enables us to search for equilibria nearly globally, by considering the geometry of the feasible region of the problem.

\subsection{The bridge-building method for the canonical problem}\label{newton_method}
The first method for the canonical problem is to solve a set of equations of Lagrange multipliers for the constraints by using the Newton's method, and then apply the ``bridge-building method'' described later to obtain other solutions. 

First, we describe the discretization of the sphere. We take \(I\times J\) points \((\lambda_i, \mu_j)\,(i=1,\cdots,I,\,j=1,\cdots,J)\) as the points of a grid on the sphere, where \(\lambda_i=2\pi i/I\), \(\mu_j\) is the \(j\)-th Gaussian node defined in Appendix, \(I\) is an even number, and \(J=I/2\). Let \(N\) be the truncation wavenumber of spherical harmonics expansion of functions on the sphere. {Since the energy conservation is imposed as one of the constraints in the MRS theory, we must evaluate the energy of each macroscopic vorticity field. By introducing the spherical harmonics expansion, the energy is expressed in a simple quadratic form of the expansion coefficients.} We impose the inequality \(I\geq 3N+1\), which is necessary for avoiding aliasing error arising from the advection term when using the spectral method on a sphere \citep[see e.g.,][]{durran2010numerical}. Although not all aliasing errors can be avoided because our calculation involves logarithmic functions, we adopt the criterion of truncation wavenumber given by the inequality. We discretize the macroscopic state \(\{r_k\}\) by
\begin{align*}
 r_{ijk} := r_k (\lambda_i,\mu_j)\qquad (1\leq i\leq I,\,1\leq j\leq J,\,1\leq k\leq K).
\end{align*}
Then the value of the macroscopic vorticity at each grid point \(\overline{q}_{ij}=\overline{q}(\lambda_i,\mu_j)\) is given by
\begin{align}
 \overline{q}_{ij}= \sum_{k=1}^K Q_k r_{ijk}. \label{qij}
\end{align}
By approximating integrals, constraint \eqref{c_imcp} becomes
\begin{align}
 \sum_{i=1}^I\sum_{j=1}^J w_j r_{ijk} -\frac{1}{2\pi} S_k =0\qquad(k=1,\cdots,K),\label{d_imcp}
\end{align}
where the \(w_j\)'s are the normalized Gaussian weights defined in the Appendix. Note that the integral of a function over a sphere is approximated as follows:
\begin{align*}
 \frac{1}{4\pi}\int_S f(\lambda,\mu) {\rm d} S \approx \frac{1}{2}\sum_{i=1}^I\sum_{j=1}^J w_j f(\lambda_i,\mu_j).
\end{align*}
Discretization of the constraint \eqref{prob} is given by
\begin{align}
 \sum_{k=1}^K r_{ijk} - 1=0 \qquad(\forall(i,j)\neq(I,J)).\label{d_prob}
\end{align}
Note that the case of \((i,j)=(I,J)\) is excluded, since this case follows if we assume that the other cases and \eqref{d_imcp} hold. We can define the discretized form of the mixing entropy as below:
\begin{align*}
 S_{\rm mix} = -\frac{1}{2}\sum_{i=1}^I\sum_{j=1}^J\sum_{k=1}^K w_j r_{ijk}\log r_{ijk}.
\end{align*}
Recall that the \(r_{ijk}\)'s are nonnegative, i.e.,
\begin{align}
 r_{ijk} \geq 0\qquad (1\leq i\leq I,\,1\leq j\leq J,\,1\leq k\leq K).\label{positiveness}
\end{align}

To make the structure of macroscopic vorticity fields clear, we expand \(\overline{q}\) by using the spherical harmonics (the definition of spherical harmonics is given in the Appendix) as
\begin{align}
\overline{q}(\lambda, \mu)=\sum_{n=0}^\infty \sum_{m=-n}^n \hat{\zeta}_{m,n} Y_{m,n}(\lambda, \mu).\label{expansion}
\end{align}
Here, \(\hat{\zeta}_{m,n}=\hat{\xi}_{m,n}-\sqrt{-1}\hat{\eta}_{m,n}\,(\hat{\xi}_{m,n},\hat{\eta}_{m,n}\in\mathbf{R}\) {are the real part and minus the imaginary part of \(\hat{\zeta}_{m,n}\), respectively, and} \(|m|\leq n)\) are the expansion coefficients, which are given by
\begin{align*}
\hat{\zeta}_{m,n}= \frac{1}{4\pi}\int_{-1}^1\int_{0}^{2\pi} \overline{q}(\lambda,\mu) Y_{m,n}(\lambda,\mu)^* d\lambda d\mu.
\end{align*}
Here, for \(a\in\mathbf{C}\), \(a^*\) denotes the complex conjugate of \(a\).

We can now formulate the canonical problem in a discretized form. We treat the \(r_{ijk}\)'s and \(\hat{\zeta}_{m,n}\)'s as the variables in the maximization problem. {Introducing also \(\hat{\zeta}_{m,n}\)'s as the variables is the key to solving this nonlinear and nonconvex maximization problem. The choice makes it easier to search for statistical equilibria with considering the energy conservation that is expressed by the spectral coefficients of the macroscopic vorticity field.}
For numerical computation, we adopt the triangular truncation in \eqref{expansion}, so that the integers \((m,n)\) take values that satisfy \(0\leq n\leq N,\,|m|\leq n\). We should note that \(\hat{\zeta}_{0,0}=0\) because the integral of \(\overline{q}\) over the sphere is zero and further observe that \(\hat{\zeta}_{-1,1},\hat{\zeta}_{0,1},\hat{\zeta}_{1,1}\) are fixed by the conservation law of angular momentum. In addition, \(\hat{\zeta}_{-m,n}=\hat{\zeta}_{m,n}^*\) must be satisfied since \(\overline{q}\) takes real values. In particular, \(\hat{\eta}_{0,n}\,(n=1,\cdots,N)\) are all zero. Therefore, among the expansion coefficients of \(\overline{q}\), there are \((N+1)^2-4\) real values that are considered to be independent. We can take \(\hat{\xi}_{m,n}\,(n\geq 2,0\leq m\leq n),\,\hat{\eta}_{m,n}\,(n\geq 2,\,1\leq m\leq n)\) as variables. Remembering that \(\overline{q}_{ij}\) is given by \eqref{qij}, the discretized relation between the expansion coefficients and the macroscopic vorticity field is expressed by
\begin{align}
 \frac{1}{2}\sum_{i,j,k}w_j Q_k r_{ijk}P_{0,n}(\mu_j)-\hat{\xi}_{0,n}=0, \label{sg_n0}\\
 \frac{1}{2}\sum_{i,j,k} w_j Q_k r_{ijk}P_{m,n}(\mu_j)\cos m\lambda_i-\hat{\xi}_{m,n}= 0,\label{sg_cos}\\
 \frac{1}{2}\sum_{i,j,k} w_j Q_k r_{ijk}P_{m,n}(\mu_j)\sin m\lambda_i-\hat{\eta}_{m,n}= 0.\label{sg_sin}
\end{align}
We impose these relations as constraints in the maximization problem. The functions \(P_{m,n}(\mu)\) are the associated Legendre functions (for the definition, see the Appendix). By using the expansion coefficients, the energy is written as 
\begin{align}
 E= \frac{1}{2}\sum_{n=2}^N\frac{\hat{\xi}_{0,n}^2}{n(n+1)}+\sum_{n=2}^N\sum_{m=1}^n \frac{\hat{\xi}_{m,n}^2+\hat{\eta}_{m,n}^2}{n(n+1)}\label{d_energy}.
\end{align}
Here, we have subtracted the contribution of the fixed coefficients \(\hat{\xi}_{0,1},\hat{\xi}_{1,1},\hat{\eta}_{1,1}\) from the energy expression.

Using the variables introduced above, the constraints of angular momentum can be expressed in the same form as \eqref{sg_n0}--\eqref{sg_sin}:
\begin{align}
  \frac{1}{2}\sum_{i,j,k}w_jQ_kr_{ijk}P_{0,1}(\mu_j) - \hat{\xi}_{0,1}=0,\label{d_am1}\\
  \frac{1}{2}\sum_{i,j,k}w_jQ_k r_{ijk}P_{1,1}(\mu_j)\cos \lambda_i - \hat{\xi}_{1,1}=0,\label{d_am2}\\
  \frac{1}{2}\sum_{i,j,k}w_jQ_k r_{ijk}P_{1,1}(\mu_j)\sin\lambda_i-\hat{\eta}_{1,1}=0,\label{d_am3}
\end{align}
where \(\hat{\xi}_{0,1},\hat{\xi}_{1,1},\hat{\eta}_{1,1}\) are the initial values of the spectral coefficients which are fixed by the angular momentum conservation. Now, the discretized form of the canonical problem is formulated as follows: maximize the free energy \(F_{\beta}=S_{\rm mix}-\beta E\) under the constraints \eqref{d_imcp}, \eqref{d_prob}, \eqref{positiveness}, \eqref{sg_n0}--\eqref{sg_sin}, and  \eqref{d_am1}--\eqref{d_am3}.

To find critical points of the free energy \(F_\beta\), we use the Lagrange multiplier method. Let \(B_k\) be the {left}-hand side of \eqref{d_imcp} and \(C_{ij}\) be the {left}-hand side of \eqref{d_prob}. Further, let \(A_n,U_{m,n}\), and \(V_{m,n}\) be the left-hand sides of the relations \eqref{sg_n0}, \eqref{sg_cos}, and \eqref{sg_sin}, respectively, and let \(A_1,U_{1,1}\), and \(V_{1,1}\) be the left-hand sides of the angular momentum conservation relations \eqref{d_am1}, \eqref{d_am2}, and \eqref{d_am3}. Then, we can introduce the Lagrangian defined by
\begin{align}
  L=& F_\beta+a_1A_1+u_{1,1}U_{1,1}+v_{1,1}V_{1,1}-\beta\sum_{n=2}^N\sum_{m=1}^n(u_{m,n}U_{m,n}+v_{m,n}V_{m,n})\nonumber\\
  &-\beta \sum_{n=2}^N a_nA_n+\frac{1}{2}\sum_{k=1}^K(b_k+1)B_k+\frac{1}{2}\sum_{i,j}w_j c_{ij}C_{ij}.\label{lagrangian}
\end{align}
Here, we have defined \(c_{IJ}:=0\) for convenience. Then, an arbitrary critical point of \(F_\beta\) under the constraints is a critical point of \(L\) with no constraints for some combination of Lagrange multipliers. The condition of critical points \(\partial L/ \partial r_{ijk}=0\) yields
\begin{align}
r_{ijk}=\exp X_{ijk}(a,u,v,b,c), \label{r_ijk}
\end{align}
where 
\begin{align}
X_{ijk}(a,u,v,b,c)=&a_1Q_kP_{0,1}(\mu_j) \nonumber\\
&+(u_{1,1}\cos \lambda_i
+v_{1,1}\sin\lambda)Q_kP_{1,1}(\mu_j)
-\beta \sum_{n=2}^N a_nQ_kP_{0,n}(\mu_j)\nonumber\\
&-\beta \sum_{n=2}^N\sum_{m=1}^n (u_{m,n}\cos m\lambda_i
+v_{m,n}\sin m\lambda_i)Q_kP_{m,n}(\mu_j)+b_k+c_{ij}.\label{xijk}
\end{align}
Here, \(a=\{a_n\}, u=\{u_{m,n}\}, v=\{v_{m,n}\}, b=\{b_k\}\), and \(c=\{c_{ij}\}\). In addition, \(\partial F_\beta/\partial\hat{\xi}_{m,n}=\partial F_\beta/\partial\hat{\eta}_{m,n}=0\) yields
\begin{align}
&\hat{\xi}_{0,n}=n(n+1)a_n\qquad(2\leq n\leq N),\label{zetan0}\\
&\hat{\xi}_{m,n}=\frac{1}{2}n(n+1)u_{m,n}\qquad(2\leq n\leq N, 1\leq m\leq n),\label{zetacos}\\
&\hat{\eta}_{m,n}=\frac{1}{2}n(n+1)v_{m,n}\qquad(2\leq n \leq N,1\leq m\leq n).\label{zetasin}
\end{align}
All the variables are thus expressed by the Lagrange multipliers. Substituting these expressions \eqref{r_ijk} and \eqref{zetan0}--\eqref{zetasin} into the constraints \eqref{d_imcp}, \eqref{d_prob}, \eqref{sg_n0}--\eqref{sg_sin}, and \eqref{d_am1}--\eqref{d_am3}, we obtain a set of equations of the multipliers. We can obtain a critical point of the free energy \(F_\beta\) by solving this system of equations numerically using Newton's method. However, in applying Newton's method, it is necessary to provide a good initial estimate of the solution, which is, in general, difficult if we want a symmetry-breaking solution. We overcome this difficulty in the three-step procedure described below.

\subsubsection{Step 1: Finding the solution to the case of \(\beta=0\)}
\label{beta=0}
\noindent When \(\beta=0\), the problem reduces to the maximization of the mixing entropy \(S_{\rm mix}\) under constraints \eqref{d_imcp}, \eqref{d_prob}, and \eqref{d_am1}--\eqref{d_am3}. We can use only the \(r_{ijk}\)'s as the variables of maximization, since the expansion coefficients of the vorticity field play no role. In this case, all the constraints are given in the form of an equality or inequality of linear functions; thus, the feasible region is a convex closed set. Since \(S_{\rm mix}\) is a strictly concave function, a critical point of \(S_{\rm mix}\) is unique if it exists in the interior of the feasible region, and \(S_{\rm mix}\) attains the maximum value at the point. Hence, iterations of Newton's method converge even if the initial guess is far from the solution. For example, we can set as \(b_k=\log (Iw_k/2),c_{ij}=0,a_n=0,u_{m,n}=v_{m,n}=0\) for the initial guess. As noted in section \ref{MRS_theory}, the solution obtained here has zonal symmetry about some rotational axis.
 
\subsubsection{Step 2: Finding a solution for given \(\beta\)}
\label{general_beta}
\noindent Let \(\beta_T\) be the given value of the inverse temperature. For \(\beta=\beta_T\), we find a critical point by changing the value of \(\beta\) gradually from \(0\) to \(\beta_T\) and by using Newton's method in each step of changing \(\beta\). Let \(y_0=(a_n,u_{m,n},v_{m,n},b_k,c_{ij})\) be the tuple of multipliers that gives the critical point for \(\beta=0\,(=:\beta_0)\), obtained in Step 1, and let \(G(y;\beta)=0\) denote the system of equations of Lagrange multipliers for an arbitrary \(\beta\). Here, \(y\) is a tuple of the multipliers and \(G\) is the vector-valued function corresponding to the left-hand sides of \eqref{d_imcp}, \eqref{d_prob}, \eqref{sg_n0}--\eqref{sg_sin}, and \eqref{d_am1}--\eqref{d_am3}. Using this notation, we now have \(G(y_0;\beta_0)=0\). First, we solve the system of equations for \(\beta=\beta_1:= \beta_0+\Delta\beta\), where \(\Delta\beta\) is a small number. Depending on the value of \(\beta_T\), \(\Delta\beta\) can be positive or negative. The initial estimate of the solution used here is \(y_0\), since the solution of \(G(y;\beta_1)=0\) is expected to be close to \(y_0\), provided that \(|\Delta\beta|\) is small. Let \(y_1\) be the obtained solution of \(G(y;\beta_1)=0\). Similarly, we solve the equation \(G(y; \beta_2)=0,\,\beta_2=\beta_1+\Delta\beta\) to obtain the solution \(y_2\) by Newton's method, in which \(y_1\) is used as the initial estimate of the solution. We continue this procedure until \(\beta\) reaches \(\beta_T\). Here, \(\Delta\beta\) can take a different value for each step, depending on whether Newton's method converges. Note that the obtained critical point tends to have the zonal symmetry inherited from the critical point for \(\beta=0\).

\subsubsection{Step 3: Bridge-building}
\label{bridgebuilding}
\noindent There may be other critical points than those obtained by the procedure described above if \(\beta\) is negative with a sufficiently large absolute value, where the concavity of \(F_\beta\) is not guaranteed \citep{robert1991statistical}. The ``bridge-building method'' described below allows us to find and compute such critical points from the already obtained one. This method introduces a numerically trackable path from an already-obtained critical point to an unknown one by considering a new problem with relaxed constraints. The method allows us to search for critical points at which the absolute value of the expansion coefficient for a given set of wavenumbers \((m,n)\) is large. 

As one can see from \eqref{r_ijk} and \eqref{xijk}, the multipliers \(a_n,u_{m,n},v_{m,n}\,(n\geq 1,\,1\leq m\leq n)\) are closely related to the spectral structure of the macroscopic vorticity field at the critical point. Let us search for a critical point at which the wavenumber \((m_0,n_0)\) coefficient is substantially different from that of the already-known critical point. Here, we assume \(n_0\geq 2,m_0\geq 1\). We can then consider the problem (BP) described below:\\
\textbf{(BP)} Find the critical points of 
\begin{align*}
  F_{\beta}-\beta \gamma U_{m_0,n_0}
\end{align*}
    under the constraints of the canonical problem, excluding  constraint \eqref{sg_cos} for \((m,n)=(m_0,n_0)\) (i.e., \(U_{m_0,n_0}=0\)). Here, \(\gamma\) is a parameter.\\

Let \(y_0\) be the already-obtained solution to the equation \(G(y; \beta)=0\), and \(u_{m_0,n_0}^0\) be the value of the multiplier \(u_{m_0,n_0}\) at \(y_0\). If \(\gamma=u_{m_0,n_0}^0\), then the already-known critical point is also a solution to (BP). The problem (BP) is solved by using Newton's method in the same way as it was used to find critical points of \(F_\beta\) in the previous procedures. We consider the problem (BP) for a \(\gamma\) whose value is slightly different from \(u_{m_0,n_0}^0\). We can then solve (BP) by choosing \(y_0\) (from which the \(u_{m_0,n_0}\) component is excluded) as the initial value of iteration, since the solution is considered to be close to \(y_0\). We repeat the procedure by changing \(\gamma\) and solving (BP), as we did for \(\beta\) in Step 2. Then, we obtain a curve of the solutions of (BP), having \(\gamma\) as the parameter. Suppose that there is another point on the curve at which \(U_{m_0,n_0}=0\) holds, and let \(\gamma_1\) be the value of \(\gamma\) at the point. Then the macroscopic state given by the solution of (BP) for \(\gamma=\gamma_1\) satisfies all constraints of the canonical problem, and is a critical point of Lagrangian \eqref{lagrangian}. This means that the macroscopic state is a critical point of the canonical problem. Note that this method also works for \(m_{0}=0\) by replacing \(U_{m_0,n_0}\) with \(A_{n_0}\) and for \(n_0=1\) by replacing \(U_{m_0,n_0}\) with \(A_1\) or \(U_{1,1}\) and considering \(F_\beta-\gamma A_1\) or \(F_\beta-\gamma U_{1,1}\) in (BP).

The procedure described above is used to track a curve of solutions to the relaxed problem (BP), starting from the already-obtained solution and reaching a new solution (Figure \ref{bridge_figure}). We call the curve a ``bridge'', since it connects two distinct solutions; accordingly, we call the method the ``bridge-building method.''  Although there is no guarantee that a new solution exists somewhere on the bridge, the method is useful for, among various purposes, finding non-zonal statistical equilibria from a known zonal equilibrium, as will be shown later in an example. Note that once we obtain a new equilibrium as the terminal of the bridge, we can also obtain a branch of critical points parameterized by \(\beta\) by tracking solutions of the canonical problem for different \(\beta\)'s from the new solution. 

{Searching a new critical point by using the bridge-building method is analogical to adjusting the internal energy of a system by changing the temperature of the external heat bath in statistical thermodynamics. Indeed, if we consider the parameter \(\gamma\) as the inverse temperature and the function \(U_{m_0,n_0}\) as the internal energy, respectively, the analogy becomes obvious.}

\begin{figure}[htbp]
\centering
\includegraphics[width=10.0cm]{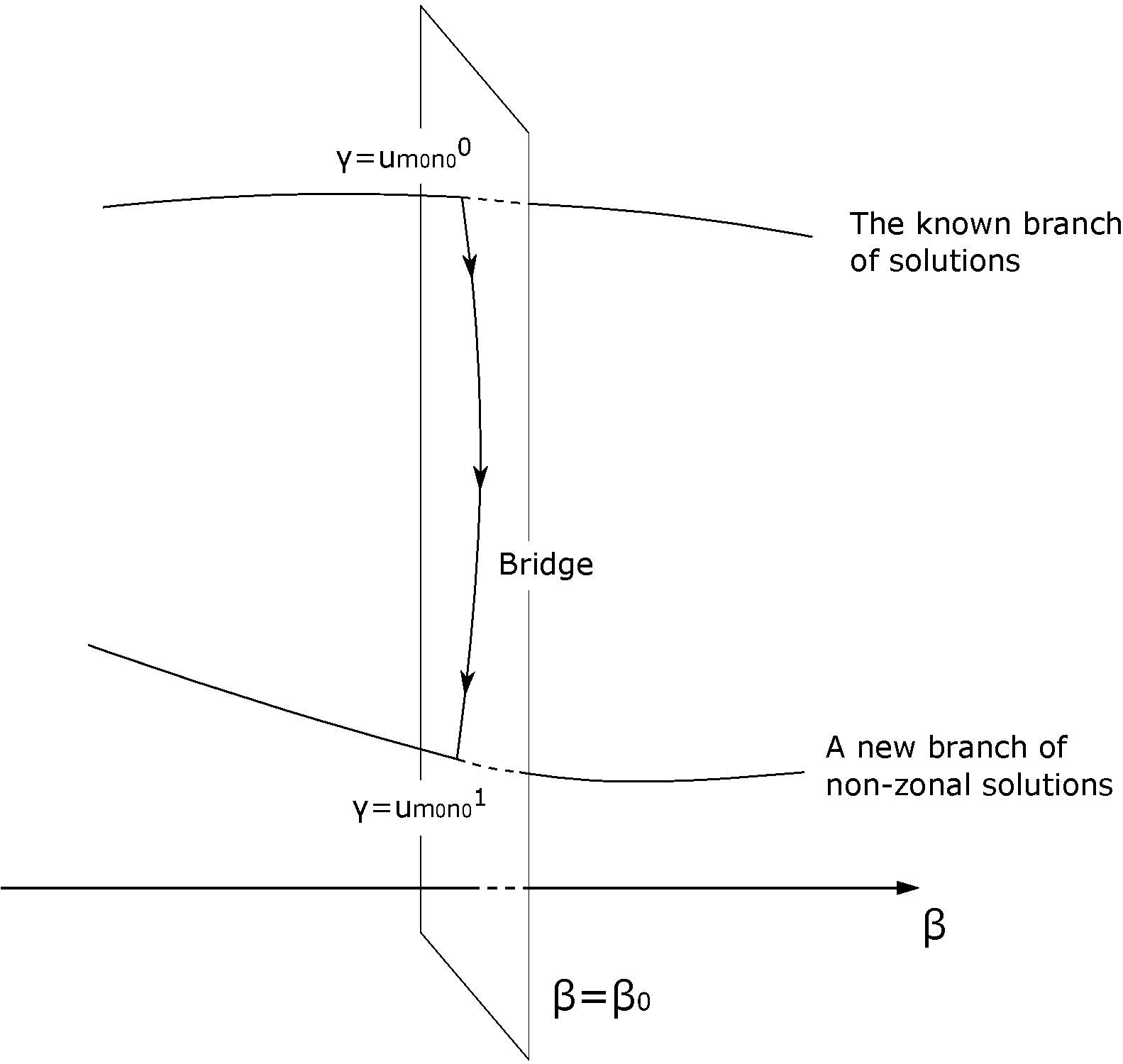}
\caption{Schematic illustration of the bridge-building method. On a plane of fixed \(\beta=\beta_0\), the problems (BP) for different values of \(\gamma\) are solved, starting from the already-known solution at \(\gamma=u_{m_0,n_0}^0\). If the solution at \(\gamma=u_{m_0,n_0}^1\) satisfies \(U_{m_0,n_0}=0\), then it is a new critical point of \(F_{\beta_0}\).}
\label{bridge_figure}
\end{figure}

\subsection{Gradient method for the canonical problem}
\label{grad_method}
 The bridge-building method provides a very powerful way to search for multiple equilibria, since we can find critical points for a wide variety of wavenumber structures by computing bridges for various \((m_0,n_0)\). However, one cannot tell whether a new solution exists on the bridge for a given \((m_0,n_0)\) until the bridge is actually computed. In addition, the nature of the critical point is not determined by this method, which means that the obtained critical points may include saddles. Thus, it would be desirable to have an algorithm that can automatically detect local maxima of \(F_\beta\). To meet this need, we propose another numerical method to solve the canonical problem using the gradient of \(F_\beta\). With the proposed gradient method, the obtained critical points are guaranteed to be local maxima of \(F_\beta\). Furthermore, if the initial point of the search is a saddle, a small perturbation will induce an automatic transition to some local maximum point.
 
 In the bridge-building method, the number of variables is \(IJK+(N+1)^2-4\), with \(IJK\) corresponding to the \(r_{ijk}\)'s and \((N+1)^2-4\) corresponding to the expansion coefficients of \(\overline{q}\). Since \(I\geq 3N+1\) and \(J=I/2\), the number of variables is \(O(N^3)\) provided that \(K\) is \(O(N)\). However, if the same variables were used in the gradient method, the cost of the computation of the gradient of \(F_\beta\) in each step would be very high. Thus, in our gradient method, we consider a subproblem in order to reduce the number of variables. The subproblem is defined as follows: Let the expansion coefficients of \(\overline{q}\), \((\hat{\xi}_{m,n},\hat{\eta}_{m,n})\),  be given and fixed. Then, we consider a subproblem in which we search for a macroscopic state \(\{r_{ijk}\}\) that maximizes the mixing entropy \(S_{\rm mix}\) under constraints corresponding to \eqref{sg_n0}--\eqref{sg_sin} and \eqref{d_am1}--\eqref{d_am3}. By solving this subproblem, macroscopic state \(\{r_{ijk}\}\) can be viewed as a function of the expansion coefficients \(\hat{\xi}_{m,n},\hat{\eta}_{m,n}\), and only \((N+1)^2-4\) expansion coefficients are needed as the variables of the maximization problem of \(F_\beta\). Note that considering only the macroscopic states determined by solving the subproblem is sufficient for solving the canonical problem, since we are interested in the maximization of \(F_\beta = S_{\rm mix}-\beta E\), and \(E\) is a function of the expansion coefficients only. The detailed procedure is explained below.
 
 First, in order to reduce the numerical cost, we modify the constraints \eqref{d_am1}--\eqref{d_am3} and \eqref{sg_n0}--\eqref{sg_sin}. In place of these constraints, we impose the following constraints:
 \begin{align}
 \overline{q}_{ij} = \sum_{n=1}^N \sum_{m=-n}^n \hat{\zeta}_{m,n}Y_{m,n}(\lambda_i,\mu_j)\qquad(\forall(i,j)\neq (I,J)),\label{sg_grad}
 \end{align}
 where \(\hat{\zeta}_{m,n}=\hat{\xi}_{m,n}-\sqrt{-1}\hat{\eta}_{m,n}\) and \(\overline{q}_{ij}\) is the value of the macroscopic vorticity field at each grid point, which is determined by \eqref{qij}. The constraints \eqref{sg_grad}, which correspond to the inverse Fourier transform on the sphere, are imposed for all the grid points except for the case of \((i,j)=(I,J)\), since it follows immediately from \eqref{d_imcp}, \eqref{d_prob}, and \eqref{sg_grad}. The constraints of angular momentum, \eqref{d_am1}--\eqref{d_am3}, are imposed by fixing \(\hat{\xi}_{0,1},\hat{\xi}_{1,1}\), and \(\hat{\eta}_{1,1}\) in \eqref{sg_grad}. The subproblem can now be formulated as follows:\\
 \textbf{(SP)}
 For given expansion coefficients \(\{\hat{\zeta}_{m,n}\}\), find the macroscopic state \(\{r_{ijk}\}\) maximizing the mixing entropy \(S_{\rm mix}\), under the constraints \eqref{d_prob}, \eqref{d_imcp}, \eqref{positiveness}, and \eqref{sg_grad}.
 
 In the subproblem (SP), the feasible region is a convex closed set of the space of \(\{r_{ijk}\}\), since the constraints consist of equalities or inequalities of linear functions. Meanwhile, \(S_{\rm mix}\) is a strictly convex function on the closed set. Therefore, if there is some critical point of \(S_{\rm mix}\) in the interior of the feasible region, then such a critical point is unique and \(S_{\rm mix}\) takes the maximum value at the point. Thus, it is not difficult to solve (SP) numerically by using Newton's method (even if the initial guess of the solution is not close to the true solution), and the values of the \(r_{ijk}\)'s can be regarded as functions of the coefficients \(\hat{\xi}_{m,n},\hat{\eta}_{m,n}\) by solving the subproblem numerically. Let \(Z=(\hat{\xi}_{0,2},\cdots,\hat{\xi}_{0,N},\hat{\xi}_{1,2},\hat{\eta}_{1,2},\cdots,\hat{\xi}_{N,N},\hat{\eta}_{N,N})\in \mathbf{R}^{(N+1)^2-4}\) denote a vector of the expansion coefficients, and let \(r_{ijk}(Z)\) and \(S_{\rm mix}(Z)\) denote \(r_{ijk}\) and \(S_{\rm mix}\) as functions of \(Z\), respectively. It can be shown that the set of \(Z\) for which the solution of (SP) exists in the interior of the feasible region for \(\{r_{ijk}\}\) is a bounded convex closed subset \(P\) of \(\mathbf{R}^{(N+1)^2-4}\). In the interior of \(P\), the discretized macroscopic state variables, the \(r_{ijk}(Z)\)'s, and therefore \(S_{\rm mix}(Z)\), are differentiable functions of \(Z\). Their derivatives can be numerically computed from the Lagrange multipliers used for solving (SP). (For the geometrical properties of the feasible region and details on the numerical implementation, see the Supplement.)
 
 In the proposed gradient method, the following steps are taken. Step 1: For a given \(Z\), solve the subproblem (SP). Step 2: Compute the gradient \(\nabla S_{\rm mix}\), and compute the gradient \(\nabla F_\beta\). Step 3: Update \(Z\) to \(Z+\varepsilon \nabla F_\beta\) and return to Step 1, where \(\varepsilon\) is the step size. Note that the process is equivalent to solving the differential equation 
 \begin{align}
 \frac{{\rm d} Z}{{\rm d} t} = \nabla F_{\beta}(Z)\label{grad_evo}
 \end{align}
 numerically by Euler's method. One can also use the Runge-Kutta method for more stable computation. The critical point of \(F_\beta\) obtained by the procedure in section \ref{general_beta} would be a reasonable choice for the initial starting point of \(Z\), since such a critical point tends to be a saddle if \(\beta\) is negative with a large absolute value.
 
\subsection{Method for solving the microcanonical problem}\label{microcanonical}
 As discussed in section \ref{MRS_theory}, the canonical and microcanonical problems have the same set of critical points, but the nature of each of the critical points may differ. Although the gradient method for the canonical problem allows us to compute a local maximum of free energy \(F_\beta=S_{\rm mix}-\beta E\), the obtained point is not necessarily a local maximum of \(S_{\rm mix}\) in the microcanonical problem for the corresponding energy. Thus, a special numerical method for the microcanonical problem is needed, one that differs from the methods used for the canonical problem. We propose a method that, with the help of the gradient method algorithm used for the canonical problem, can be effectively applied to the microcanonical problem
 
 The difficulty with the microcanonical method is that the feasible region likely consists of multiple fractions of nonconvex hypersurfaces, since the energy constraint
 \begin{align}
 E(Z) = \frac{1}{2}\sum_{n=2}^N \frac{\hat{\xi}_{0,n}^2}{n(n+1)}+\sum_{n=2}^N \sum_{m=1}^n \frac{\hat{\xi}_{m,n}^2+\hat{\zeta}_{m,n}^2}{n(n+1)}=E_0
 \label{energy_const}
 \end{align}
 is imposed. Specifically, it is not clear whether the feasible region is connected, which makes global search difficult. We tackle this difficulty by constructing a new method for computing multiple local maxima of the mixing entropy that takes into account the geometry of the feasible region.
 
  For the method proposed in this subsection, we use \eqref{sg_grad} as the relationship between the expansion coefficients and the grid values and take a point of wavenumber space \(Z=(\hat{\xi}_{0,2},\cdots,\hat{\xi}_{0,N},\hat{\xi}_{1,2},\hat{\eta}_{1,2},\cdots,\hat{\xi}_{N,N},\hat{\eta}_{N,N})\in \mathbf{R}^{(N+1)^2-4}\) as the variable of the maximization problem.  First, we consider the region in which \(Z\) can exist when the energy constraint is not imposed. Such a region is a subset of \(\mathbf{R}^{(N+1)^2-4}\) consisting of \(Z\) for which some \((r_{ijk})_{i,j,k}\) exists that satisfies the constraints \eqref{d_imcp}, \eqref{d_prob}, \eqref{positiveness}, and \eqref{sg_grad}. It can be shown that this subset is a bounded closed subset of \(\mathbf{R}^{(N+1)^2-4}\) with nonempty interior. In fact, the subset is a polytope which is realized as an intersection of multiple half-spaces. The proof of this, which is given in the Supplement, is based on the inequality conditions \eqref{positiveness} and \(IJK\geq (N+1)^2-4\), which follow from the conditions for \(I,J\), and \(N\). Let \(P\subset \mathbf{R}^{(N+1)^2-4}\) denote this convex closed set. Here, we should note the following fact: If \(Z\in \partial P\), i.e., \(Z\)  is on the boundary of \(P\), then for any \(r=(r_{ijk})_{i,j,k}\) satisfying \eqref{d_prob}, \eqref{d_imcp}, \eqref{positiveness}, and \eqref{sg_grad}, there exists some \(i_0,j_0,k_0\) such that \(r_{i_0 j_0 k_0}=0\). In addition, if \(Z\in {\rm int}P\), i.e., \(Z\) is an interior point of \(P\), then there exists some \(r=(r_{ijk})_{i,j,k}\) satisfying \eqref{d_prob}, \eqref{d_imcp}, \eqref{positiveness}, and \eqref{sg_grad} such that \(r_{ijk}>0\) for any \(i,j,k\). Therefore, when \(Z\in \partial P\), one cannot compute \(r_{ijk}(Z)\) by solving (SP), as some of the Lagrange multipliers used in solving (SP) become unbounded and the computation fails. {Indeed, the solution of (SP) can be expressed as
  \begin{align*}
      r_{ijk} = \exp (Q_k a_{ij}+ b_k + c_{ij}),
  \end{align*}
  where \(a_{ij}, b_k\) and \(c_{ij}\) are the Lagrange multipliers. The details are given in the supplement. From the above equation, at least one of the multipliers must be unbounded in order for \(r_{ijk}=0\). The physical reason is difficult to describe, but the reason is associated to the fact that the gradient of \(S_{\rm mix}\) as a function of \(r_{ijk}\)'s is unbounded if \(r_{ijk}=0\). In order for the Lagrangian \(L\) to attain extrimums, the gradient of \(S_{\rm mix}\) and the gradients of functions defining the constraints must be balanced by means of multipliers. This is achieved only by some of the multipliers being unbounded. The reason why some \(r_{ijk}\)'s are zero on the boundary of \(P\) derives from the definition of \(P\). Let us give a more detailed but not mathematically rigorous explanation of this fact. The region \(R^+\) of the macroscopic vorticity distribution \((r_{ijk})_{i,j,k}\) that can be realized by mixing vorticity patches is the set of \((r_{ijk})_{i,j,k}\) that satisfy \eqref{d_imcp}, \eqref{d_prob}, and \eqref{positiveness}. By the inequality constraint \eqref{positiveness}, this set is a (closed) hyperpolyhedral domain of dimension \(O(N^3)\), the boundary of which is defined by the equations \(r_{ijk}=0\) (for some \(i, j, k\)). The map that computes the macroscopic vorticity field \((q_{ij})_{i,j}\) from the macroscopic vorticity distribution \((r_{ijk})_{i,j,k}\) is a linear map from an \(O(N^3)\)-dimensional space to an \(O(N^2)\)-dimensional space. Then this map can be considered a ``projection'' from a high-dimensional space to a lower-dimensional space under an appropriate choice of coordinates. Hence, the boundary of the range of the possible macroscopic vorticity field is the ``shadow'' of the boundary of \(R^+\) by this ``projection''. Since \(P\) represents the feasible region of the macroscopic vorticity field in the space of the spectrum, \(r_{ijk}=0\) for some \(i,j,k\) in \(\partial P\) as well. For more rigorous argument, please refer to the supplement of the manuscript.}
 
 Based on these facts, we define a function \(V(Z)\) by
 \begin{align*}
 V(Z) = -\frac{1}{2}\sum_{i,j,k}w_j \log \tilde{r}_{ijk}(Z),
 \end{align*}
 where \(\tilde{r}_{ijk}(Z)\) is the \(r_{ijk}\) that minimizes \(-(1/2)\sum w_j\log r_{ijk}\) under the constraints \eqref{d_prob}, \eqref{d_imcp}, \eqref{positiveness}, and \eqref{sg_grad} for a given \(Z\in P\). \(\tilde{r}_{ijk}(Z)\) and \(V(Z)\) can be numerically computed in the same way that \(r_{ijk}(Z)\) and \(S_{\rm mix}(Z)\) are computed in section \ref{grad_method}. It is also similar to section \ref{grad_method} in that \(\tilde{r}_{ijk}(Z)\) and \(V(Z)\) are differentiable and their derivatives can be numerically computed. It can be verified that \(V(Z)\) is strictly convex in \(P\). The proofs of these facts are presented in the Supplement. From the above, \({\rm int}P\) can be viewed as the set of \(Z\in P\) such that \(V(Z)<+\infty\). When \(Z\in \partial P\), we cannot numerically compute \(\tilde{r}_{ijk}(Z)\) as in the case of solving (SP). For a sufficiently large \(M>0\), we can introduce a hypersurface which is defined by
 \begin{align*}
 L_M = \{Z\in P | V(Z)=M\}
 \end{align*}
 to approximate the boundary \(\partial P\) from its inner side.  For \(M>0\), \(L_M\) is homeomorphic to the \((N+1)^2-5\)-dimensional sphere, and it is the boundary of the convex set \(\{Z\in P|V(Z)\leq M\}\) in \(\mathbf{R}^{(N+1)^2-4}\). 
 
 We now apply the following strategy. Step 1: Take a point \(Z\in {\rm int}P\) that is not too close to \(\partial P\), and let \(M\) be the value of \(V(Z)\) at the point. Step 2: Consider a linear function of the expansion coefficients \(Z\) and maximize it on the hypersurface \(L_M\). Step 3: From the point obtained by Step 2, change the value of \(E\) gradually to \(E_0\) keeping \(V(Z)\) as small as possible. Step 4: Using the point \(Z\) obtained by Step 3 as the initial starting point, search for the local maximum of the mixing entropy \(S_{\rm mix}(Z)\) on the surface of \(E=E_0\). The details of each step are described below.
 
 In Step 1, we take the expansion coefficients of the critical vorticity field obtained in section \ref{general_beta} as the starting point \(Z\). 
 
 In Step 2, we maximize a linear function
 \begin{align*}
  f(Z) = \langle W, Z\rangle
 \end{align*}
 on the hypersurface \(L_M\) (Figure \ref{linearmax_figure}). Here, \(W\in \mathbf{R}^{(N+1)^2-4}\) is a non-zero vector of \(\mathbf{R}^{(N+1)^2-4}\), and \(\langle\cdot,\cdot\rangle\) denotes the standard inner product of \(\mathbf{R}^{(N+1)^2-4}\). It can be shown that this maximization problem has a unique solution (see the Supplement). We adopt the gradient projection method of \cite{tanabe1974algorithm} for numerical computation of this maximization, in which the vector field of \(\nabla f =W\) is projected on the tangent space of \(L_M\) at each point. Since the value of \(M\) is conserved on the hypersurface \(L_M\), the maximization can be computed stably unless \(M\) is too large. Being able to arbitrarily choose the vector \(W\) is helpful in obtaining \(Z\) points corresponding to the vorticity fields with various wavenumber structures. For example, if we want to obtain the vorticity field with a large coefficient of the wavenumber \((m,n)=(2,2)\), we can choose \(W\) so that \(f(Z)=\hat{\xi}_{2,2}\). 
 
 \begin{figure}[tbp]
 \centering
 \includegraphics[width=8.0cm]{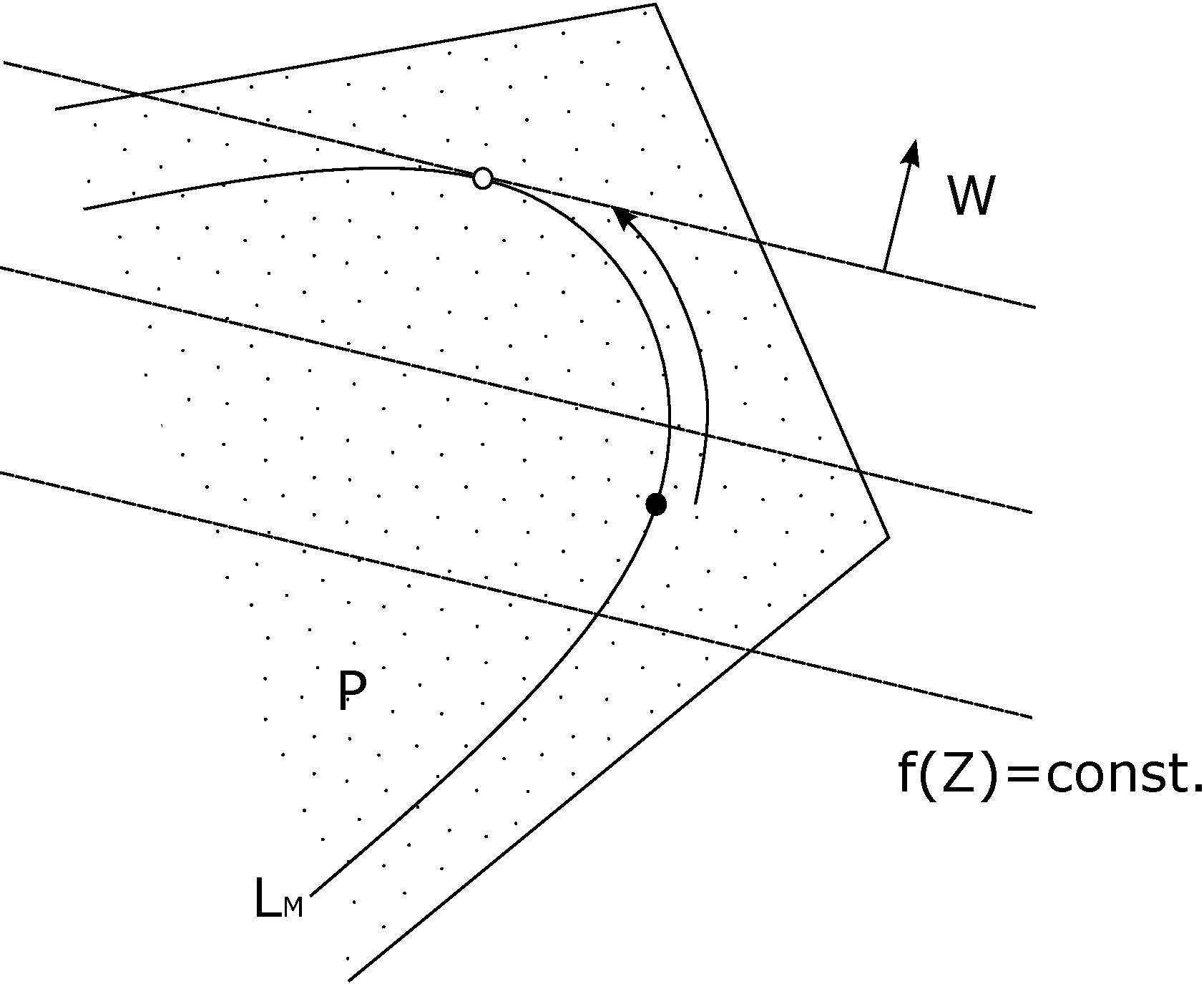}
 \caption{Schematic illustration of the maximization of a linear function. The black circle indicates the initial starting point of \(Z\), and the white circle indicates the end point obtained by maximizing \(f(Z)\) with the gradient projection method. The parallel lines are the level surfaces of \(f(Z)\).}
 \label{linearmax_figure}
 \end{figure}
 
 In Step 3, to make the value of energy \(E\) closer to \(E_0\), we numerically solve the following differential equation in \(\mathbf{R}^{(N+1)^2-4}\)
 \begin{align*}
  \frac{{\rm d} Z}{{\rm d} t}= -\nabla (E-E_2)^2,\qquad Z(0)=Z_1,
 \end{align*}
 until energy \(E\) reaches \(E_2\). Here, \(E_2=E_1+\Delta E\),  \(E_1=E(Z_1)\), and \(Z_1\) is the point \(Z\in P\) obtained by Step 2. Note that \(\Delta E\) has to be small so that \(Z\) does not touch the boundary \(\partial P\). We then consider the minimization of \begin{align*}
 V_{\rm ent}(Z) = -\frac{1}{2}\sum_{i,j,k}w_j \log r_{ijk}(Z)
 \end{align*}
 on the hypersurface of \(E=E_2\) with the gradient projection method, starting from \(Z=Z_2\). This minimization is intended to prevent the components of \((r_{ijk})_{i,j,k}\) that maximizes \(S_{\rm mix}\) under the constraint of \eqref{sg_grad} from approaching zero. Let the new \(Z_1\) be the local minimum of \(V_{\rm ent}(Z)\) obtained above, and repeat these steps until \(E\) reaches \(E_0\) (Figure \ref{upenergy_figure}).

 \begin{figure}[tbp]
 \centering
 \includegraphics[width=8.0cm]{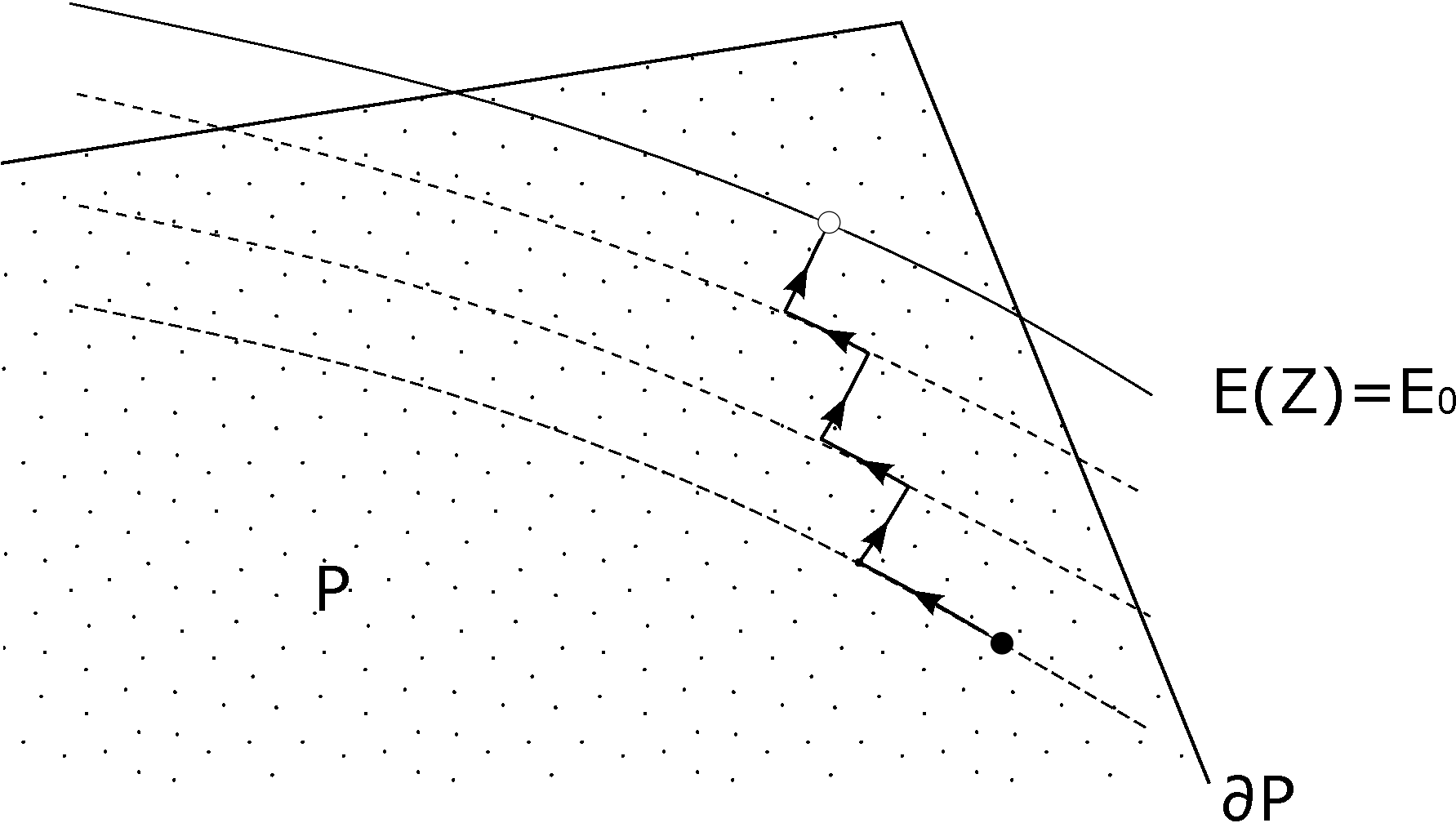}
 \caption{Schematic illustration of Step 3 in computing the equilibria of the microcanonical problem. Starting from \(Z\) obtained in Step 2 (black circle), we repeat minimization of \(V_{\rm ent}(Z)\) on the energy isosurface and moving the point \(Z\) in the normal direction of the energy isosurface alternately, until \(E(Z)\) reaches \(E_0\) (white circle).}
 \label{upenergy_figure}
 \end{figure}
 
 In Step 4, the mixing entropy is maximized. The gradient projection method is applied on the hypersurface defined by \(E(Z)=E_0\), and the point obtained in Step 3 is used as the initial starting point for the maximization. To solve the maximization problem globally is difficult since not only is the feasible region of the microcanonical problem not necessarily connected, as explained above, but it is also the case that the gradient projection method is suitable for a local search rather than a global search. However, we can perform a quasi-global search if we can generate a large number of initial starting points according to Step 2, from each of which we follow Step 3 and Step 4. 
 
\section{Example of computing MRS statistical equilibria}\label{example}
In this section, we show an example of applying the computational methods described in section \ref{methods}. In addition, the computed statistical equilibrium is compared to the final state of the time evolution from a vorticity field.

{In the following, numerical calculations are done by the double-precision computations with an accuracy of about 15 digits, and convergence determinations of the Newton method are done with an accuracy of about \(10^{-10}\) relative error. Therefore, the relative error in the calculation results is considered to be about \(10^{-10}\), but for the sake of simplicity, the numerical values are represented by 6 digits.}

In the numerical calculations presented in this section, ispack-3.0.1 (\url{http://www.gfd-dennou.org/arch/ispack/}), which is designed based on \cite{ishioka2018}, is used for the spherical harmonic transform such as \eqref{expansion} and the equation below it.

\subsection{Initial vorticity field}
We use the following initial vorticity field (Figure \ref{initial_figure}) as a test case for the computational methods described in section \ref{methods}.
 \begin{align}
 q_{\rm initial}(\lambda, \mu) = 12\times 0.07\Omega P_{0,3}(\mu) + 2\Omega\mu. \label{ini_vor}
\end{align}
\begin{figure}[htbp]
\centering
\includegraphics[angle=270,width=10.0cm]{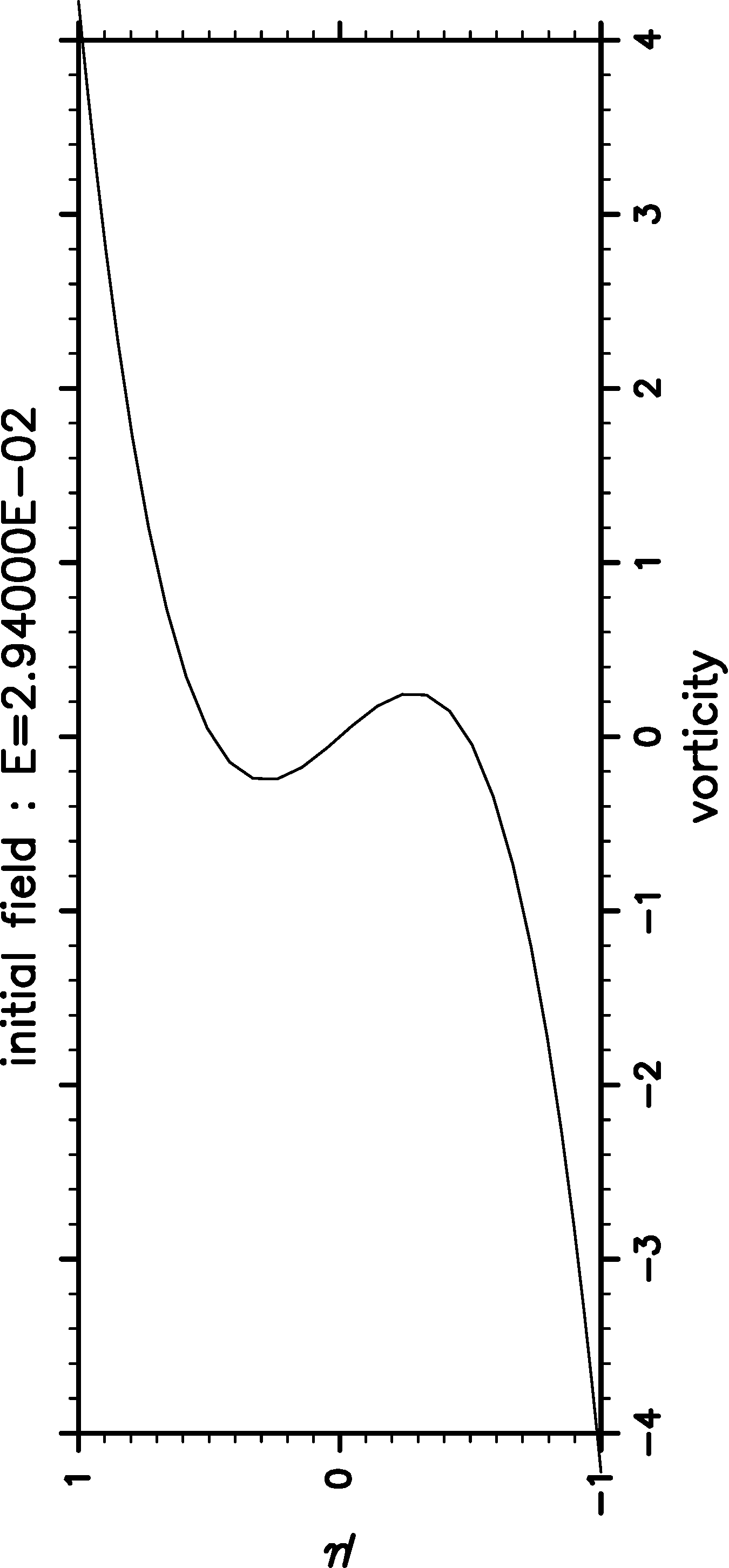}
\caption{The initial zonal vorticity profile defined by \eqref{ini_vor}. The vertical axis is \(\mu\), and the horizontal axis is the vorticity \(q\).}
\label{initial_figure}
\end{figure}
This initial vorticity profile was also used in \cite{ishioka1998}.  By assuming that the sphere is rotating at angular velocity \(\Omega\), the first term of the right-hand side of \eqref{ini_vor} can be considered as the relative vorticity and the second term as the planetary vorticity associated with the rotation. This initial vorticity field not only satisfies the Rayleigh necessary condition for barotropic instability but also consists of the lowest wavenumber zonal modes that can be barotropically unstable \citep{ishioka1990instability}. Indeed, this zonal field is unstable to small disturbances as is shown in Subsection 4.3. We can assume that \(\Omega=1\), which is equivalent to scaling time \(t\) so that the sphere rotates \(1\) radian while \(t\) advances by 1. 

For numerical computing, the continuous initial field defined by \eqref{ini_vor} is discretized into multiple vorticity patches. Let the number of vorticity patches \(K\) be equal to \(J\), the number of grid points in the latitudinal direction. We set the value of each of the vorticity patches \(Q_k\,(k=1,\cdots,K)\) as
\begin{align*}
 Q_k = q_{\rm initial}(0, \mu_k) \quad (k=1,\cdots,K).
\end{align*}
We set the area \(S_k\) of each patch \(Q_k\) as \(S_k = 2\pi Iw_k\), using the Gaussian weight \(w_k\). Here, \(2\pi I w_k\) is the area of the part of the sphere represented by the Gaussian latitude \(\mu_k\). 

\subsection{Time integration}
To compare the MRS statistical equilibria with the end state of the time integration of the vorticity equation, instead of \eqref{R=1eq}, we numerically integrate the following vorticity equation on the unit sphere with the viscosity:
 \begin{align}
 \frac{\partial q}{\partial t}+\left(\frac{\partial \psi}{\partial \lambda}\frac{\partial q}{\partial \mu} - \frac{\partial q}{\partial \lambda}\frac{\partial \psi}{\partial \mu}\right)= \nu (\Delta + 2)q.\label{NSvis}
 \end{align}
 Here, \(q\) is the vorticity, and \(\psi\) is the stream function satisfying \(q=\Delta\psi\). The coefficient \(\nu\) appearing in the right-hand side is the viscosity coefficient. In order for the conservation of angular momentum to hold, we add 2 to the Laplacian in the viscosity term. {By using this form of viscosity, the angular momentum is exactly conserved. The eigenvalues of the Laplacian on (unit) sphere are \(-n(n+1) \,(n=0,1,2,...)\) and for each \(n\) the corresponding eigenspace is spanned by the spherical harmonics \(Y_{m,n} \,(|m|\leq n)\). In the spectral form of the Navier-Stokes equation, the viscosity term of ``Laplacian plus 2'' form is, for wavenumber \((m,n)\), 
\begin{align*}
    \{-n(n+1)+2\}\hat{\zeta}_{m,n}.
\end{align*}
Thus, for \(n=1\), the viscosity terms vanish. As long as we use the spectral method, the spectral coefficients of \(n=1\)  are all conserved except for the rounding error, which corresponds to the conservation of angular momentum. This type of modification of the viscosity term is commonly adopted in numerical models of global circulation of the atmosphere.}

Although the MRS theory is based on the inviscid incompressible flow system in two dimensions, we integrate \eqref{NSvis} with a weak viscosity in place of the inviscid vorticity equation \eqref{R=1eq}. This is because numerical time integration of the inviscid vorticity equation fails as the enstrophy cascades to the scale of the truncation wavenumber and accumulates there. Such compromises have been made in previous studies on the MRS theory \citep[e.g.,][]{sommeria1991final,thess1994inertial,ishioka1998}, and it is considered that the removal of fine variations of the vorticity field by the viscosity is consistent with taking a local average of the microscopic vorticity field in the MRS theory \citep{robert1991statistical,sommeria1991final,ishioka1998}. { We should also note that the viscosity smoothes out the vorticity distribution, which may prevent the vorticity field from reaching the equilibrium that the MRS theory predicts.}
 
 Equation \eqref{NSvis} is integrated from the initial vorticity field defined by \eqref{ini_vor} plus the disturbance \(q_{\rm d}\) defined below:
 \begin{align*}
 q_{\rm d}(\lambda,\mu) = \alpha \left[e^{\gamma (\cos \varphi -1)}-\frac{1-e^{-2\gamma}}{2\gamma}\right].
\end{align*}
Here, \(\varphi\) denotes the angular distance from the center of the disturbance \((\lambda_d,\mu_d)\), which is given by
\begin{align*}
\cos \varphi = \mu_{\rm d} \mu + \sqrt{(1-\mu_d^2)(1-\mu^2)}\cos (\lambda-\lambda_d).
\end{align*}
This disturbance has a Gaussian-like distribution and was used in \cite{ishioka1994non}. We set the magnitude of the disturbance as \(\alpha= 10^{-3}/(2\pi)\), the width parameter of the disturbance as \(\gamma = 100\), and the center as \((\lambda_{\rm d} ,\mu_{\rm d})=(0,\sqrt{2}/2)\). 
 For the time integration, \eqref{NSvis} is discretized using the spectral method with spherical harmonics, the triangular truncation wavenumber of which is \(N=682\). The nonlinear term in \eqref{NSvis} is evaluated via the transform method using \(2048 \times 1024\) grid points on the sphere. The viscosity coefficient is set to \(\nu = 1/[2\pi(N(N+1)-2)]\), which means that the \(e\)-folding time at the truncation wavenumber corresponds to the rotation period of the sphere. The classical 4th-order Runge-Kutta method is used for the time integration with the time step of \(\Delta t = 2\pi/1000\).
 
 In \cite{ishioka1998}, the vorticity equation was numerically integrated with the spherical harmonics spectral method whose truncation wavenumber was \(N=42\); a \(128\times 64\) grid on the sphere was used for the transform method. There, in place of the viscosity term, a hyper-viscosity term in the form of a Laplacian to the 10th power was used, and the \(e\)-folding time at the truncation wavenumber was set to 0.1 rotation period of the sphere. Thus, compared with \cite{ishioka1998}, we conduct a longer-time computation with a much finer resolution, and with the normal viscosity.
 
\subsection{Result of the time integration}

Figure \ref{fig_timeevo} shows the time evolution of the vorticity field.
\begin{figure}[htbp]
\centering
\includegraphics[width=12.0cm]{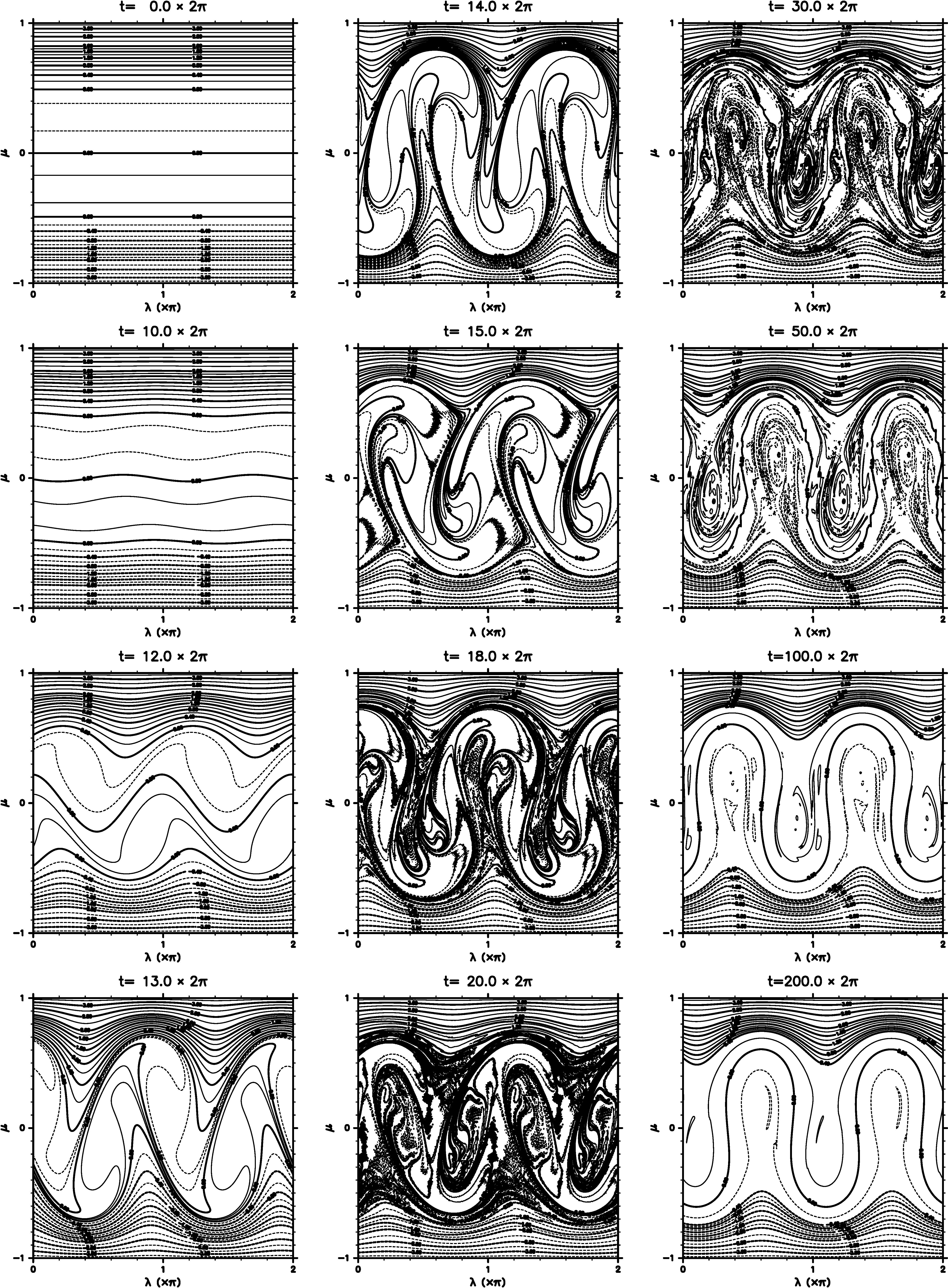}
\caption{Time evolution of the vorticity field. The panels show snapshots at \(t=(0, 12, 13, 14, 15, 18, 20, 30, 50, 100, 200)\times 2\pi\).  The horizontal axis is \(\lambda\); the vertical axis is \(\mu\). The contour interval is set to 0.2, where the value of the vorticity is between -2.0 and 2.0, and 0.4 otherwise.}
\label{fig_timeevo}
\end{figure}
At \(t=0\), the vorticity field is almost zonal since the amplitude of the initial disturbance is small. As the disturbance grows by the barotropic instability, a wavy structure becomes visible at \(t=10\times 2\pi\). The disturbance then grows significantly from \(t=15\times 2\pi\) to \(t=20\times 2\pi\), and the fluid is stirred especially around the equator. At \(t=20\times 2\pi\), the contours of the turbulent vorticity field are filamentary, which means that the vorticity field has very fine structures. As time goes on, these fine structures are smoothed out by the viscosity (\(t=(30, 50, 100)\times 2\pi\)) and eventually, a large-scale wavenumber 2 structure emerges (\(t=200\times 2\pi)\). The wavenumber 2 structure translates westward if seen from the observer rotating with the sphere (in the absolute reference frame, it moves eastward). {Because of the viscosity included in \eqref{NSvis}, the energy is not exactly conserved in the time evolution. At \(t=200\times 2\pi\), however, the energy \(E\) satisfies \(E=0.991694E_0\), where \(E_0\) is the initial value of the energy. That is, the energy is nearly conserved in this case.}

\subsection{Computed statistical equilibria}
We now compute the statistical equilibria that satisfy all the constraints, including energy conservation for the initial vorticity field defined by \eqref{ini_vor}, using the three methods proposed in section \ref{methods}. The following subsections explain details of the application procedures of the three methods and show the obtained statistical equilibria. 

The computations to search the statistical equilibria were done on a PC whose CPU was Intel Xeon W-2135 which had 6 cores of 3.7GHz. When we used ispack-3.0.1, multi-thread computations of 6 threads but without MPI parallelization were performed. Each search computation was completed in several hours.

\subsubsection{The bridge-building method for the canonical problem}
First, we compute the critical point for \(\beta=0\) following the procedure described in section \ref{beta=0}. The resulting vorticity distribution is shown in Figure \ref{bridgezonal_figure}(a). We then compute a zonal critical point for \(\beta= { -132.569}\) (Figure \ref{bridgezonal_figure}(b)), where the value of \(\beta\) is chosen so that energy \(E\) equals \(2.64\times 10^{-2}\). This choice of the value of \(E\) is somewhat arbitrary, but it is chosen here because it is close to \(E_0=2.94\times 10^{-2}\) but not too close. When \(E\) approaches the value \(E_0\) of the initial vorticity field, the numerical computation of zonal critical points becomes difficult, as some components of the macroscopic state \((r_{ijk})\) approach zero, which leads to some of the multipliers overflowing in the computation. Hence, we take as the starting point of the bridge-building method a zonal vorticity field that has somewhat lower energy \(E=2.64\times 10^{-2}\) than the energy of the initial vorticity field. Note that in the steps described above, the considered vorticity distributions are zonally symmetric, and therefore the computational cost can be greatly reduced by assuming that all the \(r_{ijk}\)'s are independent of \(i\). 

Next, the bridge-building method (section \ref{bridgebuilding}) is applied, with the value of \(\beta\) fixed at \({ -132.569}\). To search for an equilibrium whose corresponding vorticity field structure is similar to that which appears in the end state of the time integration, the direction of bridge \((m_0,n_0)\) in the problem (BP) is set to \((2,2)\). As a result, a non-zonal statistical equilibrium for \(\beta = { -132.569}\) is obtained, which is shown in Figure \ref{bridgeM2N2_figure}(a). By gradually decreasing the value of \(\beta\) and computing the equilibrium for each \(\beta\), we reach the statistical equilibrium whose energy value equals to \(E_0\) (Figure \ref{bridgeM2N2_figure}(b)). The structure of the vorticity field corresponding to this statistical equilibrium is very similar to that of the end state of the time integration (the panel of \(t=200\times 2\pi\) in Figure \ref{fig_timeevo}). We can choose another bridge as \((m_0,n_0)=(1,2)\) and compute an equilibrium for \(\beta = { -132.569}\) (Figure \ref{bridgeM1N2_figure}(a)) and an equilibrium whose energy value is equal to \(E_0\) (Figure \ref{bridgeM1N2_figure}(b)). However, the wavenumber 1 type of the vorticity field has no counterpart in the time evolution. Furthermore, \cite{ishioka1998} did not obtain such a wavenumber 1 pattern of statistical equilibrium. 

{There is no guiding principle on how to choose \((m_0, n_0)\), which leaves arbitrariness in the bridge-building method. However, we do not think it is necessary to search for all wavenumbers. With the results of simpler statistical mechanics theories such as \cite{herbert2013additional} in mind, we expect the statistical equilibrium solution to have a large scale spatial distribution, and therefore we believe that it is sufficient to apply the bridge-building method for small \(m_0, n_0\). In fact, we computed several cases of \((m_0,n_0)\) other than the two cases shown above, such as \((m_0,n_0)=(1,3), (1,4), (1,5), (2,3), (2,4), (2,5), (3,3), (3,4)\), and \((3,5)\). In these computations, we could not reach critical points of the free energy \(F_\beta\).}

\begin{figure}[htbp]
 \centering
 \includegraphics[width=10.0cm]{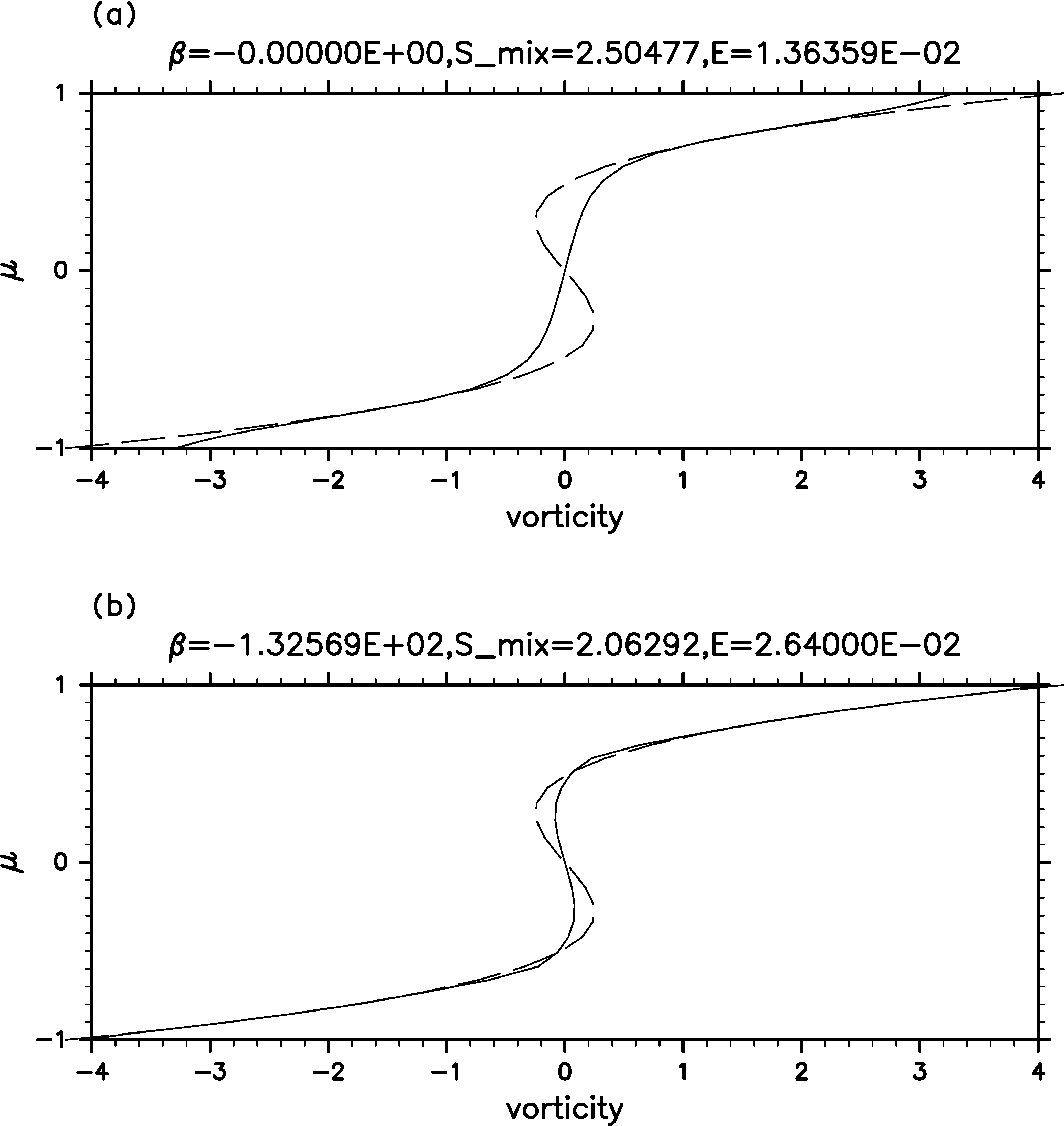}
 \caption{Zonal vorticity profiles corresponding to critical points for the canonical problem with fixed \(\beta\)'s. The vertical axis is \(\mu\); the horizontal axis is the vorticity \(q\). (a) The vorticity profile for \(\beta=0\) (solid line). (b) The vorticity profile for \(\beta={ -132.569}\) (solid line). The corresponding values of \(S_{\rm mix}\) and \(E\) are shown above each panel. In both panels, the broken lines show the initial vorticity profile defined by \eqref{ini_vor}.}
 \label{bridgezonal_figure}
\end{figure}

\begin{figure}[htbp]
 \centering
 \includegraphics[angle=270,width=12.0cm]{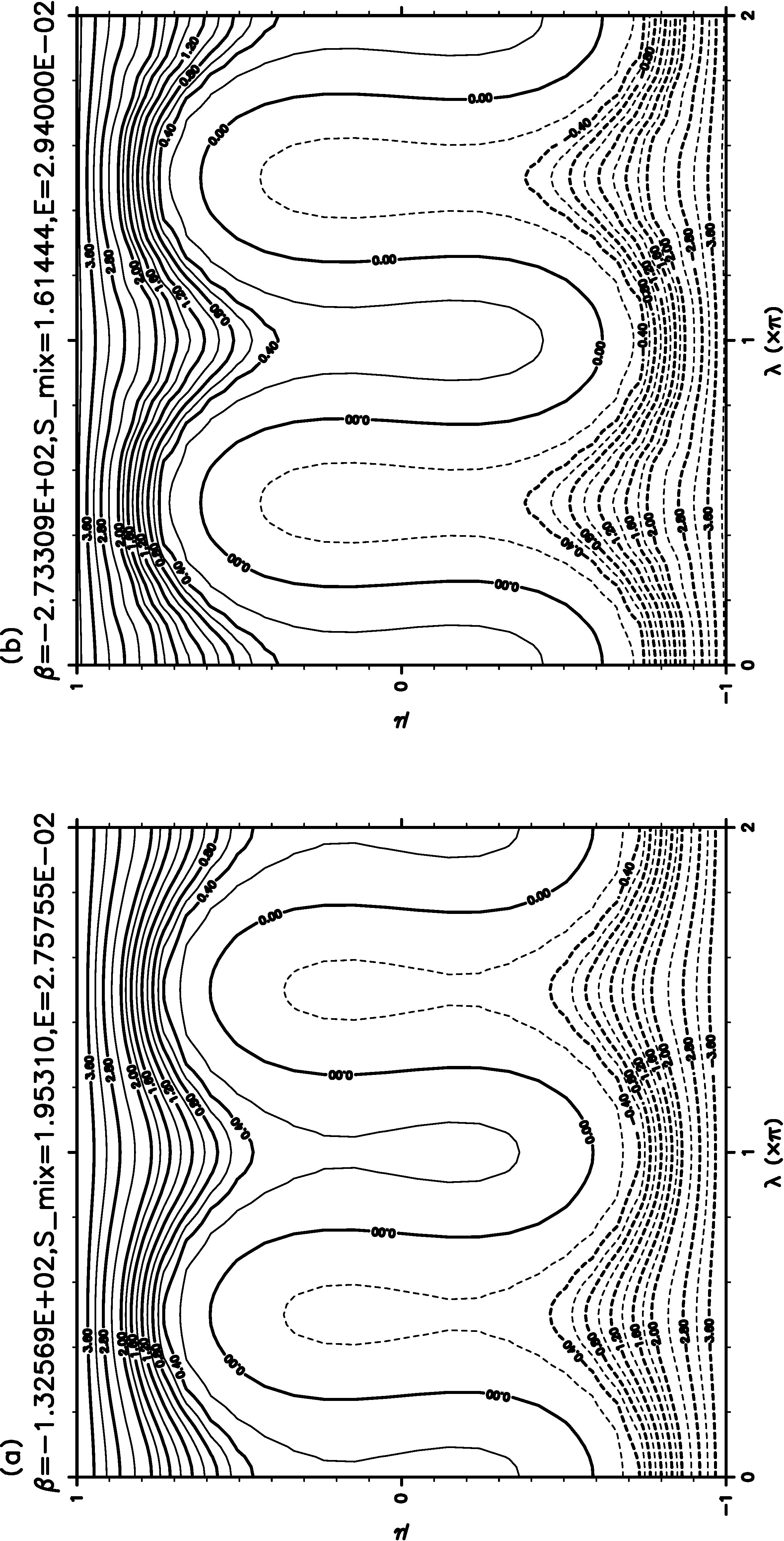}
 \caption{Macroscopic vorticity fields corresponding to solutions of the canonical problem. (a) The solution for \(\beta={ -132.569}\) computed by applying the bridge-building method for \((m_0,n_0)=(2,2)\). (b) The solution obtained by changing the value of \(\beta\) from the value of (a) so that \(E=E_0\) is satisfied. In both figures, the horizontal axis is \(\lambda\), the vertical axis is \(\mu\), and the contour interval is set in the same way as in Figure \ref{fig_timeevo}.}
 \label{bridgeM2N2_figure}   
\end{figure}

\begin{figure}[htbp]
 \centering
 \includegraphics[angle=270,width=12.0cm]{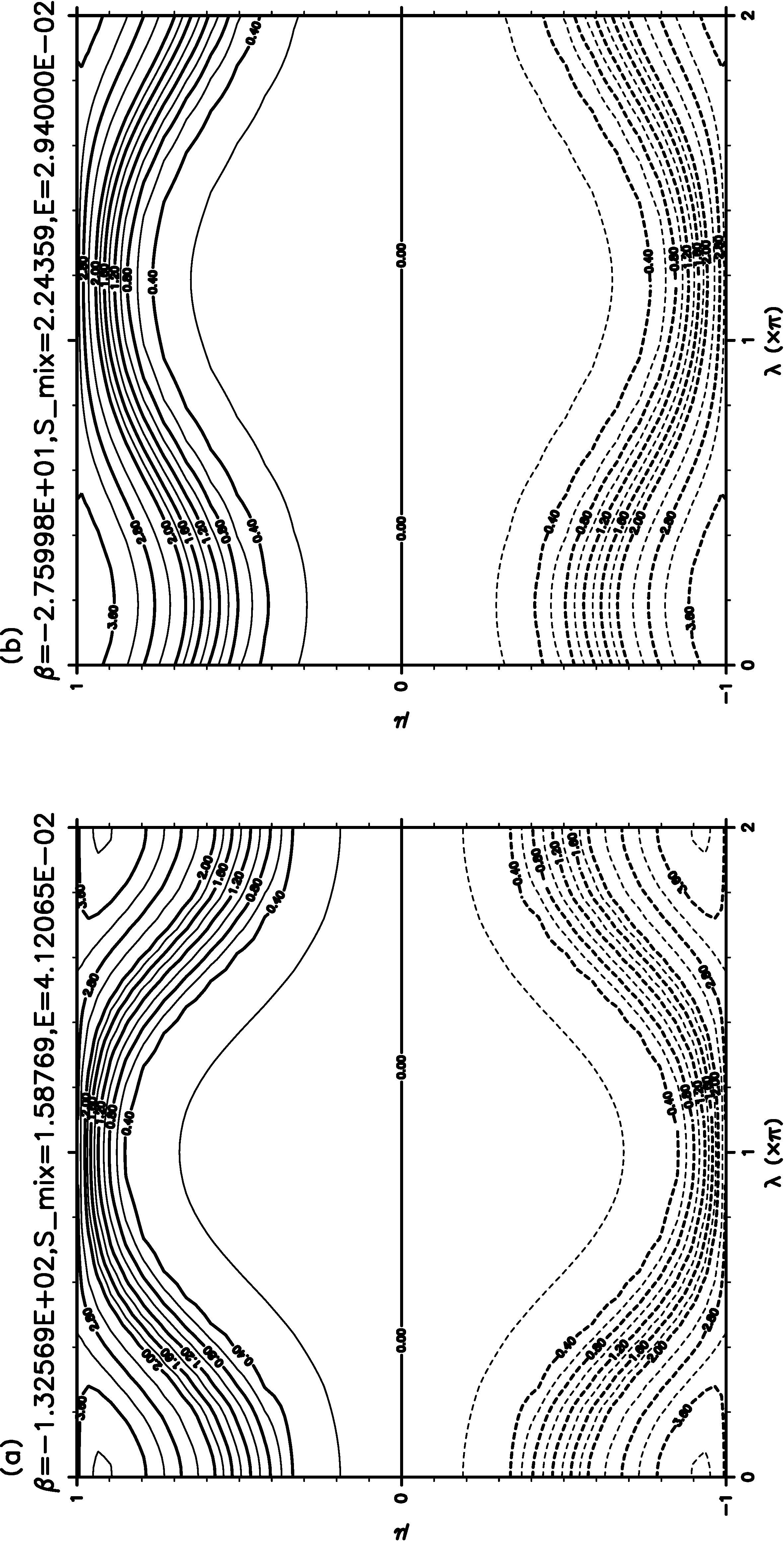}
 \caption{Same as Figure \ref{bridgeM2N2_figure} except that the bridge-building method for \((m_0,n_0)=(1,2)\) is applied. (a) The solution for \(\beta={ -132.569}\). (b) The solution obtained by changing the value of \(\beta\) from the value of (a) so that \(E=E_0\) is satisfied.}
 \label{bridgeM1N2_figure}
\end{figure}

\subsubsection{The gradient method for the canonical problem}
We apply the gradient method described in section \ref{grad_method} by setting the zonal critical point for \(\beta={ -132.569}\) that is obtained in the preparation phase of the bridge-building method (Figure \ref{bridgezonal_figure}(b)) as the starting point for the evolution. For the evolution of the gradient method, the classical fourth-order Runge-Kutta method is used. If we proceed with the calculation by following the method described in section \ref{grad_method}, we reach a wavenumber-1 type solution (not shown) which corresponds to the solution shown as Figure \ref{bridgeM1N2_figure}(b). To find the wavenumber-2 type equilibrium that is relevant to the end state of the time evolution of the vorticity equation, we modify the algorithm of the gradient method by adding frictional terms \(-k\hat{\xi}_{1,2}\) and \(-k\hat{\eta}_{1,2}\) (\(k\) is a positive constant) to the right-hand sides  of equation \eqref{grad_evo} for the \(\hat{\xi}_{1,2}\) and \(\hat{\eta}_{1,2}\) components, respectively.  By adopting this modified algorithm, the wavenumber-2 type solution is obtained (Figure \ref{grad_figure}(a)). The vorticity distribution for this solution is nearly identical to that for the solution computed using the bridge-building method (Figure \ref{bridgeM2N2_figure}(a)) except for the phase shift in the east-west direction. By tracking the branch of solutions starting from the solution for \(\beta={ -132.569}\) (Figure \ref{grad_figure}(a)) and lowering the value of \(\beta\) gradually, we can reach the solution energy value that satisfies \(E=E_0\).  Figure \ref{grad_figure}(b) shows its vorticity distribution; the corresponding \(\beta\) value is \(\beta={-285.194}\). However, the solutions obtained by the modified gradient method for both values of \(\beta\) turn out to be nothing more than saddles of \(F_\beta\), as the solutions transition to the wavenumber-1 type solutions if the frictional terms are removed. Note that there is a slight difference between the outcomes of the bridge-building method and the (modified) gradient method, as seen in the values of \(\beta\) and the values of the mixing entropy \(S_{\rm mix}\) for the equilibrium solutions for \(E=E_0\). That is, as shown in Figure \ref{bridgeM2N2_figure}(b) and Figure \ref{grad_figure}(b), \(\beta={-273.309}\) and \(S_{\rm mix}=1.61444\) for the bridge-building method, while \(\beta={-285.194}\) and \(S_{\rm mix}=1.60753\) for the gradient method. These differences can be attributed to the different approaches of discretization of the constraints between the two methods. 

\begin{figure}[htbp]
 \centering
 \includegraphics[angle=270,width=12.0cm]{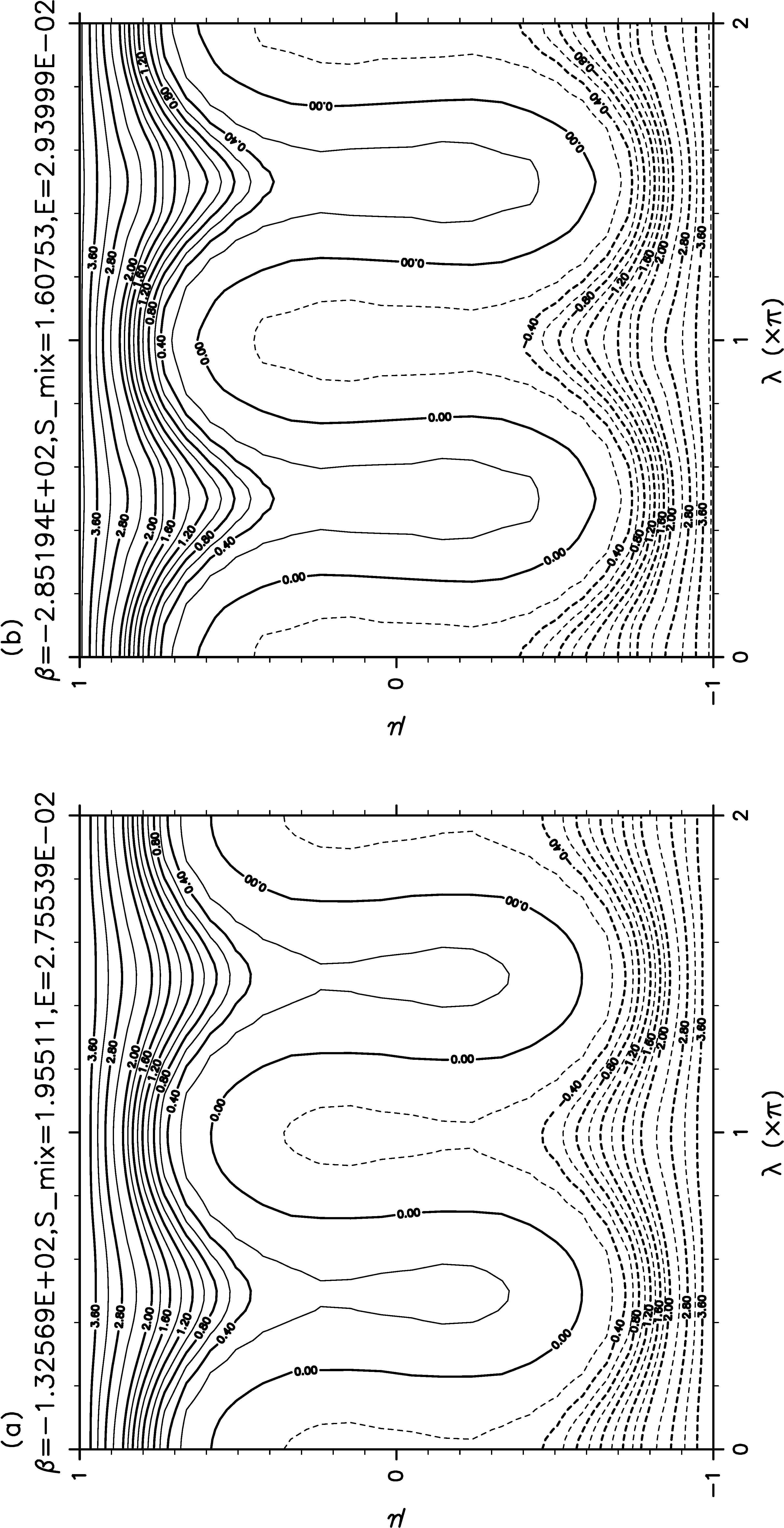}
 \caption{Same as Figure \ref{bridgeM2N2_figure} except that the modified gradient method is applied. (a) The solution for \(\beta={ -132.569}\). (b) The solution obtained by changing the value of \(\beta\) from the value of (a) so that \(E=E_0\) is satisfied.}
 \label{grad_figure}
\end{figure}
 
\subsubsection{The method for the microcanonical problem}
We apply the method described in section \ref{microcanonical} to solve the microcanonical problem for the initial vorticity field \eqref{ini_vor}, using the zonal critical point for \(\beta={ -132.569}\) that is obtained in the preparation phase of the bridge-building method (Figure \ref{bridgezonal_figure}(b)) as the starting point for the search. To search for a wavenumber-2 type equilibrium that is relevant to the end state of the time evolution of the vorticity equation, we maximize the linear function defined as
\begin{align}
f(Z)=\hat{\xi}_{2,2}
\label{fZ2}
\end{align}
on the surface \(L_M\), with \(M\) set as the value of \(V(Z)\) at the zonal critical point. By solving this maximization problem, we reach a point the corresponding vorticity field of which has a wave-like structure of wavenumber 2 (Figure \ref{microM2N2_figure}(a)). 
Next, we follow Step 3 in section \ref{microcanonical} to increase energy \(E\) to \(E_0\) gradually, enabling us to then reach a point on the energy hypersurface of \(E=E_0\) that can be used as a starting point for the Step 4 search (Figure \ref{microM2N2_figure}(b)). Finally, by maximizing the mixing entropy \(S_{\rm mix}\) on the energy hypersurface, a statistical equilibrium of wavenumber-2 type structure is obtained (Figure \ref{microM2N2_figure}(c)). The vorticity field for this equilibrium is nearly identical to that of the equilibria obtained by the other two methods, the bridge-building method and the gradient method, for the canonical problem (Figure \ref{bridgeM2N2_figure}(b) and Figure \ref{grad_figure}(b), respectively). We should note that if we were to adopt the gradient projection method naively in the last phase of the maximization process, the constraint \(E=E_0\) would be violated slightly because of the curvature of the energy hypersurface; in addition, the convergence of the maximization would become slow. To avoid these drawbacks, we introduce a local coordinate system on the energy hypersurface so that the energy constraint is satisfied exactly. This local coordinate system consists of the coefficients \((Z)\) except for \(\hat{\xi}_{0,3}\).  That is, we take the coefficient \(\hat{\xi}_{0,3}\) as a function of the coordinate variables, which is defined as follows:
\begin{align*}
\hat{\xi}_{0,3} = \sqrt{24} \sqrt{E_0-\sum_{n\neq 3}\frac{\hat{\xi}_{0,n}^2}{2n(n+1)} -\sum_{|m|\leq n}\frac{\hat{\xi}_{m,n}^2+\hat{\eta}_{m,n}^2}{n(n+1)}},
\end{align*}
{in order for the constraint \eqref{energy_const} to be satisfied automatically.}
The use of this local coordinate system is valid because \(\hat{\xi}_{0,3}> 0\) holds in a neighborhood of the initial point for searching. Using this local coordinate system, we are able to reach a point \(Z\) which is sufficiently close to a critical point of \(S_{\rm mix}\) by advancing in the direction of the gradient of \(S_{\rm mix}\). We can then accelerate the convergence by using Newton's method.

A wavenumber-1 type of equilibrium solution is also obtained (Figure \ref{microM1N2_figure}) if, instead of setting \(f(Z)\) as \eqref{fZ2}, we set \(f(Z)=\hat{\xi}_{1,2}\).
To examine whether the obtained equilibrium solutions for the microcanonical problem are exactly local maxima of \(S_{\rm mix}\), we compute the eigenvalues of the Hessian matrices of \(S_{\rm mix}\) on the local coordinate system at these solution points.  All the eigenvalues for the wavenumber-1 type solution are negative, which means the solution point is a local maximum of \(S_{\rm mix}\). On the other hand, two of the eigenvalues for the wavenumber-2 type solution are positive, and the others are negative. Thus, the wavenumber-2 type solution is not a local maximum but a saddle. When the point is slightly pushed from the wavenumber-2 type solution in the direction of the eigenvector corresponding to either of the positive eigenvalues, the point moves toward the wavenumber-1 type equilibrium solution. 
\begin{figure}[htbp]
 \centering
 \includegraphics[angle=270,width=14.0cm]{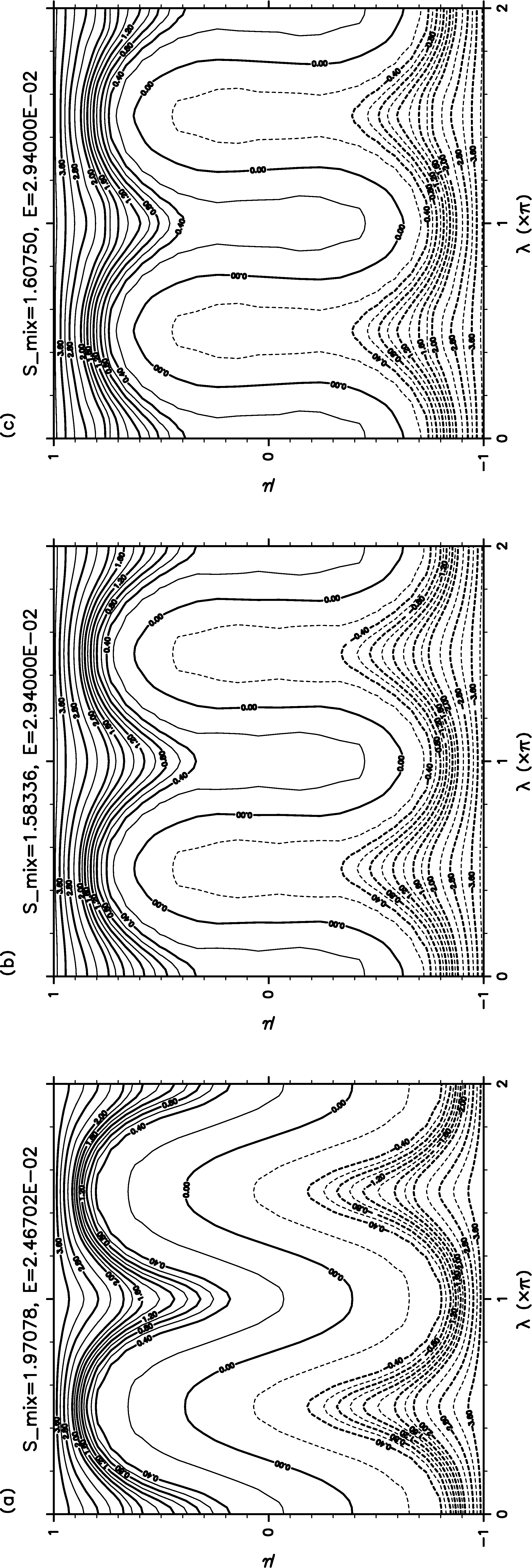}
 \caption{Macroscopic vorticity fields corresponding to the solutions of the canonical problem. (a)The solution of the maximization problem of the linear function \(f(Z)=\hat{\xi}_{2,2}\) on \(L_M\), starting from the point of Figure \ref{bridgezonal_figure}(c). (b) The point on the energy-hypersurface of \(E=E_0\) reached by following Step 3 in section \ref{microcanonical}. (c) The equilibrium solution obtained by the maximization of the mixing entropy on the energy-hypersurface of \(E=E_0\). In each panel, the horizontal axis is \(\lambda\), the vertical axis is \(\mu\), and the contour interval is set in the same way as in Figure \ref{fig_timeevo}.} 
 \label{microM2N2_figure}
 \end{figure}

 \begin{figure}[htbp]
 \centering
 \includegraphics[angle=270,width=14.0cm]{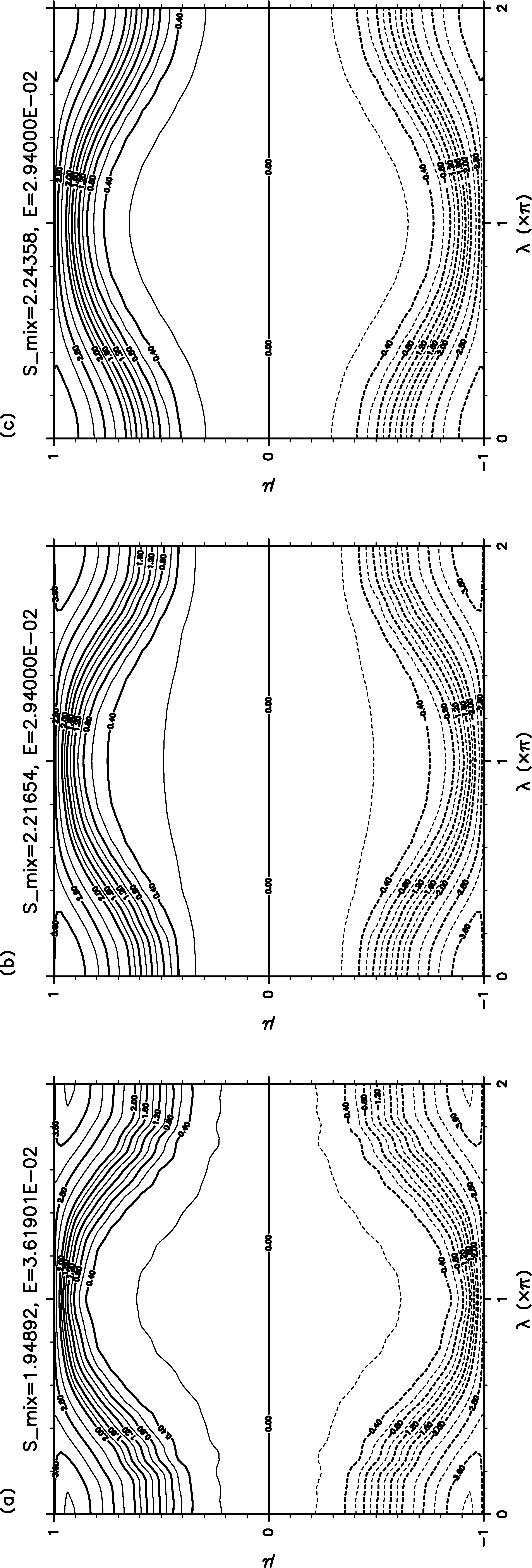}
 \caption{Same as Figure \ref{microM2N2_figure} except that the linear function is set as \(f(Z)=\hat{\xi}_{1,2}\).}
 \label{microM1N2_figure}
 \end{figure}
 
 \subsection{Relation between vorticity and stream function (\(\overline{q}\)-\(\overline{\psi}\) relation)}
\label{qpsirel_section}
For a critical point of the mixing entropy with a non-zero inverse temperature \(\beta\), there exist constants \(\Omega_{1},\Omega_2,\Omega_3\) and a function \(f\) such that the relation
\begin{align}
    \overline{q}(\lambda,\mu)=f(\overline{\psi}(\lambda,\mu)+\Omega_1\mu+\Omega_{2}\sqrt{1-\mu^2}\cos \lambda + \Omega_3 \sqrt{1-\mu^2}\sin \lambda)\label{qpsirel_theory}
\end{align}
holds between the values of the macroscopic vorticity field \(\overline{q}\) and the macroscopic stream function \(\overline{\psi}\). For the initial vorticity field \eqref{ini_vor} considered here, two components \(M_2^{\rm ini}\) and \(M_3^{\rm ini}\) of the angular momentum are zero, while \(M_1^{\rm ini}\neq 0\). This yields \(\Omega_{2}=\Omega_3=0\) (For the relations between the values of \(\Omega_{1},\Omega_2,\Omega_3\) and the angular momentum, see the Supplement.). That is, the macroscopic vorticity field \(\overline{q}\) corresponding to a critical point of the mixing entropy is a stationary solution of the Euler equation when viewed by an observer rotating about the axis passing through the north and south poles at an angular velocity \(\Omega_1\). In other words, the macroscopic vorticity field \(\overline{q}\) is rotating like a rigid body at angular velocity \(\Omega_1\). For the computed solutions of the microcanonical problem (Figures \ref{microM2N2_figure}(b) and \ref{microM1N2_figure}(c)), we can draw a scatter plot of the macroscopic vorticity \(\overline{q}\) versus the modified stream function \(\overline{\psi}+\Omega_1 \mu\), where the value of \(\Omega_1\) is determined by the method described below. If the functional relationship \eqref{qpsirel_theory} with \(\Omega_2=\Omega_3=0\) holds, the Jacobian 
\begin{align*}
\frac{\partial(\overline{\psi}+\Omega_1\mu)}{\partial \lambda}\frac{\partial \overline{q}}{\partial \mu}-\frac{\partial \overline{q}}{\partial \lambda}\frac{\partial(\overline{\psi}+\Omega_1\mu)}{\partial \mu}
\end{align*}
should vanish. Therefore, 
\begin{align}
\frac{\partial \overline{\psi}}{\partial \lambda}\frac{\partial \overline{q}}{\partial \mu}-\frac{\partial \overline{q}}{\partial \lambda}\frac{\partial \overline{\psi}}{\partial \mu} =  \Omega_1 \frac{\partial \overline{q}}{\partial \lambda} \label{angularv_det}
\end{align}
should hold at each point \((\lambda,\mu)\). The value of \(\Omega_1\) is determined so that a norm of the error of \eqref{angularv_det} is minimized. The norm used here is defined for a function \(F\) with zero mean by
\begin{align*}
||F|| := \left(-\int_{-1}^1\int_{0}^{2\pi} F(\lambda,\mu)G(\lambda,\mu) {\rm d}\lambda {\rm d}\mu \right)^{\frac{1}{2}},
\end{align*}
where \(G\) is a function such that \(\Delta G=F\). The scatter plots for Figures \ref{microM2N2_figure}(b) and \ref{microM1N2_figure}(c) are shown in Figures {\ref{qpsirel_ES_figure}(a)} and {\ref{qpsirel_ES_figure}(b)}, respectively. For comparison, we also determine the value of \(\Omega_1\) and draw a scatter plot of the vorticity \(q\) versus the modified stream function \(\psi+\Omega_1 \mu\) for the end state of the time integration (panel for \(t=200\times 2\pi\) in Figure \ref{fig_timeevo}). While the statistical equilibria defined in the framework of energy-enstrophy theory have a linear functional relationship (i.e. a linear \(f\)) in \eqref{qpsirel_theory} \citep{herbert2012statistical}, the scatter plots shown in Figure \ref{qpsirel_TE_figure} have clearly nonlinear \(\overline{q}\)-\(\overline{\psi}\) relations. The functional form of the \(\overline{q}\)-\(\overline{\psi}\) relation for the wavenumber-2 type solution ({Figure \ref{qpsirel_ES_figure}(a)}) is very similar that of \(q\)-\(\psi\) relation for the end state of the time integration (Figure \ref{qpsirel_TE_figure}), especially in the part of \(|\overline{q}|>0.5\) (\(|q|>0.5\)). Conversely, in the part where the absolute values of \(q\) and \(\overline{q}\) are small, there is a difference between them. The region of small values of \(|q|\) is located at low latitudes on the sphere, where strong mixing occurs in the time evolution. Hence, the small but finite viscosity used in the time integration may be the cause of the difference. However, it is also possible that the resolution of the discretization used in the computation of the statistical equilibria may still be insufficient to represent detailed structures in low latitudes.

\begin{figure}[htbp]
\centering
\includegraphics[angle=270,width=14.0cm]{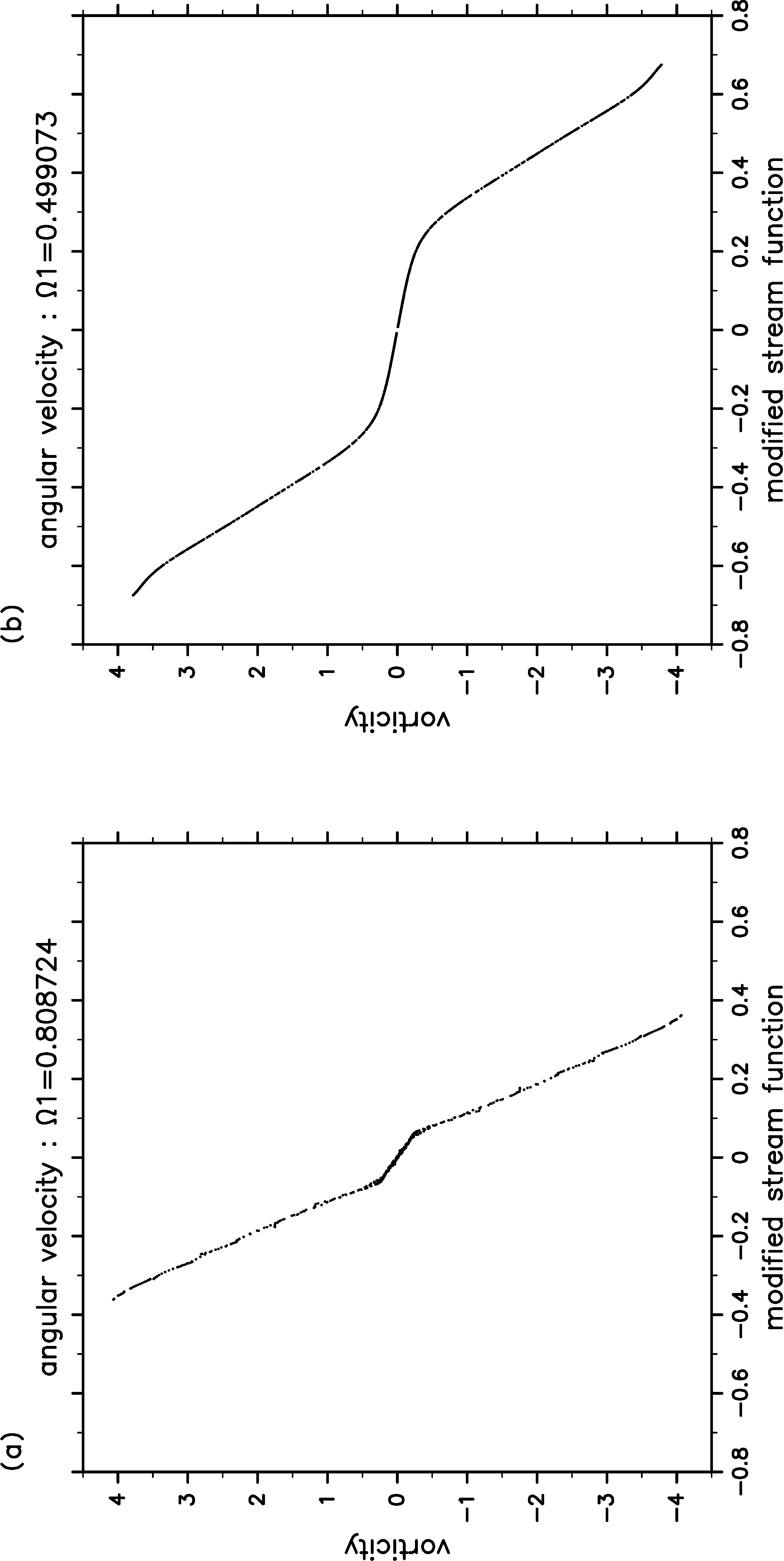}
\caption{Scatter plot of the macroscopic vorticity \(\overline{q}\) versus the modified stream function \(\overline{\psi}+\Omega_1 \mu\). (a) For the wavenumber-2 type solution (Figure \ref{microM2N2_figure}(c)). (b) For the wavenumber-1 type solution (Figure \ref{microM1N2_figure}(c)). In each figure, the horizontal axis is the modified stream function \(\overline{\psi}+\Omega_1 \mu\), and the vertical axis is \(\overline{q}\). The value of the determined angular velocity \(\Omega_1\) is shown above each plot. The scatter plots are drawn using the values at \(32\times 64=2048\) grid points in each figure.}
\label{qpsirel_ES_figure}
\end{figure}

\begin{figure}[htbp]
\centering
\includegraphics[angle=270,width=7.0cm]{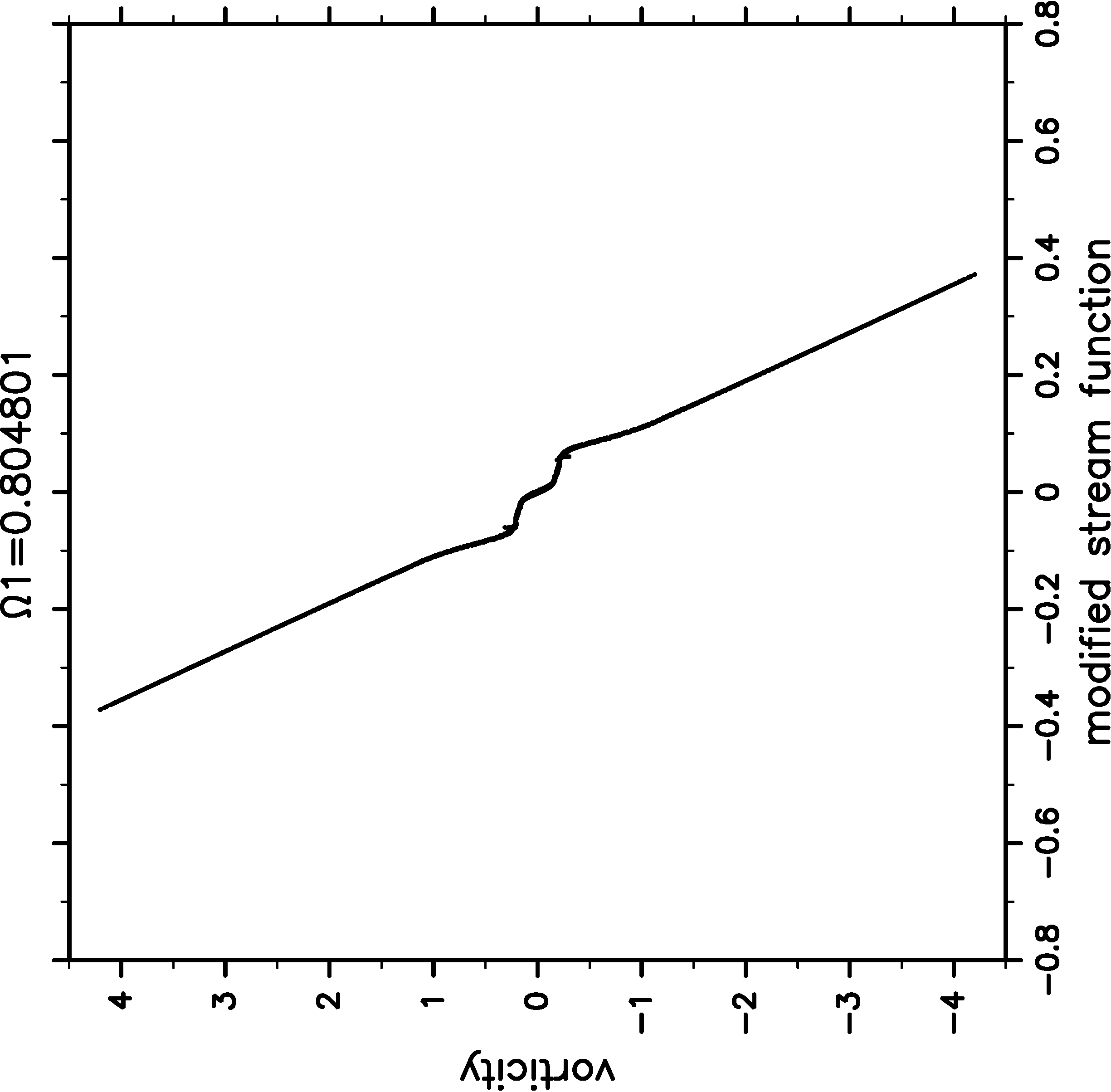}
\caption{Same as Figure \ref{qpsirel_ES_figure} except that the scatter plot is drawn for the vorticity \(q\) versus the modified stream function \(\psi +\Omega_1 \mu\) for the end state in the time integration (panel for \(t=200\times 2\pi\) in Figure \ref{fig_timeevo}). The scatter plots are drawn using the values at \(64\times 128=8192\) grid points.}
\label{qpsirel_TE_figure}
\end{figure}

\section{Discussion}\label{discussion}
We have proposed here three new methods for computing statistical equilibria on the basis of the MRS theory. The new methods have a number of advantages over the methods proposed in previous studies. One advantage is that they are applicable to computations of non-zonal statistical equilibria. 
In most of the previous methods other than the gradient method proposed by \cite{ishioka1998} and the continuation method proposed by \cite{thess1994inertial}, the path of the solution search has been strongly dependent on the initially given vorticity field, and it was difficult to control the path as pointed out by \cite{bouchet2012statistical}. 
That is, if the initially given vorticity field was zonal, the computed statistical equilibrium also tended to be zonal, and it often became just a saddle of \(F_\beta\) or \(S_{\rm mix}\).  One might think that given a non-zonal starting point for the search one could overcome this difficulty, but that in itself is quite a difficult problem for general initial vorticity distributions. In contrast, the proposed new methods make it possible to search for statistical equilibria with broken zonal symmetry. In particular, the gradient method for the canonical problem can automatically lead to non-zonal solutions.  In the bridge-building method and the method for the microcanonical problem, however, one needs to choose the direction of the expansion coefficients to search; i.e., one needs to choose the bridge-direction \((m_0,n_0)\) in (BP) or the linear function \(f(Z)\) in the method for the microcanonical problem. One might think that the need to make such a choice can lead to difficulties since it is necessary to try a wide range of search directions if the approximate distribution of the vorticity field corresponding to the relevant statistical equilibrium is not known in advance.  However, because only low-wavenumber solutions are relevant in actual applications of the MRS theory, it would be sufficient to check just a few search directions. 

Another advantage of the new methods over those in previous studies, such as \cite{robert1992relaxation} and \cite{ishioka1998}, is that they are computationally efficient even when the number of vorticity patches increases. 
Here, the method proposed by \cite{robert1992relaxation}, which is a relaxation method, requires that a system of differential equations consisting of the same number of equations as the number of patches be solved. This makes efficient computation difficult when the number of vorticity patches increases. Similar difficulty also arises in the method proposed by \cite{ishioka1998}, which is a gradient method for the microcanonical problem. There, the number of variables has the order of \(N^3\) (\(N\): the truncation wavenumber) and the applicability of the method is limited to low-resolution cases. In contrast, the gradient method for the microcanonical problem proposed in the present manuscript, which also solves the microcanonical problem, realizes more efficient computation by reducing the number of variables for which gradients are computed to \(O(N^2)\) through the subproblem (SP). 
Among the previous studies, \cite{prieto2001analytical} used a modified version of the iterative method proposed by \cite{turkington1996statistical} to find the equilibrium solution for the case of a large number of vortex patches, but there, only the solution with zonal symmetry is calculated and they do not seem to find a non-zonal solution.
The other advantage of the new methods is that the gradient method for the microcanonical problem makes it possible to determine whether a critical point is a saddle or a local extremum by computing the eigenvalues of the Hessian matrix of \(S_{\rm mix}\) as described in section \ref{example}. 
There is no guarantee, however, that the computed equilibrium is exactly the global maximum of \(F_\beta\) or \(S_{\rm mix}\) even if it is confirmed to be a local extremum. This difficulty comes from the nonlinear nature of the problem, and more mathematically sophisticated treatment of the maximization problem is needed. 

As mentioned above, our proposed new methods allow us to obtain non-zonal equilibrium solutions with higher resolution than in previous studies. In fact, in section \ref{example}, we computed the equilibrium solution for the same barotropically unstable zonal initial field considered in \cite{ishioka1998}, but with higher resolution. In both \cite{ishioka1998} and the present manuscript, the statistical equilibria with a wave-like structure of zonal wavenumber 2, which seem relevant to the end state of the time integration, have been obtained, but the equilibrium solution of \cite{ishioka1998}, which was computed with the resolution of a \(32\times 16\) grid and the truncation wavenumber \(N=10\), had a wavy structure at high latitudes with a smaller amplitude than that of the wavy structure appearing in the end state of the integration and the equilibrium solution obtained in section \ref{example} in the present manuscript \(N=21\) results. We can conclude that this discrepancy is due to the low resolution used in \cite{ishioka1998}.

For the wavenumber-2 type of equilibrium solution mentioned above, we have shown that the wavenumber-2 solution is not a local maximum of both \(F_\beta\) for the canonical problem and \(S_{\rm mix}\) for the microcanonical problem. For the canonical (microcanonical) problem, there exists a wavenumber-1 type equilibrium with a larger value of \(F_\beta\) (\(S_{\rm mix}\)). Therefore, the wavenumber-2 type solution is not a local maximum, but a saddle for both the canonical problem and the microcanonical problem. The reason that the wavenumber-2 type solution was obtained even using the gradient method (or the gradient projection method) is that this point is stable to perturbations in all but two directions. Meanwhile, in the framework of the linear stability analysis \citep{ishioka1990instability}, the initial zonal vorticity field has been proved to be unstable only to perturbations of zonal wavenumber 2. Hence, there is no room for a disturbance of wavenumber 1 to grow spontaneously from this initial vorticity field. {The reason we introduced the frictional terms in the gradient method in section \ref{example} is to suppress the growth of the wavenumber-1 perturbation. This suppression mimics the fact that the wavenumber-1 perturbation cannot grow spontaneously in the time evolution of the vorticity equation.} The maximum entropy principle by \cite{robert1991maximum} and \cite{robert1991statistical} states that the vast majority of microscopic states concentrate on the macroscopic state that maximizes the mixing entropy. In the case that we have considered in the present manuscript, however, the vorticity field did not evolve toward the state of the wavenumber-1 type equilibrium and settled for the ``second best'' wavenumber-2 type equilibrium. This example shows that saddles of the mixing entropy can have some physical meaning if, for some reason, the maximum of the mixing entropy cannot be reached.

In the present manuscript, we have computed the statistical equilibrium solutions according to the original MRS theory on the spherical geometry, while \cite{herbert2013additional} analytically studied statistical equilibrium solutions on the basis of the energy-enstrophy measure on spherical geometry.
Our finding in section \ref{example} that the wavenumber-1 type and wavenumber-2 type of solutions were obtained from the initial vorticity field \eqref{ini_vor} is consistent with the result of \cite{herbert2013additional}. In Herbert's solutions, the energy of a macroscopic vorticity field in equilibria condensed to the modes with wavenumbers of \(n\leq 2\). In the wavenumber-1 type equilibrium obtained in the present manuscript, the expansion coefficient of \((m,n)=(1,2)\) has a large amplitude, while in the wavenumber-2 type, the expansion coefficient of \((m,n)=(2,2)\) has a large amplitude. In the framework of energy-enstrophy theory, however, the initial vorticity distribution is not conserved, whereas it is conserved in the original MRS theory. 
This difference leads to a difference between the statistical equilibria predicted by the two theories.
Indeed, the statistical equilibria obtained in section \ref{example} have nonlinear \(\overline{q}\)-\(\overline{\psi}\) relations, which are different from the linear \(\overline{q}\)-\(\overline{\psi}\) relations predicted by energy-enstrophy theory. 
As shown in section \ref{example}, the \({q}\)-\({\psi}\) relation calculated for the end state of the time integration is clearly nonlinear, and we believe that the equilibrium solution obtained according to the original MRS theory has more relevance.

{Note that the methods proposed in this manuscript can be extended to other geometries than sphere (e.g., two-dimensional torus or rectangular domains of \(\mathbf{R}^2\)) with replacing the spherical harmonics by the eigenfunctions of Laplace-Beltrami operator on each geometry. Furthermore, the methods can be extended to 1.5-layer quasi-geostrophic flow system, since the flow is incompressible and the stream function is related to the potential vorticity by an elliptic equation there.}

The mathematical setup used for the microcanonical problem in the present manuscript can be applied to other problems in which the conservation of both energy and all Casimir invariants is imposed, such as the problem proposed by \cite{shepherd1988rigorous} of evaluating the upper bound for the maximal growth of disturbances from barotropically unstable zonal flows. In \cite{ishioka1996numerical}, a numerical method for the upper-bound problem was proposed in which the mathematical formulation was similar to that in the present manuscript. Furthermore, \cite{ishioka2013proof} proposed a more efficient algorithm to compute the upper bound. This upper bound was, however, calculated without imposing the energy constraint, and thus the obtained upper bound may be looser than the best. Numerical methods for computing the upper bound imposing the energy constraint might be constructed by using the way of exploring the feasible region of the macroscopic vorticity field that has been proposed for solving the microcanonical problem in the present manuscript.

\section{Conclusions}\label{conclusion}
We have proposed three new methods for numerically solving the maximization problem of mixing entropy by which the statistical equilibrium states based on the Miller-Robert-Sommeria (MRS) theory for two-dimensional turbulence are obtained. For the canonical problem, the bridge-building method (section \ref{newton_method}), which is suitable for searching multiple equilibria, and the gradient method (section \ref{grad_method}), which can automatically detect non-zonal equilibria, were proposed. By using the bridge-building method, we were able to find a new branch of non-zonal solutions from one already computed, such as the branch of zonal solutions. In the gradient method, the computational cost was reduced by introducing a subproblem and taking the expansion coefficients of the macroscopic vorticity field as the only independent variables. For the microcanonical problem, a method that makes use of the geometry of the feasible region of the expansion coefficients was proposed. In this method, a variety of initial starting points for the search can be generated by maximizing linear functions on a hypersurface that is homeomorphic to a (higher dimensional) sphere. All of the proposed methods made the computation of multiple equilibrium solutions easier than those offered in previous studies. 

To demonstrate the effectiveness of the proposed methods, we computed the statistical equilibria for the initial vorticity field in \eqref{ini_vor}, which was zonal and barotropically unstable (Figure \ref{initial_figure}). As a result, a wavenumber-2 type solution (Figures \ref{bridgeM2N2_figure}(b), \ref{grad_figure}(b), and \ref{microM2N2_figure}(c)) as well as a wavenumber-1 type solution (Figures \ref{bridgeM1N2_figure}(b) and \ref{microM1N2_figure}(c)) were obtained. 
The structure of the vorticity field corresponding to the wavenumber-2 type statistical equilibrium was very similar to that of the end state of the time integration of the vorticity equation (Figure \ref{fig_timeevo}). It was also shown, however, that the wavenumber-1 solution was exactly a local maximum of the mixing entropy, but the wavenumber-2 solution was just a saddle of the mixing entropy. We believe that the reason that the wavenumber-2 type solution was relevant to the end state of the time integration is that the initial zonal vorticity profile was linearly unstable only for wavenumber 2 disturbance. 

The methods proposed in the present manuscript enable us to compute statistical equilibria for initial vorticity fields which consist of many levels of vorticity patches. Hence, the methods are expected to be applicable to a wide variety of geophysical flows since these methods are able to compute statistical equilibria without losing information regarding the initial vorticity fields (i.e., conservation of all Casimir invariants is satisfied in the microscopic sense). Furthermore, the method that we have proposed for the microcanonical problem will provide a new direction for dealing with other problems in which the conservation of all Casimir invariants and the conservation of energy are both required. 

\section*{Acknowledgments}
{We thank two anonymous reviewers for their helpful comments.}
This work was supported by JSPS KAKENHI Grant Numbers 20K04061.
This work was also supported by MEXT as
“Program for Promoting Researches on the Supercomputer Fugaku”
(Toward a unified view of the universe: from large scale structures
to planets, {JPMXP1020200109}) and
``Exploratory Challenge on Post-K computer''
(Elucidation of the Birth of Exoplanets [Second Earth] and the
Environmental Variations of Planets in the Solar System).

{This work used computational resources of the K computer and supercomputer
Fugaku provided by the RIKEN Center for Computational Science
through the HPCI System Research Project
(Project ID: hp160254, hp170225, hp180199, hp190170, hp200124, hp210164).}

The GFD-DENNOU Library
(http://www.gfd-dennou.org/arch/dcl/) was used to draw the figures.

\appendix
\section{Associated Legendre functions and spherical harmonics}
The Legendre polynomials \(P_{n}(\mu),\,n=0,1,2,\cdots\) are defined as
\begin{align*}
 P_n(\mu)= \frac{\sqrt{2n+1}}{2^n n!}\frac{{\rm d}^n}{{\rm d}\mu^n}(\mu^2-1)^n.
\end{align*}
For integers \((m,n)\) with \(|m|\leq n\), the associated Legendre function \(P_{m,n}(\mu)\) is defined as
\begin{align*}
 P_{m,n}(\mu) = \sqrt{\frac{(n-|m|)!}{(n+|m|)!}}(1-\mu^2)^{\frac{|m|}{2}}\frac{{\rm d}^{|m|}}{{\rm d}\mu^{|m|}}P_n(\mu).
\end{align*}
The associated Legendre functions have an orthogonal property for each \(m\):
\begin{align*}
 \int_{-1}^1 P_{m,n}(\mu)P_{m,l}(\mu) {\rm d}\mu = 2 \delta_{nl}.
\end{align*}
Here, \(\delta_{nl}\) is Kronecker delta.
Let \((\lambda,\mu)\) be the longitude-latitude coordinates on the sphere. For integers \((m,n)\) with \(|m|\leq n\), we define complex-valued continuous functions on the sphere \(Y_{m,n}(\lambda,\mu)\) by
\begin{align*}
 Y_{m,n}(\lambda,\mu)=P_{m,n}(\mu) e^{\sqrt{-1}m\lambda},
\end{align*}
and we call them spherical harmonics. The family of spherical harmonics \(\{Y_{m,n}\}_{|m|\leq n}\) consists of an orthogonal system about the \(L^2\)-inner product on the sphere, i.e.,
\begin{align*}
 \frac{1}{4\pi}\int_S Y_{m,n}(\lambda,\mu)Y_{p,q}^*(\lambda,\mu) {\rm d}S = \delta_{mp}\delta_{nq}
\end{align*}
Furthermore, the system is a basis of the \(L^2\)-space on the sphere, satisfying completeness \citep[e.g.][]{stein2016introduction}.

\section{Gaussian weights and approximating integration}
Each Legendre polynomial \(P_n(\mu)\) has \(n\) zeros in the interval \((-1,1)\), since the \(P_n(\mu)\)'s are orthogonal polynomials. For an integer \(J\), let \(\mu_1,\cdots,\mu_J\) be the sequence of the zeros of \(P_J(\mu)\), with \(\mu_1<\cdots<\mu_J\). The \(\mu_j\)'s are called the Gaussian nodes. The Gaussian weights \(\tilde{w}_j\,(j=1,\cdots,J)\) are defined as
\begin{align*}
 \tilde{w}_j = \frac{2(2J+1)}{(1-\mu_j^2)P_{J}'(\mu_j)^2}.
\end{align*}
By using the Gaussian nodes and the Gaussian weights, the integral of a function \(g(\mu)\) on the interval \([-1,1]\) is approximated by the Gauss-Legendre quadrature in the following way:
\begin{align}
 \int_{-1}^1 g(\mu) {\rm d}\mu \approx \sum_{j=1}^J \tilde{w}_j g(\mu_j)\label{GLint}.
\end{align}
The left- and right-hand sides of \eqref{GLint} are rigorously identical if \(g(\mu)\) is a polynomial in \(\mu\) of degree less than \(2J\) \citep[see e.g.,][]{durran2010numerical}. In particular,
\begin{align*}
 \tilde{w}_1+\cdots+\tilde{w}_J =2.
\end{align*}
When we consider approximating integrals of a function over the sphere, we discretize the sphere as we have done in section \ref{methods}, i.e., taking \(I\times J\) grid points \((\lambda_i,\mu_j)\) on the sphere and the integral of a function \(g(\lambda,\mu)\) over the sphere is approximated as
\begin{align*}
\frac{1}{4\pi}\int_S g(\lambda,\mu) {\rm d} S \approx \frac{1}{2I}\sum_{i=1}^I \sum_{j=1}^J \tilde{w}_j g(\lambda_i,\mu_j).
\end{align*}
By introducing the normalized Gaussian weight \(w_j=\tilde{w}_j/I\), the approximation formula
\begin{align}
 \frac{1}{4\pi}\int_S g(\lambda,\mu) {\rm d} S \approx \frac{1}{2}\sum_{i=1}^I \sum_{j=1}^J w_j g(\lambda_i,\mu_j)\label{int_sphere}
\end{align}
is obtained.
\bibliographystyle{jphysicsB}
\bibliography{reflist_abb}

%% file: Ryono_Ishioka_supplement_R2.tex
\section{Geometry of the feasible region}
\label{s-Shape_domain}
In the main part of the manuscript, to construct the gradient methods for the canonical and microcanonical problem, we considered 
\begin{align*}
 Z=(\hat{\xi}_{0,2},\cdots,\hat{\xi}_{0,N},\hat{\xi}_{1,2},\hat{\eta}_{1,2},\cdots,\hat{\xi}_{N,N},\hat{\eta}_{N,N})\in \mathbf{R}^{(N+1)^2-4}
\end{align*}
as the independent variable of the maximization problem. In this section, we present some of the geometrical properties of the feasible region of \(Z\), i.e., the set of values that \(Z\) can take. 

We assume that  \(I\geq 3N+1, J=I/2, K\geq 3\) as in the main part of the manuscript. We also assume that the vorticity values \(Q_k\) and areas of the patches \(S_k\) satisfy the following:
\begin{itemize}
\item \(S_k >0 \quad(k=1,\cdots, K),\)
\item \(\sum_{k=1}^K Q_kS_k = 0,\)
\item \(Q_k\neq Q_l\) for some \(k,l\).
\end{itemize}
All of these assumptions are physically reasonable.

The feasible region \(P\) is the set of \(Z\in \mathbf{R}^{(N+1)^2-4}\) such that some \(\{r_{ijk}\}_{i,j,k}\) exists and satisfies the following conditions:
\begin{align}
 &\sum_{k=1}^K r_{ijk} - 1=0 \qquad(\forall(i,j)\neq(I,J)),\label{s-d_prob}\\
 &\sum_{i=1}^I\sum_{j=1}^J w_j r_{ijk} -\frac{1}{2\pi} S_k =0\qquad(k=1,\cdots,K),\label{s-d_imcp}\\
 &r_{ijk} \geq 0\qquad (1\leq i\leq I,\,1\leq j\leq J,\,1\leq k\leq K),\label{s-positiveness}\\
 &\overline{q}_{ij} = \sum_{n=1}^N \sum_{m=-n}^n \hat{\zeta}_{m,n}Y_{m,n}(\lambda_i,\mu_j)\qquad(\forall(i,j)\neq (I,J)),\label{s-sg_grad}
\end{align}
where
\begin{align}
 \overline{q}_{ij}= \sum_{k=1}^K Q_k r_{ijk}. \label{s-qij}
\end{align}
To describe the properties of \(P\), we define the following sets:
\begin{align*}
 &R = \left\{(r_{ijk})\in \mathbf{R}^{IJK}\middle|\sum_{i,j}w_jr_{ijk}=\frac{S_k}{2\pi}\,(\forall k),\,\sum_{k=1}^K r_{ijk}=1\,(\forall i,j)\right\},\\
 &R^+ = R\cap \{(r_{ijk})\in\mathbf{R}^{IJK}| r_{ijk}\geq 0\,(\forall i,j,k)\}.
\end{align*}
By calculating the rank of the linear map that defines \(R\), we make
the following claim:
\begin{claim}
 \(R\) is an \((IJ-1)(K-1)\)-dimensional hyperplane in  \(\mathbf{R}^{IJK}\).
\end{claim}

We endow \(R\) with the structure of topological space induced from \(\mathbf{R}^{IJK}\). By noting that if \((r_{ijk})\in R^+\) then \(0\leq r_{ijk} \leq 1\) holds for all \(i,j,k\), the claim below follows:
\begin{claim}
\(R^+\) is a bounded closed convex subset of \(R\). 
\end{claim}
Let \(Q\) be the map that gives the grid values of the macroscopic vorticity \(\overline{q}_{ij}\) from the macroscopic state \(\{r_{ijk}\}\). In other words, we define the map \(Q:R\rightarrow\mathbf{R}^{IJ}\) by
\begin{align*}
 Q : (r_{ijk})_{i,j,k} \mapsto (q_{ij})_{i,j}= \left( \sum_{k=1}^K Q_k r_{ijk}\right)_{i,j}\in \mathbf{R}^{IJ}.
\end{align*}
By the condition that \(Q_k\neq Q_l\) for some \(k,l\quad(k\neq l)\), we can show the following: 
\begin{claim}
The rank of the linear map \(Q:R\rightarrow \mathbf{R}^{IJ}\) is \(IJ-1\). If \((q_{ij})\in Q(R)\), then
\begin{align}
 \sum_{i,j}w_j q_{ij} = 0. \label{s-wjqij}
\end{align}
Therefore, the image of \(R\) under \(Q\) is identical to the hypersurface \(H\subset \mathbf{R}^{IJ}\) defined by the equation \eqref{s-wjqij}.
\end{claim}
Let us define the map \(F_I : \mathbf{R}^{(N+1)^2-4} \rightarrow \mathbf{R}^{IJ}\) that sends the tuple of expansion coefficients \(Z=(\hat{\xi}_{0,2},\cdots,\hat{\xi}_{0,N},\hat{\xi}_{1,2},\hat{\eta}_{1,2},\cdots,\hat{\xi}_{N,N},\hat{\eta}_{N,N})\in \mathbf{R}^{(N+1)^2-4}\) to the tuple of corresponding grid values as
\begin{align*}
 \left(\sum_{n=1}^N\sum_{m=-n}^n \hat{\zeta}_{m,n} Y_{m,n}(\lambda_i,\mu_j)\right)_{i,j}\in \mathbf{R}^{IJ},
\end{align*}
where \(\hat{\zeta}_{m,n} = \hat{\xi}_{m,n} - \sqrt{-1}\hat{\eta}_{m,n},\,\hat{\zeta}_{-m,n}=\hat{\zeta}_{m,n}^*\), and \(\hat{\xi}_{0,1},\hat{\xi}_{1,1},\hat{\eta}_{1,1}\) are the fixed values given by the initial angular momentum. The following claim is verified: 
\begin{claim}
 The map \(F_I\) is injective, and therefore it has rank \((N+1)^2-4\). In addition, the image of \(F_I\) is included in the hyperplane \(H\).
\end{claim}
\begin{proof}
 The proof is based on the following fact: let \(m,n,p,q\) be integers with \(0\leq n\leq N,0\leq q\leq N, |m|\leq n,|p|\leq q\), then
 \begin{align}
    \frac{1}{2} \sum_{i,j}w_j Y_{m,n}(\lambda_i,\mu_j)Y_{p,q}^*(\lambda_i,\mu_j)=\delta_{mp}\delta_{nq}.\label{s-d_ortho}
 \end{align}
 Indeed,
 \begin{align*}
 Y_{m,n}(\lambda,\mu)Y_{p,q}(\lambda,\mu) &= P_{m,n}(\mu)P_{p,q}(\mu) e^{\sqrt{-1}m\lambda}e^{-\sqrt{-1}p\lambda}\\
 &= P_{m,n}(\mu)P_{p,q}(\mu)e^{\sqrt{-1}(m-p)\lambda}
 \end{align*}
 and
 \begin{align*}
   \sum_{i=1}^I e^{\sqrt{-1}(m-p) \lambda_i}= \sum_{i=1}^I e^{\sqrt{-1}\frac{2\pi(m-p)i}{I}}=I\delta_{mp}
 \end{align*}
 which comes from \(|m-p|\leq 2N< I\) show that \eqref{s-d_ortho} holds when \(m\neq p\). When \(m=p\), \(P_{m,n}(\mu)P_{m,q}(\mu)\) is a polynomial in \(\mu\) of degree \(n+q\). Thus, noting that \(n+q\leq 2N < 3N\leq 2J-1\) and applying the Gauss-Legendre quadrature (see the Appendix to the main manuscript), we obtain
 \begin{align*}
    \sum_{j=1}^J \tilde{w}_j P_{m,n}(\mu_j)P_{m,q}(\mu_j) &= \int_{-1}^1 P_{m,n}(\mu)P_{m,q}(\mu) {\rm d}\mu \\
&= 2 \delta_{nq}.
 \end{align*}
 Thus, \eqref{s-d_ortho} holds even when \(m=p\). We now prove the main claim. Let the expansion coefficients \(\hat{\zeta}_{m,n}=\hat{\xi}_{m,n}-\sqrt{-1}\hat{\eta}_{m,n}\) satisfy
\begin{align*}
 \sum_{n=2}^N \sum_{m=-n}^n \hat{\zeta}_{m,n}Y_{m,n}(\lambda_i,\mu_j) = 0\qquad(\forall i,j).
 \end{align*}
 Multiplying \(w_jY_{p,q}^*(\lambda_i,\mu_j)\) in the above equation, summing over \(i,j\), and applying \eqref{s-d_ortho}, we obtain \(\hat{\zeta}_{p,q}=0\). Since \(p,q\) are arbitrary, \(\hat{\xi}_{m,n}=0\) and \(\hat{\eta}_{m,n}=0\) hold for any \((m,n)\,(2\leq n\leq N, 0\leq m\leq n)\). This shows that \(F_I\) is injective. The second part of the claim can be easily verified using formula \eqref{s-d_ortho} for \((p,q)=(0,0)\).
\end{proof}

 Let the hyperplane \(H\) be endowed with the induced topology from \(\mathbf{R}^{IJ}\). Since \((N+1)^2-4\leq IJ-1\), the map \(F_I\) embeds \(\mathbf{R}^{(N+1)^2-4}\) in \(H\). Thus, the feasible region \(P\subset \mathbf{R}^{(N+1)^2-4}\) can be identified with the intersection \(Q(R^+)\cap F_I(\mathbf{R}^{(N+1)^2-4})\).
 Let \(H^+=Q(R^+)\). Here, we make the following assumption:
 \begin{assum}
 The intersection \(H^+\cap F_I(\mathbf{R}^{(N+1)^2-4})\) has at least one interior point of \(H^+\) in \(H\).
 \end{assum}
 This assumption holds for generic values of the \(Q_k\)'s and angular momentum if the condition \(H^+\cap F_I(\mathbf{R}^{(N+1)^2-4})\neq \emptyset\) is fulfilled (this condition should hold for a physically reasonable choice of angular momentum).
 We can identify the following relationship between \(R^+\) and \(H^+\):
 \begin{claim}
  Suppose that \(q=(q_{ij})_{i,j}\in H^+\) belongs to the boundary \(\partial H^+\) of \(H^+\). Any \(r=(r_{ijk})_{i,j,k}\in R^+\) such that \(Q(r)=q\) belongs to the boundary \(\partial R^+\) of \(R^+\).
 \end{claim}
 \begin{proof}
 Assume that there exists a point \(r=(r_{ijk})\in R^+\) such that \(Q(r)=q\) and \(r\) is an interior point of \(R^+\). Since \(q\) belongs to the boundary of \(H^+\), we can take the sequence \(\{q_n\}_
{n=1,2,\cdots}\subset H\) such that each \(q_n\) does not belong to \(H^+\) and \(\{q_n\}_{n=1,2,\cdots}\) converges to \(q\). Since \(Q\) is an affine map of rank \(IJ-1\), there exists a sequence \(\{r^n\}_{n=1,2,\cdots}\subset R\) such that \(Q(r^n)=q_n\,(n=1,2,\cdots)\) and \(\{r^n\}_{n=1,2,\cdots}\) converges to \(r\). For sufficiently large \(n\), \(r^n\) belongs to \(R^+\) since \(r\) is an interior point of \(R^+\). In contrast, \(Q(r^n)=q_n \notin H^+\). This contradicts the definition of \(H^+\).  
 \end{proof}
  \begin{claim}\label{s-claim6}
  If \(q = (q_{ij})_{i,j}\in H^+\) is an interior point of \(H^+\), then there exists an interior point \(r=(r_{ijk})_{i,j,k}\) of \(R^+\) such that \(Q(r)=q\).
 \end{claim}
 To prove this claim, we use the following lemma:
 \begin{lemma}
 Let \(m\geq 1,n\geq 1\). Let \(T: \mathbf{R}^{m+n}\rightarrow \mathbf{R}^n\) be an affine map of rank \(n\) and \(F_i : \mathbf{R}^{m+n}\rightarrow \mathbf{R}\,\,(i=1,\cdots, k)\) be affine maps. Define the closed set \(D\in \mathbf{R}^{m+n}\) by
 \begin{align*}
 D= \{u=(x_1,\cdots,x_m,y_1,\cdots,y_n)\in \mathbf{R}^{m+n}| F_i(u) \geq 0 \,(\forall i=1,\cdots,k)\}.
 \end{align*}
 Suppose that there exists some \(u_*\in D\) such that \(F_i(u_*)>0\) holds for all \(i=1,\cdots,k\).\\
 If \(v\in \mathbf{R}^n\) is an interior point of \(T(D)\), then \(T^{-1}(v)\), i.e., the inverse image of \(v\), includes at least one interior point of \(D\).
 \end{lemma}
 \begin{proof}
 Since \({\rm rank}\, T = n\), by a transformation of coordinates we only need to prove the lemma when \(T\) is given by
 \begin{align*}
 T: (x_1,\cdots, x_m,y_1,\cdots,y_n) \mapsto (y_1,\cdots,y_n).
 \end{align*}
 We assume that there exists some interior point \(v_0=(y_1,\cdots,y_n)\in T(D)\) of \(T(D)\) such that \(T^{-1}(v_0)\) does not include any interior point of \(D\), and we derive a contradiction. Let \(u_0\) be an element of \(T^{-1}(v_0)\cap D\). Since \(u_0\) is not an interior point of \(D\), at least one of \(F_1(u_0),\cdots,F_k(u_0)\) is zero. Changing the index if needed, we may suppose that \(F_1(u_0)=F_2(u_0)=\cdots=F_l(u_0)=0,\,F_j(u_0)>0\,(j=l+1,\cdots,k)\). Let \(\langle \cdot,\cdot\rangle\) be the standard inner product of \(\mathbf{R}^{m+n}\). We take vectors \(f_i\in \mathbf{R}^{m+n}\) and scalars \(c_i\in \mathbf{R}\) such that \(F_i(u)=\langle f_i, u\rangle+c_i\).
 
 For any vector \(\Delta u\in \mathbf{R}^{m+n}\) having the form
 \begin{align*}
 \Delta u=(\Delta u_1,\cdots,\Delta u_m,0,\cdots,0),
 \end{align*}
 there exists some \(i\in \{1,\cdots,l\}\) such that 
 \(\langle f_i,\Delta u\rangle\leq 0\). 
 If not, there exists some \(\Delta u\in \mathbf{R}^{m+n}\) such that 
 \(\langle f_i, \Delta u\rangle>0\)
 holds for all \(i=1,\cdots,l\). Taking sufficiently small \(\varepsilon>0\) and letting \(u=u_0+\varepsilon \Delta u\), we then have \(T(u)=T(u_0)=v_0\) and \(F_i(u)>0\) holds for all \(i=1,\cdots,k\). Hence, \(u\) is an element of \(T^{-1}(v_0)\) and simultaneously an interior point of \(D\), which contradicts the definition of \(v_0\).
 
 For \(t\in \mathbf{R}\), we let \(u_t = u_0 + t(u_0-u_*)\in \mathbf{R}^{m+n}\). If \(t\) is a small positive number, 
 \begin{align*}
  F_i(u_t) = (1+t)F_i(u_0) - tF(u_*)=-tF_i(u_*)<0\qquad(i=1,\cdots,l),\\
  F_i(u_t)>0\qquad(i=l+1\cdots,k).
 \end{align*}
 Let \(v_t = T(u_t)\). Here, we can show that \(T^{-1}(v_t)\cap D =\emptyset\), i.e., \(v_t\notin T(D)\). Indeed, denoting \(u_t\) by \(u_t=(x_1^t,\cdots,x_m^t,y_1^t,\cdots,y_n^t)\) leads to \(v_t=(y_1^t,\cdots,y_n^t)\), and any \(u'\in T^{-1}(v_t)\) can be written in the form \(u'=(x_1,\cdots,x_m,y_1^t,\cdots,y_n^t)\). Letting
 \begin{align*}
 \Delta u=(x_1-x_1^t,\cdots,x_m-x_m^t,0,\cdots,0)\in \mathbf{R}^{m+n},
 \end{align*}
 then \(\langle f_i ,\Delta u\rangle\leq 0\) holds for some \(i=1,\cdots,l\). Therefore, 
 \begin{align*}
 F_i(u')=F_i(u_t) + \langle f_i ,\Delta u\rangle <0,
 \end{align*}
 which yields \(u'\notin D\). Hence, \(T^{-1}(v_t)\cap D =\emptyset\) follows. 
 
 By taking the limit \(t\to 0\), \(v_t\) can be located arbitrarily close to \(v_0\), which contradicts the assumption that \(v_0\) is an interior point of \(T(D)\). We have proven the lemma.
 \end{proof}
 Now, let us prove Claim \ref{s-claim6}. Since the set \(R\) is a hyperplane of dimension \((IJ-1)(K-1)\), by taking an appropriate coordinate \((x_1,\cdots,x_{(IJ-1)(K-1)})\), the set \(R\) can be identified with \(\mathbf{R}^{(IJ-1)(K-1)}\). Then, each \(r_{ijk}\) can be expressed by an affine map of \((x_1,\cdots,x_{(IJ-1)(K-1)})\). 
 Consider the following point that corresponds to a uniform distribution of each vorticity patch on all grids, which is defined as
 \begin{align*}
 r=(r_{ijk})_{i,j,k},\qquad r_{ijk} = \frac{S_k}{S_1+\cdots+S_K}=\frac{S_k}{4\pi}.
 \end{align*}
 This point belongs to \(R^+\), and satisfies the condition \(r_{ijk}>0\) for all \(i,j,k\).
 Thus, we can take \(R^+\) as the set \(D\) in Lemma 1. In addition, the map \(Q:R\rightarrow H\) is an affine map of rank \(IJ-1={\rm dim}\, H\). Therefore, identifying \(H\) with \(\mathbf{R}^{IJ-1}\), we can apply Lemma 1. Since \(H^+=Q(R^+)\), Claim \ref{s-claim6} has been proven. 
 
 Next, we describe the properties of the mixing entropy \(S_{\rm mix}\) as a function of expansion coefficients. We also describe the properties of the function \(V(Z)\) that is used in the numerical method for the microcanonical problem in the main part of the manuscript. 
 \begin{claim}
 \label{s-claim7}
 Suppose that \(Z\in P\) belongs to the interior of \(P\). Then the interior of \(R^+\cap Q^{-1}(F_I(Z))\) includes a point \(r=r(Z)\) at which the function
 \begin{align*}
 S_{\rm mix} (r) = -\frac{1}{2}\sum_{i,j,k} w_j r_{ijk}\log r_{ijk},\qquad r=(r_{ijk})_{i,j,k}
 \end{align*}
 attains the global maximum in \(R\cap Q^{-1}(F_I(Z))\). Such a point is unique. (In the definition of \(S_{\rm mix}(r)\), we consider that \(r_{ijk}\log r_{ijk}=0\) when \(r_{ijk}=0\).)
 \end{claim}
 \begin{proof}
 By Claim \ref{s-claim6}, \(R^+\cap Q^{-1}(F_I(Z))\) has nonempty interior. First, we will show that \(S_{\rm mix}(r)\) is strictly concave in the interior of \(R^+\cap Q^{-1}(F_I(Z))\). Since \({\rm dim}\,R\cap Q^{-1}(F_I(Z))=(IJ-1)(K-2)\), taking  appropriate coordinates \((x_1,\cdots,x_L)\,(L=(IJ-1)(K-2))\), each \(r_{ijk}\) can be expressed by an affine map of \((x_1,\cdots,x_L)\). Thus, by some constants \(\rho_{ijk}^l,c_{ijk}\in \mathbf{R}\), \(r_{ijk}\) can be written in the following form:
 \begin{align*}
 r_{ijk} = \sum_{l=1}^L \rho_{ijk}^l x_l + c_{ijk}.
 \end{align*}
 When \(x=(x_1,\cdots,x_L)\) is an interior point of \(R^+\cap Q^{-1}(F_I(Z))\), each \(r_{ijk}\) is positive, and for \(1\leq l,m\leq L\)
 \begin{align*}
 \frac{\partial^2 S_{\rm mix}}{\partial x_l \partial x_m} = -\frac{1}{2}\sum_{i,j,k} \frac{w_j}{r_{ijk}}\rho_{ijk}^l\rho_{ijk}^m.
 \end{align*}
 Here, the Hessian matrix \({\rm Hes}(S_{\rm mix})\) is negative definite. Indeed, for \(u=(u_1,\cdots,u_L)\in \mathbf{R}^{L}\), 
 \begin{align*}
 \langle {\rm Hes}(S_{\rm mix}) u,u\rangle &= -\frac{1}{2}\sum_{i,j,k}\sum_{l,m=1}^L\frac{w_j}{r_{ijk}}\rho_{ijk}^l\rho_{ijk}^mu_l u_m\\
 &= -\frac{1}{2}\sum_{i,j,k}\frac{w_j}{r_{ijk}}\left(\sum_{l=1}^L \rho_{ijk}^l u_l\right)^2\leq 0.
 \end{align*}
 The equality holds if and only if
 \begin{align*}
 \sum_{l=1}^L \rho_{ijk}^l u_l = 0
 \end{align*}
 holds for all \(i,j,k\). Since the inclusion map \(\iota: R\cap Q^{-1}(F_I(Z))\rightarrow \mathbf{R}^{IJK}\) has rank \((IJ-1)(K-2)\), the condition for equality is equivalent to \(u=0\). Therefore, \(S_{\rm mix}\) is strictly concave in the interior of \(R^+\cap Q^{-1}(F_I(Z))\). 
 
 Second, we will show that the point that gives the maximum of \(S_{\rm mix}\) in \(R^+\cap Q^{-1}(F_I(Z))\) cannot belong to the boundary of \(R^+\cap Q^{-1}(F_I(Z))\). Let \(r^{\rm b} = (r_{ijk}^{\rm b})_{i,j,k}\in R^+ \cap Q^{-1}(F_I(Z))\) be an arbitrary point on the boundary. 
 Then, \(r^{\rm b}_{i_0j_0k_0}=0\) holds for some \(i_0,j_0,k_0\).
 Let \(r^{\rm i} = (r_{ijk}^{\rm i})_{i,j,k}\) be an interior point of \(R^+\cap Q^{-1}(F_I(Z))\). Since \(R^+\cap Q^{-1}(F_I(Z))\) is a convex set, \(r^t=(1-t)r^{\rm i}+t r^{\rm b}\) belongs to \(R^+\cap Q^{-1}(F_I(Z))\) for all \(t\in [0,1]\). In particular, \(r^t\) is an interior point when \(0\leq t <1\). Now, 
 \begin{align*}
  \frac{{\rm d}}{{\rm d} t} S_{\rm mix}(r^t) = -\frac{1}{2}\sum_{i,j,k}w_j (1+\log r_{ijk}^t)(r_{ijk}^{\rm b}-r_{ijk}^{\rm i}).
 \end{align*}
 By noting that \(-w_{j_0}(1+\log r_{i_0j_0k_0}^t)(r_{i_0j_0k_0}^{\rm b}-r_{i_0j_0k_0}^{\rm i})\to -\infty\) as \(t\to 1\), we see that the derivative of \(-S_{\rm mix}(r^t)\) can be arbitrarily large in a neighborhood of \(t=1\). Thus, in any neighborhood of \(r^b\), there exists an interior point \(r\) such that \(S_{\rm mix}(r) > S_{\rm mix}(r^{\rm b})\). This means that \(S_{\rm mix}\) cannot attain the maximum value on the boundary of \(R^+\cap Q^{-1}(F_I(Z))\). Since \(S_{\rm mix}\) is a continuous function on the compact set \(R^+\cap Q^{-1}(F_I(Z))\), it attains a maximum value in \(R^+\cap Q^{-1}(F_I(Z))\). As we have shown, the maximum point can only be in the interior of \(R^+\cap Q^{-1}(F_I(Z))\). Meanwhile, the strictly concave function \(S_{\rm mix}\) has a unique critical point if it exists in the interior of \(R^+\cap Q^{-1}(F_I(Z))\), and \(S_{\rm mix}\) attains a global maximum in \(R^+\cap Q^{-1}(F_I(Z))\) at the point. The claim has been proven.
 \end{proof}
 \begin{claim}
 \label{s-claim8}
 In Claim \ref{s-claim7}, the map that sends \(Z\) to \(r(Z)\) is a smooth map from the interior of \(P\) to \(R^+\).
 \end{claim}
 \begin{proof}
 We only need to prove the claim in a neighborhood of an interior point \(Z_0\) of \(P\). Let \(r^0=r(Z_0)\). The numbers \(c_{ijk}\) that appeared in the proof of Claim \ref{s-claim7} can be considered to be affine functions of \(Z\). Then, \(S_{\rm mix}\) is a function of \(x_1,\cdots,x_L\) and \(Z\). Let \(x_1^0,\cdots,x_L^0\) be the coordinates of the point \(r^0\). As we have shown in the proof of Claim \ref{s-claim7}, the matrix
 \begin{align*}
  \left( \frac{\partial^2 S_{\rm mix}}{\partial x_l \partial x_m}(x_1^0,\cdots,x_L^0,Z_0)\right)_{1\leq l,m\leq L}
 \end{align*}
 is negative definite. In particular, this matrix is regular. For each \(Z\) in the interior of \(P\), the maximum point \(r_{\rm max}\) is the solution to the equation
 \begin{align*}
   \frac{\partial S_{\rm mix}}{\partial x_l}(x_1,\cdots,x_L,Z)=0\qquad(l=1,\cdots,L).
 \end{align*}
 By using the implicit function theorem, we find that such \(x_1,\cdots,x_L\) are smooth functions of \(Z\) locally in a neighborhood of \(Z=Z_0\). Therefore, the claim follows. 
 \end{proof}
 The following claim can be proven in almost the same way as Claims \ref{s-claim7} and \ref{s-claim8}.
 \begin{claim}
 Suppose that \(Z\in P\) belongs to the interior of \(P\). Then the interior of \(R^+\cap Q^{-1}(F_I(Z))\) includes a point \(r=\tilde{r}(Z)\) at which the function
 \begin{align*}
 V(r) = -\frac{1}{2} \sum_{i,j,k} w_j \log r_{ijk}
 \end{align*}
 attains the global minimum in \(R\cap Q^{-1}(F_I(Z))\). Such a point is unique. In addition, the map that sends \(Z\) to \(\tilde{r}(Z)\) is a smooth map from the interior of \(P\) to \(R^+\).
 \end{claim}
 By the above discussion, the functions \(S_{\rm mix}(r(Z))\) and \(V(\tilde{r}(Z))\) are defined for all interior points \(Z\) of \(P\). We can prove the following:
 \begin{claim}
 The function \(S_{\rm mix}(r(Z))\) is strictly concave in the interior of \(P\), and \(V(\tilde{r}(Z))\) is strictly convex in the interior of \(P\).
 \end{claim}
 \begin{proof}
 Let \(Z_1, Z_2\) be interior points of \(P\). For \(0\leq t\leq 1\), \((1-t)r(Z_1)+tr(Z_2)\) is an interior point of \(R^+\), and we have
 \begin{align*}
 (1-t)S_{\rm mix}(r(Z_1))+tS_{\rm mix}(r(Z_2)) \leq S_{\rm mix}((1-t)r(Z_1)+tr(Z_2)),
 \end{align*} 
 since \(S_{\rm mix}(r)\) is a concave function of \(r\). The equality holds if and only if \(t=0\) or \(t=1\). Since \((1-t)r(Z_1)+tr(Z_2)\in R^+ \cap Q^{-1}(F_I((1-t)Z_1+tZ_2))\) and  \(r((1-t)Z_1+tZ_2)\) is defined so as to maximize \(S_{\rm mix}\) for the point \((1-t)Z_1+tZ_2\), we have
 \begin{align*}
 S_{\rm mix}((1-t)r(Z_1)+tr(Z_2)) \leq S_{\rm mix}(r((1-t)Z_1+tZ_2)).
 \end{align*}
 By the above two inequalities, it has been proven that \(S_{\rm mix}(r(Z))\) is a strictly concave function. The strict convexity of \(V(\tilde{r}(Z))\) can be proven in almost the same way.
 \end{proof}
 Note that we have denoted \(S_{\rm mix}(r(Z))\) and \(V(\tilde{r}(Z))\) by simply \(S_{\rm mix}(Z)\) and \(V(Z)\), respectively, in the main part of the manuscript. In the remainder of this supplementary material, we will use the simplified notation unless there is a risk of confusion.
 \begin{claim}
 \label{s-claim11}
 Let \(M_{\rm min}\in \mathbf{R}\) be the minimum value of \(V(Z)\) in the interior of \(P\). If \(M>M_{\rm min}\), then the level hypersurface
 \begin{align*}
 L_M= \{Z\in P| V(Z)=M\}
 \end{align*}
 of \(V(Z)\) is a compact \((N+1)^2-5\)-dimensional differential manifold. Furthermore, for a non-zero vector \(\zeta\in \mathbf{R}^{(N+1)^2-4}\), the function defined on \(L_M\)
 \begin{align*}
 f(Z) = \langle \zeta, Z \rangle \qquad(Z\in L_M)
 \end{align*}
 has a unique maximum point. The hypersurface \(L_M\) is homeomorphic to an \((N+1)^2-5\)-dimensional sphere
 \begin{align*}
 S^{(N+1)^2-5} = \{(x_1,\cdots, x_{(N+1)^2-4})\in \mathbf{R}^{(N+1)^2-4}| x_1^2+\cdots+ x_{(N+1)^2-4}^2=1\}.
 \end{align*}
 \end{claim}
 \begin{proof}
 The set \(L_M\) is a nonempty set since \(V(Z)\) can take an arbitrarily large value in the interior of \(P\). Because the function \(V(Z)\) is strictly convex, it has a unique critical point in the interior of \(P\), at which \(V(Z)\) attains the minimum value \(M_{\rm mix}\). Therefore, \(L_M\) does not include any critical points of \(V(Z)\), which yields that \(L_M\) is an \((N+1)^2-5\)-dimensional differential manifold. Since \(P\subset \mathbf{R}^{(N+1)^2-4}\) is bounded, the set \(L_M\) is compact.
 
 At an arbitrary point \(Z_0\in L_M\), the vector \(\nu_0=\nabla V(Z_0)\) is a normal vector of \(L_M\) in \(\mathbf{R}^{(N+1)^2-4}\). The tangent space \(T_{Z_0}L_M\) of \(L_M\) at \(Z_0\) is an \((N+1)^2-5\)-dimensional hyperplane which includes \(Z_0\) and is orthogonal to \(\nu_0\). By taking an orthogonal coordinate system \((z_1,\cdots,z_{(N+1)^2-5})\) on \(T_{Z_0}L_M\) and a coordinate \(z_{(N+1)^2-4}\) in the direction of \(-\nu_0\), an orthogonal coordinate system \((z_1,\cdots,z_{(N+1)^2-5},z_{(N+1)^2-4})\) of \(\mathbf{R}^{(N+1)^2-4}\) is obtained. Using this coordinate system, the set \(L_M\) in a neighborhood of \(Z_0\) can be realized as the graph of a strictly convex function \(z_{(N+1)^2-4}= g_{Z_0}(z_1,\cdots,z_{(N+1)^2-4})\). We use this fact to prove the second part of the claim.
 
 Since \(f(Z)\) is a continuous function on the compact set \(L_M\), it attains the minimum value in \(L_M\). Assume that there exist two distinct points \(Z_1\) and \(Z_2\) at which \(f(Z)\) attains the minimum value. By taking the above local coordinate system and the strictly convex function for each of \(Z_1, Z_2\), we obtain that both \(\nabla V(Z_0)\) and \(\nabla V(Z_1)\) must be parallel to \(\zeta\) and have the same direction. From the convexity of \(V(Z)\), \(V(Z)\leq V(Z_0)=V(Z_1)\) for \(Z\) on the line segment from \(Z_0\) to \(Z_1\). However, this cannot occur since the functions \(g_{Z_0}\) and \(g_{Z_1}\) that represent \(L_M\) in a neighborhood of the two points are both strictly convex. Thus, the maximum point of \(f(Z)\) is unique. 
 
 We can now construct a homeomorphism from \(L_M\) to \(S^{(N+1)^2-5}\). Let \(\nu : L_M \rightarrow S^{(N+1)^2-5}\) be a map that sends \(Z\in L_M\) to the (outward) unit normal vector \(\nu(Z)\in S^{(N+1)^2-5}\) of \(L_M\) at \(Z\). In addition, let \(\mu : S^{(N+1)^2-5}\rightarrow L_M\) be a map that sends \(\zeta \in S^{(N+1)^2-5}\) to the point \(Z=\mu (\zeta)\in L_M\) which maximizes \(f(Z)=\langle \zeta,Z\rangle \) on \(L_M\). Now both \(\nu \circ \mu : S^{(N+1)^2-5}\rightarrow S^{(N+1)^2-5}\) and \(\mu \circ \nu : L_M\rightarrow L_M\) are confirmed to be identity maps. Thus, \(\mu\) and \(\nu\) are bijections and inverse maps of each other. Therefore, \(\nu\) is a homeomorphism from \(L_M\) to \(S^{(N+1)^2-5}\) since \(\nu\) is a bijection from a compact space \(L_M\) to a Hausdorff space \(S^{(N+1)^2-5}\).
 \end{proof}
 
 The last part of Claim \ref{s-claim11} states that the level hypersurface \(L_M\) can be parameterized by a sphere. This fact guarantees a quasi-global search for the critical points, although it is nearly impossible to find the feasible region \(P\) explicitly.
 
 \section{On the implementation of the gradient method}
 In the gradient method for the canonical problem, the macroscopic state \(\{r_{ijk}\}\) is considered to be a function of the expansion coefficients \(Z\) of the macroscopic vorticity field through the subproblem (SP). In this section, we give details on how to solve the subproblem (SP) numerically and also describe how to compute the gradient of the mixing entropy \(S_{\rm mix}\). 
 
 \subsection{Numerical solution to the subproblem and reduction of computational cost}
 The subproblem is solved by using Newton's method. Let \(A_{ij}, B_k, C_{ij}\) be the functions which define the constraints of (SP):
 \begin{align*}
  &A_{ij}= \sum_{k=1}^K Q_k r_{ijk} - g_{ij}(Z)\qquad&((i,j)\neq (I,J)),\\
  &B_k = \sum_{i,j} w_j r_{ijk} -\frac{1}{2\pi }S_k \qquad&(k=1,\cdots, K),\\
  &C_{ij} = \sum_{k=1}^K r_{ijk} -1\qquad&((i,j)\neq (I,J)),
 \end{align*}
 where 
 \begin{align*}
  g_{ij}(Z) = \sum_{n=1}^N \sum_{m=-n}^n \hat{\zeta}_{m,n} Y_{m,n}(\lambda_i,\mu_j).
 \end{align*}
 Let 
 \begin{align*}
  L= S_{\rm mix} + \frac{1}{2} \sum_{k=1}^K (b_k +1) B_k+\frac{1}{2} \sum_{(i,j)\neq (I,J)} w_j (a_{ij} A_{ij}+c_{ij}C_{ij})
 \end{align*}
 be the Lagrangian. Here, \(a_{ij},b_k,c_{ij}\) are the multipliers that correspond to \(A_{ij},B_k,C_{ij}\), respectively. From \(\partial L /\partial r_{ijk}=0\), we obtain
 \begin{align}
  r_{ijk} = \exp (Q_k a_{ij} + b_k + c_{ij} ). \label{s-SP_rijk}
 \end{align}
 Here, for simplicity, we let \(a_{IJ}=c_{IJ}=0\). By substituting these \(r_{ijk}\)'s into the constraints \(A_{ij}=0, B_{k}=0, C_{ij}=0\), we obtain a system of \(K+2(IJ-1)\) nonlinear equations for the same number of multipliers. We can obtain the solution to (SP) by solving the system of equations using Newton's method. 
 
 In the implementation of the above procedure, we can reduce the computational cost by making use of an algebraic feature of the system of equations, as we describe in the following: Let \(x=(a_{11},\cdots,a_{I,J-1},b_1,\cdots,b_K,c_{11},\cdots,c_{I,J-1})\) be a tuple of multipliers, and let 
 \begin{align*}
  F(x) = (A_{11}(x),\cdots,A_{I,J-1}(x),B_{1}(x),\cdots,B_{K}(x),C_{11}(x),\cdots,C_{I,J-1}(x))
 \end{align*}
 be a vector-valued function which consists of the functions defining the constraints. Let \(JF(x)\) denote the Jacobian matrix of \(F:\mathbf{R}^{K+2(IJ-1)}\rightarrow \mathbf{R}^{K+2(IJ-1)}\) at \(x\in \mathbf{R}^{K+2(IJ-1)}\). The algorithm of Newton's method is summarized as follows: 
 \begin{itemize}
 \item Choose an initial estimate \(x_0\) of the solution.
 \item Let \(x_k\) be the \(k\)-th estimate of the solution. By solving the equation \(JF(x_k) \Delta x = -F(x_k)\), obtain the \((k+1)\)-th estimate \(x_{k+1}\) by \(x_{k+1}= x_k +\Delta x\).
 \item The sequence \(x_0,\cdots, x_k,\cdots\) converges to a solution of \(F(x)=0\).
 \end{itemize}
 Thus, we need to solve linear equations which have \(JF(x_k)\) as the coefficient matrix for solving (SP). Let us introduce the following two vectors:
 \begin{align}
  \Delta x = (\Delta a_{11},\cdots,&\Delta a_{I,J-1},\Delta b_1,\cdots,\Delta b_K,\Delta c_{11},\cdots, \Delta c_{I,J-1}),\nonumber
 \end{align}
 \begin{align}
   \Delta F = (\Delta A_{11}&,\cdots,\Delta A_{I,J-1},\Delta B_1,\cdots,\Delta B_K,\Delta C_{11},\cdots, \Delta C_{I,J-1}).\nonumber
 \end{align}
 Then the linear equation \(JF(x) \Delta x=\Delta F\) can be written as 
 \begin{align}
  s_{ij}\Delta a_{ij}   + \sum_{k=1}^K Q_k r_{ijk} \Delta b_k + t_{ij} \Delta c_{ij}= \Delta A_{ij}\qquad((i,j)\neq (I,J)),\label{s-Del_Aij}\\
 \sum_{(i,j)\neq(I,J)} w_jQ_k r_{ijk} \Delta a_{ij} + \sum_{i,j} w_j r_{ijk} \Delta b_k +\sum_{(i,j)\neq (I,J)} w_k r_{ijk} \Delta c_{ij} \nonumber\\
 \qquad\qquad\qquad\qquad\qquad\qquad\qquad= \Delta B_k\qquad(k=1,\cdots,K),
 \label{s-Del_Bk}\\
  t_{ij} \Delta a_{ij} + \sum_{k=1}^K r_{ijk} \Delta b_{k} + u_{ij} \Delta c_{ij} = \Delta C_{ij}\qquad((i,j)\neq (I,J)),\label{s-Del_Cij}
 \end{align}
 where 
 \begin{align*}
  s_{ij} = \sum_{k=1}^K Q_k^2 r_{ijk},\qquad t_{ij} = \sum_{k=1}^K Q_k r_{ijk},\qquad u_{ij}=\sum_{k=1}^K r_{ijk},
 \end{align*}
 and the \(r_{ijk}\)'s are determined from \(x\) by \eqref{s-SP_rijk}. The number of unknown variables of the equations \eqref{s-Del_Aij}--\eqref{s-Del_Cij} is \(K+2(IJ-1)\), which has order \(O(N^2)\). Thus, the numerical cost of solving the linear equation is \(O(N^6)\). In fact, however, it is not wise to solve the equation directly. By transforming \eqref{s-Del_Aij} and \eqref{s-Del_Cij}, we obtain
 \begin{align}
  \Delta a_{ij} = \frac{1}{D_{ij}}&\Biggl[t_{ij} \left(\Delta C_{ij} -\sum_{k=1}^K r_{ijk} \Delta b_{k}\right) -u_{ij} \left(\Delta A_{ij} -\sum_{k=1}^K Q_k r_{ijk} \Delta b_k\right)\Biggr],\label{s-Delaij}
 \end{align}
 \begin{align}
  \Delta c_{ij} = \frac{1}{D_{ij}}&\Biggl[t_{ij} \left(\Delta A_{ij} -\sum_{k=1}^K Q_k r_{ijk} \Delta b_{k}\right) -s_{ij} \left(\Delta C_{ij} -\sum_{k=1}^K  r_{ijk} \Delta b_k\right)\Biggr],\label{s-Delcij}
 \end{align}
 where \(D_{ij}=t_{ij}^2 - s_{ij}u_{ij}\). Note that \(D_{ij}\leq 0\) follows from the Cauchy-Schwartz inequality and that the equality does not hold under the assumption that \(Q_k\neq Q_l\) for some \(k\neq l\). Substituting \eqref{s-Delaij} and \eqref{s-Delcij} into \eqref{s-Del_Bk}, we obtain a system of linear equations of only the \(\Delta b_k\)'s. Since the number of unknown variables is now \(K\), the cost for solving this system of linear equations is \(O(N^3)\) provided that \(K\) is \(O(N)\). Thus, we can reduce the computational cost of solving (SP) by using the above reduced system of linear equations rather than the original system \(JF(x)\Delta x = \Delta F\).
 
 \subsection{Computing the gradient of the mixing entropy}
 We have shown that the mixing entropy \(S_{\rm mix}\) as a function of the expansion coefficients \(Z\) is differentiable in the interior of the feasible region \(P\). In this subsection, we describe how to numerically calculate the derivatives of \(S_{\rm mix}\).
 
 Let \(x_0\in \mathbf{R}^{K+2(IJ-1)}\) be the tuple of multipliers that corresponds to the solution to (SP) at \(Z=Z_0\). Thus, \(x_0\) satisfies \(F(x_0)=0\). In addition, we assume that the Jacobian matrix \(JF(x_0)\) of \(F\) at \(x=x_0\) is regular. This assumption is not artificial, since, if \(x_0\) can be calculated by Newton's method, then \(JF(x_0)\) should have an inverse matrix. By the inverse function theorem, there exists a function \(G(y)\) defined in the neighborhood of \(y=0\) such that \(F(x)\) and \(G(y)\) are the inverse functions of each other. The Jacobian matrix \(JG(0)\) of \(G\) at \(y=0\) is the inverse matrix of \(JF(x_0)\). The derivative of \(S_{\rm mix}\) about \(y\) at \(y=0\) is given by
 \begin{align*}
  \frac{\partial S_{\rm mix}}{\partial y_l} = -\frac{1}{2}\sum_{i,j,k} w_j (1+\log r_{ijk})  \sum_{p=1}^L\frac{\partial r_{ijk}}{\partial x_p}\frac{\partial x_p}{\partial y_l},
 \end{align*}
 where \(L=K+2(IJ-1)\). The components of \(x,y\) are written as \(x_l,y_l\), respectively, and \(\partial G_p/\partial y_l\) is symbolically denoted by \(\partial x_p/\partial y_l\). Recalling that \(r_{ijk} = \exp(Q_k a_{ij}+b_k + c_{ij})\), we have
 \begin{align*}
 \sum_{p=1}^L\frac{\partial r_{ijk}}{\partial x_p}\frac{\partial x_p}{\partial y_l}=r_{ijk}\frac{\partial}{\partial y_l}(Q_k a_{ij}+b_k +c_{ij}).
 \end{align*}
 If \(Z\) is changed to \(Z_0+\Delta Z\) (\(\Delta Z\): small), the nonlinear equations that produce the solution to (SP) become
 \begin{align*}
 A_{ij}(x)=g_{ij}(\Delta Z),\qquad B_k(x)=0,\qquad C_{ij}(x)=0.
 \end{align*}
 Hence, the derivative of \(S_{\rm mix}\) about the \(m\)-th component \(Z_m\) of \(Z\) is given by
 \begin{align}
 \frac{\partial S_{\rm mix}}{\partial Z_m} &= \sum_{l=1}^L\frac{\partial S_{\rm mix}}{\partial y_l}\frac{\partial y_l}{\partial Z_m}= \sum_{(i,j)\neq(I,J)}\frac{\partial S_{\rm mix}}{\partial A_{ij}}\frac{\partial g_{ij}}{\partial Z_m}\nonumber\\
  &= -\frac{1}{2}\sum_{i,j,k}\left[w_j (1+\log r_{ijk})r_{ijk} \sum_{(i',j')\neq(I,J)}\frac{\partial}{\partial A_{i'j'}}(Q_k a_{ij}+b_k+c_{ij})\frac{\partial g_{i'j'} }{\partial Z_m}\right].
 \end{align}
 Here, 
 \begin{align*}
 \sum_{(i',j')\neq(I,J)}\frac{\partial}{\partial A_{i'j'}}(Q_k a_{ij}+b_k+c_{ij})\frac{\partial g_{i'j'} }{\partial Z_m}
 \end{align*}
 can be obtained from the solution to the linear equation \(JF(x_0) \Delta x = \Delta F\) for the case of 
 \begin{align*}
 \Delta A_{ij} = \frac{\partial g_{ij}}{\partial Z_m},\qquad\Delta B_k=0,\qquad\Delta C_{ij}=0,
 \end{align*}
 because
 \begin{align*}
 \sum_{i',j'} \frac{\partial a_{ij}}{\partial A_{i'j'}}\frac{\partial g_{i'j'}}{\partial Z_m},\qquad\sum_{i',j'} \frac{\partial b_k}{\partial A_{i'j'}}\frac{\partial g_{i'j'}}{\partial Z_m},\qquad\sum_{i',j'} \frac{\partial c_{ij}}{\partial A_{i'j'}}\frac{\partial g_{i'j'}}{\partial Z_m}
 \end{align*}
 are the components of the solution to the linear equation. Computation of the gradient of \(S_{\rm mix}\) can be performed efficiently using the procedure  described above.
 
 \section{Arguments on \(\overline{q}\)-\(\overline{\psi}\) relations}
 For a statistical equilibrium with non-zero inverse temperature \(\beta\), for some constants \(\Omega_1,\Omega_2,\Omega_3\) and some function \(f\), the macroscopic vorticity field \(\overline{q}\) and the macroscopic stream function \(\overline{\psi}\) are linked by the relation
 \begin{align}
    \overline{q}(\lambda,\mu)=f(\overline{\psi}(\lambda,\mu)+\Omega_1\mu+\Omega_{2}\sqrt{1-\mu^2}\cos \lambda + \Omega_3 \sqrt{1-\mu^2}\sin \lambda).\label{s-qpsirel_theory}
 \end{align}
 The proof is nearly parallel to the discussion by \cite{robert1991statistical} or \cite{herbert2012statistical}. The equilibrium macroscopic state \(r_k(\lambda,\mu)\,(k=1,\cdots,K)\) is a critical point of the mixing entropy
 \begin{align*}
 S_{\rm mix} := -\frac{1}{4\pi}\sum_{k=1}^K\int_S r_k(x) \log r_k(x) {\rm d} S
 \end{align*}
 under the constraints 
 \begin{align*}
& \int_S r_k(x) {\rm d} S =S_k \qquad(k=1,\cdots,K),\\
& \sum_{k=1}^K r_k(x) = 1\qquad(\forall x \in S),\\
& M_1:= \frac{1}{4\pi}\int_S \overline{q} \mu {\rm d} S =M_1^{\rm ini},\\
& M_2:=\frac{1}{4\pi}\int_S \overline{q} \sqrt{1-\mu^2} \cos \lambda {\rm d} S =M_2^{\rm ini},\\
& M_3:=\frac{1}{4\pi}\int_S \overline{q} \sqrt{1-\mu^2} \sin \lambda {\rm d} S =M_3^{\rm ini},\\
& E := -\frac{1}{4\pi}\int_S \frac{1}{2}\overline{\psi} \overline{q} {\rm d} S =E_0.
\end{align*}
 Thus, the macroscopic state is a critical point of the Lagrangian
 \begin{align*}
 L&=S_{\rm mix} - \frac{1}{4\pi}\sum_{k=1}^K \alpha_k \left(\int_S r_k {\rm d} S\right) -\beta E \\
 &\qquad - \gamma_1 M_1 -\gamma_2 M_2 -\gamma_3 M_3-\frac{1}{4\pi}\int_S \delta(\lambda,\mu) \sum_{k=1}^K r_k(\lambda,\mu) {\rm d} S
 \end{align*}
 for some multipliers \(\alpha_k\,(k=1,\cdots,K),\,\beta,\,\gamma_i\,(i=1,2,3),\,\delta (\lambda,\mu)\). Therefore, each \(r_k\) can be written as
 \begin{align*}
  r_k(\lambda,\mu) = \frac{1}{Z}\exp \left(-\alpha_k -Q_k(-\beta\overline{\psi} +\gamma_1 \mu +\gamma_2 \sqrt{1-\mu^2} \cos\lambda +\gamma_3 \sqrt{1-\mu^2} \sin \lambda)\right), \\
  Z = \sum_{k=1}^K \exp \left(-\alpha_k -Q_k(-\beta\overline{\psi} +\gamma_1 \mu +\gamma_2 \sqrt{1-\mu^2} \cos\lambda +\gamma_3 \sqrt{1-\mu^2} \sin \lambda)\right).
 \end{align*}
 Hence, when \(\beta\neq 0\), 
 \begin{align*}
 \overline{q}(\lambda,\mu) = \sum_{k=1}^K Q_k r_k(\lambda,\mu)= \frac{1}{\beta Z} \frac{{\rm d} Z}{{\rm d}\overline{\psi}}.
 \end{align*}
 By defining 
 \begin{align*}
 &f(X)= \frac{1}{\beta g(X)}g'(X),\\
 &g(X) = \sum_{k=1}^K \exp \left(-\alpha_k + \beta Q_k X\right),\\
 &\Omega_i = -\frac{\gamma_i}{\beta}\qquad(i=1,2,3),
 \end{align*}
 we have relation \eqref{s-qpsirel_theory}.
 
 Next, we show that \(\Omega_2\) and \(\Omega_3\) must be zero in \eqref{s-qpsirel_theory} when \(M_1^{\rm ini}\neq 0\) and \(M_2^{\rm ini}=M_3^{\rm ini}=0\). If \(\Omega_1=\Omega_2=\Omega_3=0\), we have nothing to prove. Thus, we assume that at least one of \(\Omega_i\,(i=1,2,3)\) is not zero. Assume that the sphere is embedded in \(xyz\)-space, and let the point \((\lambda,\mu)\) on the sphere correspond to \(x=\sqrt{1-\mu^2}\cos \lambda,\,y=\sqrt{1-\mu^2}\sin\lambda,\,z=\mu\). Taking the orthogonal matrix \(P=(p_{ij})\in SO(3)\) such that
 \begin{align}
    \left(
	\begin{array}{c}
	0\\0\\ \Omega' 
	\end{array}\right)
    =
    \left(
    \begin{array}{ccc}
    p_{11} & p_{12} & p_{13}\\
    p_{21} & p_{22} & p_{23}\\
    p_{31} & p_{23} & p_{33}
    \end{array}
    \right)\left(
    \begin{array}{c}
    \Omega_2\\ \Omega_3 \\ \Omega_1
    \end{array}
    \right),\label{s-omega_transform}
 \end{align}
 we define a new orthogonal system \((x',y',z')\) by the rotation of coordinates represented by \(P\), i.e.,
 \begin{align*}
	\left(\begin{array}{c}
	x'\\y'\\ z'      
	\end{array}\right)
    =
    \left(\begin{array}{ccc}
    p_{11} & p_{12} & p_{13}\\
    p_{21} & p_{22} & p_{23}\\
    p_{31} & p_{23} & p_{33}
    \end{array}\right)
    \left(\begin{array}{c}
    x\\ y \\ z
    \end{array}\right).
 \end{align*}
 For the new coordinate system, we define the longitude-latitude coordinate system on the sphere \((\lambda',\mu')\). We can confirm the relation
 \begin{align*}
 \Omega'\mu' = \Omega_1 \mu + \Omega_2 \sqrt{1-\mu^2}\cos\lambda + \Omega_3 \sqrt{1-\mu^2}\sin\lambda.
 \end{align*}
 Defining vectors \(M={}^t(M_2, M_3, M_1)\) and \(M'= {}^t(M'_2,M'_3,M'_1)\), where
 \begin{align*}
 M'_1 = \int_S \overline{q}\mu' {\rm d} S,\\
 M'_2 = \int_S \overline{q} \sqrt{1-\mu'^2} \cos\lambda' {\rm d} S,\\
 M'_3 = \int_S \overline{q}\sqrt{1-\mu'^2} \sin\lambda' {\rm d} S,\\
 \end{align*}
 we have \(M' = PM\) (note the order of the index). Now, we can consider the macroscopic
 vorticity field \(\overline{q}\) to be rotating at an angular velocity \(\Omega'\) about the \(z'\)-axis. Thus, it is a solution of the Euler equation having the form \(\overline{q}(\lambda',\mu',t)= F(\lambda' - \Omega' t ,\mu')\). Hence,
 \begin{align*}
\frac{{\rm d} M'_1}{{\rm d} t} = \int_{-1}^{1}\int_{0}^{2\pi} (-\Omega')\frac{\partial F}{\partial \lambda'}(\lambda'-\Omega' t,\mu')\mu' {\rm d}\lambda' {\rm d}\mu' =0.
\end{align*}
Using integration by parts, we have
 \begin{align*}
 \frac{{\rm d} M'_2}{{\rm d} t} &= \int_{-1}^1 \int_0^{2\pi} (-\Omega')\frac{\partial F}{\partial \lambda'}(\lambda'-\Omega' t,\mu')\sqrt{1-\mu'^2}\cos \lambda' {\rm d}\lambda' {\rm d}\mu' \\
 &= \Omega' \int_{-1}^1 \int_{0}^{2\pi} F(\lambda'-\Omega't,\mu')\frac{\partial}{\partial\lambda'}(\sqrt{1-\mu'^2}\cos \lambda') {\rm d}\lambda' {\rm d}\mu'\\
 &= -\Omega' \int_{-1}^1 \int_{0}^{2\pi} F(\lambda'-\Omega't,\mu')\sqrt{1-\mu'^2}\sin \lambda' {\rm d}\lambda' {\rm d}\mu' = -\Omega' M'_3.
 \end{align*}
 Similarly,
 \begin{align*}
 \frac{dM'_3}{dt} = \Omega' M'_2.
 \end{align*}
 From the conservation law of angular momentum, \(M'_2=M'_3=0\) is necessary. \(M={}^t PM'\) yields the relations \(M_2 = p_{31}M'_1\) and \(M_3=p_{32}M'_1\). \(M'_1\neq 0\) (since \(M\neq 0\)) and \(M_2=M_3=0\) yield \(p_{31}=p_{32}=0\). Since \(P\) is an orthogonal matrix, \(p_{33}=1,\,p_{13}=p_{23}=0\). We have \(p_{11}\Omega_2 + p_{12}\Omega_3=0,\,p_{21}\Omega_2+p_{22}\Omega_3 =0\) from \eqref{s-omega_transform}. Meanwhile, the matrix 
 \begin{align*}
 \left(\begin{array}{cc}
 p_{11} & p_{12}\\
 p_{21} & p_{22}
 \end{array}\right)
 \end{align*}
 is regular since \(\det P=1\neq 0\). Therefore, \(\Omega_2=\Omega_3=0\) and \(\Omega'=\Omega_1\) follow.

%% file: reflist_abb.bib
@article{bouchet2008simpler,
  title={Simpler variational problems for statistical equilibria of the 2D Euler equation and other systems with long range interactions},
  author={Bouchet, Freddy},
  journal={Physica D},
  volume={237},
  number={14-17},
  pages={1976--1981},
  year={2008},
  publisher={Elsevier}
}

@article{bouchet2002emergence,
  title={Emergence of intense jets and Jupiter's Great Red Spot as maximum-entropy structures},
  author={Bouchet, Freddy and Sommeria, Jo{\"e}l},
  journal={J. Fluid Mech.},
  volume={464},
  pages={165--207},
  year={2002},
  publisher={Cambridge University Press}
}

@article{bouchet2005classification,
  title={Classification of phase transitions and ensemble inequivalence, in systems with long range interactions},
  author={Bouchet, Freddy and Barre, Julien},
  journal={J. Stat. Phys.},
  volume={118},
  number={5},
  pages={1073--1105},
  year={2005},
  publisher={Springer}
}

@article{bouchet2012statistical,
  title={Statistical mechanics of two-dimensional and geophysical flows},
  author={Bouchet, Freddy and Venaille, Antoine},
  journal={Phys. Rep.},
  volume={515},
  number={5},
  pages={227--295},
  year={2012},
  publisher={Elsevier}
}

@article{chavanis1996classification,
  title={Classification of self-organized vortices in two-dimensional turbulence: the case of a bounded domain},
  author={Chavanis, Pierre-Henri and Sommeria, Joel},
  journal={J. Fluid Mech.},
  volume={314},
  pages={267--297},
  year={1996},
  publisher={Cambridge University Press}
}

@book{durran2010numerical,
  title={Numerical methods for fluid dynamics: With applications to geophysics},
  author={Durran, Dale R},
  volume={32},
  year={2010},
  publisher={Springer Science \& Business Media}
}

@article{herbert2012statistical,
  title={Statistical mechanics of quasi-geostrophic flows on a rotating sphere},
  author={Herbert, Corentin and Dubrulle, B{\'e}rengere and Chavanis, Pierre-Henri and Paillard, Didier},
  journal={J. Stat. Mech-Theory E.},
  volume={2012},
  number={05},
  pages={P05023},
  year={2012},
  publisher={IOP Publishing}
}

@article{herbert2013additional,
  title={Additional invariants and statistical equilibria for the 2D Euler equations on a spherical domain},
  author={Herbert, Corentin},
  journal={J. Stat. Phys.},
  volume={152},
  number={6},
  pages={1084--1114},
  year={2013},
  publisher={Springer}
}

@article{ishioka1994non,
  title={Non-linear evolution of a barotropically unstable circumpolar vortex},
  author={Ishioka, Keiichi and Yoden, Shigeo},
  journal={J. Meteorol. Soc. Jpn. Ser. II},
  volume={72},
  number={1},
  pages={63--80},
  year={1994},
  publisher={Meteorological Society of Japan}
}

@article{ishioka1996numerical,
  title={Numerical methods of estimating bounds on the non-linear saturation of barotropic instability},
  author={Ishioka, Keiichi and Yoden, Shigeo},
  journal={J. Meteorol. Soc. Jpn. Ser. II},
  volume={74},
  number={2},
  pages={167--174},
  year={1996},
  publisher={Meteorological Society of Japan}
}

@article{ishioka1990instability,
  title={The stability of zonal flows on a sphere},
  author={Ishioka, Keiichi and Yoden, Shigeo},
  journal={Tsukumo Earth Science},
  number={24},
  pages={p46--61},
  year={1990},
  publisher={}
}

@article{ishioka1998,
  title={},
  author={Ishioka, Keiichi},
  journal={RIMS Kokyuroku},
  volume={1029},
  pages={164--175},
  year={1998},
  publisher={Research Institute for Mathematical Sciences (RIMS)}
}

@article{ishioka2013proof,
  title={A proof for the equivalence of two upper bounds for the growth of disturbances from barotropic instability},
  author={Ishioka, Keiichi},
  journal={J. Meteorol. Soc. Jpn. Ser. II},
  volume={91},
  number={6},
  pages={843--850},
  year={2013},
  publisher={Meteorological Society of Japan}
}

@article{joyce1973negative,
  title={Negative temperature states for the two-dimensional guiding-centre plasma},
  author={Joyce, Glenn and Montgomery, David},
  journal={J. Plasma Phys.},
  volume={10},
  number={1},
  pages={107--121},
  year={1973},
  publisher={Cambridge University Press}
}

@article{montgomery1974statistical,
  title={Statistical mechanics of “negative temperature” states},
  author={Montgomery, David and Joyce, Glenn},
  journal={Phys. Fluids},
  volume={17},
  number={6},
  pages={1139--1145},
  year={1974},
  publisher={American Institute of Physics}
}

@article{kraichnan1975statistical,
  title={Statistical dynamics of two-dimensional flow},
  author={Kraichnan, Robert H},
  journal={J. Fluid Mech.},
  volume={67},
  number={1},
  pages={155--175},
  year={1975},
  publisher={Cambridge University Press}
}

@article{kraichnan1980two,
  title={Two-dimensional turbulence},
  author={Kraichnan, Robert H and Montgomery, David},
  journal={Rep. Prog. Phys.},
  volume={43},
  number={5},
  pages={547},
  year={1980},
  publisher={IOP Publishing}
}

@article{mcwilliams1984emergence,
  title={The emergence of isolated coherent vortices in turbulent flow},
  author={McWilliams, James C},
  journal={J. Fluid Mech.},
  volume={146},
  pages={21--43},
  year={1984},
  publisher={Cambridge University Press}
}

@article{miller1990statistical,
  title={Statistical mechanics of Euler equations in two dimensions},
  author={Miller, Jonathan},
  journal={Phys. Rev. Lett.},
  volume={65},
  number={17},
  pages={2137},
  year={1990},
  publisher={APS}
}

@article{miller1992statistical,
  title={Statistical mechanics, Euler’s equation, and Jupiter’s Red Spot},
  author={Miller, Jonathan and Weichman, Peter B and Cross, MC},
  journal={Phys. Rev. A},
  volume={45},
  number={4},
  pages={2328},
  year={1992},
  publisher={APS}
}

@article{onsager1949statistical,
  title={Statistical hydrodynamics},
  author={Onsager, Lars},
  journal={Il Nuovo Cimento (1943-1954)},
  volume={6},
  number={2},
  pages={279--287},
  year={1949},
  publisher={Springer}
}

@article{prieto2001analytical,
  title={Analytical predictions for zonally symmetric equilibrium states of the stratospheric polar vortex},
  author={Prieto, Ricardo and Schubert, Wayne H},
  journal={J. Atmos. Sci.},
  volume={58},
  number={18},
  pages={2709--2728},
  year={2001},
  publisher={American Meteorological Society}
}

@article{robert1991statistical,
  title={Statistical equilibrium states for two-dimensional flows},
  author={Robert, Raoul and Sommeria, Joel},
  journal={J. Fluid Mech.},
  volume={229},
  pages={291--310},
  year={1991},
  publisher={Cambridge University Press}
}

@article{robert1992relaxation,
  title={Relaxation towards a statistical equilibrium state in two-dimensional perfect fluid dynamics},
  author={Robert, Raoul and Sommeria, Jo{\"e}l},
  journal={Phys. Rev. Lett.},
  volume={69},
  number={19},
  pages={2776},
  year={1992},
  publisher={APS}
}

@article{robert1991maximum,
  title={A maximum-entropy principle for two-dimensional perfect fluid dynamics},
  author={Robert, Raoul},
  journal={J. Stat. Phys.},
  volume={65},
  number={3},
  pages={531--553},
  year={1991},
  publisher={Springer}
}

@article{schubert1999polygonal,
  title={Polygonal eyewalls, asymmetric eye contraction, and potential vorticity mixing in hurricanes},
  author={Schubert, Wayne H and Montgomery, Michael T and Taft, Richard K and Guinn, Thomas A and Fulton, Scott R and Kossin, James P and Edwards, James P},
  journal={J. Atmos. Sci.},
  volume={56},
  number={9},
  pages={1197--1223},
  year={1999}
}

@article{shepherd1988rigorous,
  title={Rigorous bounds on the nonlinear saturation of instabilities to parallel shear flows},
  author={Shepherd, Theodore G},
  journal={J. Fluid Mech.},
  volume={196},
  pages={291--322},
  year={1988},
  publisher={Cambridge University Press}
}

@book{stein2016introduction,
  title={Introduction to Fourier Analysis on Euclidean Spaces (PMS-32), Volume 32},
  author={Stein, Elias M and Weiss, Guido},
  year={2016},
  publisher={Princeton university press}
}

@article{sommeria1991final,
  title={Final equilibrium state of a two-dimensional shear layer},
  author={Sommeria, J and Staquet, C and Robert, R},
  journal={J. Fluid Mech.},
  volume={233},
  pages={661--689},
  year={1991},
  publisher={Cambridge University Press}
}

@article{tanabe1974algorithm,
  title={An algorithm for constrained maximization in nonlinear programming},
  author={Tanabe, KUNIO},
  journal={J. Oper. Res. Soc. Jpn.},
  volume={17},
  number={4},
  pages={184--201},
  year={1974}
}

@article{thess1994inertial,
  title={Inertial organization of a two-dimensional turbulent vortex street},
  author={Thess, A and Sommeria, J and J{\"u}ttner, B},
  journal={Phys. Fluids},
  volume={6},
  number={7},
  pages={2417--2429},
  year={1994},
  publisher={American Institute of Physics}
}

@article{turkington1996statistical,
  title={Statistical equilibrium computations of coherent structures in turbulent shear layers},
  author={Turkington, Bruce and Whitaker, Nathaniel},
  journal={SIAM J. Sci. Comput.},
  volume={17},
  number={6},
  pages={1414--1433},
  year={1996},
  publisher={SIAM}
}

@article{yasuda2017new,
  title={A new interpretation of vortex-split sudden stratospheric warmings in terms of equilibrium statistical mechanics},
  author={Yasuda, Yuki and Bouchet, Freddy and Venaille, Antoine},
  journal={J. Atmos. Sci.},
  volume={74},
  number={12},
  pages={3915--3936},
  year={2017},
  publisher={American Meteorological Society}
}

@article{yoden1993numerical,
  title={A numerical experiment on two-dimensional decaying turbulence on a rotating sphere},
  author={Yoden, Shigeo and Yamada, Michio},
  journal={J. Atmos. Sci.},
  volume={50},
  number={4},
  pages={631--644},
  year={1993}
}

@article{venaille2011oceanic,
  title={Oceanic rings and jets as statistical equilibrium states},
  author={Venaille, Antoine and Bouchet, Freddy},
  journal={J. Phys. Oceanogr.},
  volume={41},
  number={10},
  pages={1860--1873},
  year={2011}
}

@article{turkington2001statistical,
  title={Statistical equilibrium predictions of jets and spots on Jupiter},
  author={Turkington, Bruce and Majda, Andrew and Haven, Kyle and DiBattista, Mark},
  journal={P. NATL. ACAD. Sci. USA.},
  volume={98},
  number={22},
  pages={12346--12350},
  year={2001},
  publisher={National Acad Sciences}
}

@article{michel1994statistical,
  title={Statistical mechanical theory of the great red spot of Jupiter},
  author={Michel, Julien and Robert, Raoul},
  journal={J. Stat. Phys.},
  volume={77},
  number={3},
  pages={645--666},
  year={1994},
  publisher={Springer}
}

@article{chavanis2005statistical,
  title={Statistical mechanics of geophysical turbulence: application to jovian flows and Jupiter’s great red spot},
  author={Chavanis, Pierre-Henri},
  journal={Physica D},
  volume={200},
  number={3-4},
  pages={257--272},
  year={2005},
  publisher={Elsevier}
}

@article{ishioka2018,
  title={A new recurrence formula for efficient computation of spherical harmonic transform},
  author={Ishioka, Keiichi},
  journal={J. Meteor. Soc. Japan},
  volume={96},
  number={},
  pages={241--249},
  year={2018},
  publisher={J-STAGE}
}
